\documentclass[11pt,a4paper]{article}

\usepackage[utf8]{inputenc}
\usepackage{authblk}
\usepackage{amsfonts}
\usepackage{amsmath}
\usepackage{amsthm}
\usepackage{booktabs}
\usepackage{cite}
\usepackage{color}
\usepackage{framed}
\usepackage{fullpage}
\usepackage{graphicx}
\usepackage[hidelinks]{hyperref}
\hypersetup{pdfstartview=XYZ}
\usepackage{times}
\usepackage{xspace}

\usepackage{subcaption}
\usepackage{rotating}

\newcommand{\uu}{u}
\newcommand{\poly}{\mathrm{poly}\xspace}
\newcommand{\eps}{\varepsilon}
\newcommand{\act}{\mathbf{\colon}\quad}
\newcommand{\E}{\mathbb E}
\newcommand{\Var}{\mathrm {Var}}
\newcommand{\R}{\mathbb R}
\newcommand{\Z}{\mathbb Z}
\newcommand{\many}{\text{$^\#$}\!\!}

\newcommand{\modd}{\,\mathrm{mod}\,}
\newcommand{\const}{\mathrm {const}}

\newcommand{\todo}[1]{}

\newcommand{\z}{\bar{z}}
\newcommand{\lno}{\ln^\circ}

\newcommand{\on}[1]{\|#1\|_1}

\newtheorem{theorem}{Theorem}
\newtheorem{lemma}{Lemma}

\newtheorem{proposition}{Proposition}
\newtheorem{observation}{Observation}

\begin{document}

\title{\bf Universal Protocols for Information Dissemination\\ Using Emergent Signals%
}

\author[1]{Bart\l{}omiej Dudek}
\author[2]{Adrian Kosowski\footnote{Corresponding Author. Email: adrian.kosowski@inria.fr}}

\affil[1]{University of Wroc\l{}aw, Poland}
\affil[2]{Inria Paris, France}

\date{}

\maketitle

\begin{abstract}
We consider a population of $n$ agents which communicate with each other in a decentralized manner, through random pairwise interactions. One or more agents in the population may act as authoritative sources of information, and the objective of the remaining agents is to obtain information from or about these source agents. We study two basic tasks: \emph{broadcasting}, in which the agents are to learn the bit-state of an authoritative source which is present in the population, and \emph{source detection}, in which the agents are required to decide if at least one source agent is present in the population or not.

We focus on designing protocols which meet two natural conditions: (1) universality, i.e., independence of population size, and (2) rapid convergence to a correct global state after a reconfiguration, such as a change in the state of a source agent. Our main positive result is to show that both of these constraints can be met. For both the broadcasting problem and the source detection problem, we obtain solutions with a convergence time of $O(\log^2 n)$ rounds, w.h.p., from any starting configuration. The solution to broadcasting is exact, which means that all agents reach the state broadcast by the source, while the solution to source detection admits one-sided error on a $\eps$-fraction of the population (which is unavoidable for this problem). Both protocols are easy to implement in practice and have a compact formulation.

Our protocols exploit the properties of self-organizing oscillatory dynamics. On the hardness side, our main structural insight is to prove that \emph{any} protocol which meets the constraints of universality and of rapid convergence after reconfiguration must display a form of non-stationary behavior (of which oscillatory dynamics are an example). We also observe that the periodicity of the oscillatory behavior of the protocol, when present, must necessarily depend on the number $\many X$ of source agents present in the population. For instance, our protocols inherently rely on the emergence of a signal passing through the population, whose period is $\Theta(\log \frac{n}{\many X})$ rounds for most starting configurations. The design of clocks with tunable frequency may be of independent interest, notably in modeling biological networks.
\end{abstract}

\thispagestyle{empty}
 \vspace{1.5cm}
 Population protocols, Broadcasting, Distributed clock synchronization.
 \vspace{1cm}


\section{Introduction}

Information-spreading protocols, and more broadly epidemic processes, appear in nature, social interactions between humans, as well as in man-made technology, such as computer networks. For some protocols we have a reasonable understanding of the extent to which the information has already spread, i.e., we can identify where the information is located at a given step of the process: we can intuitively say which nodes (or agents) are ``informed'' and which nodes are ``uninformed''. This is the case for usual protocols in which uninformed agents become informed upon meeting a previously informed agent, cf. e.g.\ mechanisms of rumor spreading and opinion spreading models studied in the theory community~\cite{rspushpull,gnw}). Arguably, most man-made networking protocols for information dissemination also belong to this category.

By contrast, there exists a broad category of complex systems for which it is impossible to locate which agents have acquired some knowledge, and which are as yet devoid of it. In fact, the question of ``where the information learned by the system is located'' becomes somewhat fuzzy, as in the case of both biological and synthetic neural networks. In such a perspective, information (or knowledge) becomes a global property of the entire system, whereas the state of an individual agent represents in principle its \emph{activation}, rather than whether it is informed or not. As such, knowledge has to be treated as an emergent property of the system, i.e., a global property not resulting directly from the local states of its agents. The convergence from an uninformed population to an informed population over time is far from monotonous. Even so, once some form of ``signal'' representing global knowledge has emerged, agents may try to read and copy this signal into their local state, thus each of them eventually also becomes informed. At a very informal conceptual level, we refer to this category of information-dissemination protocols as protocols with \emph{emergent behavior}. At a more technical level, emergent protocols essentially need to rely on non-linear dynamical effects, which typically include oscillatory behavior, chaotic effects, or a combination of both. (This can be contrasted with simple epidemic protocols for information-spreading, in which nodes do not become deactivated.)

This work exhibits a simple yet fundamental information-spreading scenario which can only be addressed efficiently using emergent protocols. Both the efficient operation of the designed protocols, and the need for non-stationary dynamical effects in any efficient protocol for the considered problems, can be formalized through rigorous theoretical analysis. Our goal in doing this is twofold: to better understand the need for emergent behavior in real-world information spreading, and to display the applicability of such protocols in man-made information spreading designs. For the latter, we describe an interpretation of information as a (quasi-)\emph{periodic signal}, which can be both decoded from states of individual nodes, and encoded into them.

\subsection{Problems and Model}

We consider a population of $n$ identical agents, each of which may be in a \emph{constant} number of possible states. Interactions between agents are pairwise and random. A fair scheduler picks a pair of interacting agents independently and uniformly at random in each step. The protocol definition is provided through a finite sequence of state transition rules, following the precise conventions of the randomized Population Protocol model~\cite{AADFP06,AAER07} or (equivalently) of Chemical Reaction Networks~\cite{pplb1}.\footnote{The activation model is thus asynchronous. The same protocols may be deployed in a synchronous setting, with scheduler activations following, e.g., the independent random matching model (with only minor changes to the analysis) or the \cal{PULL} model~\cite{rspushpull} (at the cost of significantly complicating details of the protocol formulation).}

The \emph{input} to the problem is given by fixing the state of some subset of agents, to some state of the protocol, which is not available to any of the other agents. Intuitively, the agents whose state has been fixed are to be interpreted as authoritative sources of information, which is to be detected and disseminated through the network (i.e., as the rumor source node, broadcasting station, etc.). For example, the problem of spreading a bit of information through the system is formally defined below.
\begin{description}
\itemsep0em
\item[Problem \textsc{BitBroadcast}]
\item[Input States:] $X_1, X_2$.
\item[Promise:] The population contains a non-zero number of agents in exactly one of the two input states $\{X_1, X_2\}$.
\item[Question:] Decide if the input state present in the population is $X_1$ or $X_2$.
\end{description}
We can, e.g., consider that the transmitting station (or stations) choose whether to be in state $X_1$ or $X_2$ in a way external to the protocol, and thus transmit the ``bit'' value $1$ or $2$, respectively, through the network. Broadcasting a bit is one of the most fundamental networking primitives.

The definition of the population protocol includes a partition of the set of states of the protocol into those corresponding to the possible answers to the problem. When the protocol is executed on the population, the \emph{output} of each agent may be read at every step by checking, for each agent, whether its state belongs to the subset of an output state with a given answer (in this case, the answer of the agent will be the ``bit'' it has learned, i.e., 1 or 2). We will call a protocol \emph{exact} if it eventually converges to a configuration, such that starting from this configuration all agents always provide the correct answer. We will say it operates \emph{with $\eps$-error}, for a given constant $\eps>0$, if starting from some step, at any given step of the protocol, at most an $\eps$-fraction of the population holds the incorrect answer, with probability $1 - O(1/n)$.

Time is measured in \emph{steps} of the scheduler, with $n$ time steps called a \emph{round}, with the expected number of activations of each agent per round being a constant. Our objective is to design protocols which converge to the desired outcome rapidly. Specifically, a protocol is expected to converge in $O(\poly \log n)$ rounds (i.e., in $O(n \poly \log n)$ steps), with probability $1 - O(1/n)$, starting from \emph{any possible starting configuration} of states in the population, conforming to the promise of the problem.\footnote{We adhere to this strong requirement for self-stabilizing (or self-organizing) behavior from \emph{any} initial configuration in the design of our protocols. The presented impossibility results still hold under significantly weaker assumptions.}

Motivated by both applications and also a need for a better understanding of the broadcasting problem, we also consider a variant of the broadcasting problem in which no promise on the presence of the source is given. This problem, called \textsc{Detection}, is formally defined below.

\begin{description}
\itemsep0em
\item[Problem \textsc{Detection}]
\item[Input State:] $X$.
\item[Question:] Decide if at least one agent in state $X$ is present in the population.
\end{description}

Detection of the presence of a source is a task which is not easier than broadcasting a bit. Indeed, any detection protocol is readily converted into a broadcasting protocol for states $\{X_1, X_2\}$ by identifying $X = X_1$ and treating $X_2$ as a dummy state which does not enter into any interactions (i.e., is effectively not visible in the network). Intuitively, the detection task in the considered setting is much harder: a source $X$ may disappear from the network at any time, forcing other agents to spontaneously ``unlearn'' the outdated information about the presence of the source. This property is inherently linked to the application of the \textsc{Detection} problem in suppressing false rumors or outdated information in social interactions. Specifically, it may happen that a certain part of a population find themselves in an informed state before the original rumor source is identified as a source of false information, a false rumor may be propagated accidentally because of an agent which previously changed state from ``uninformed'' to ``informed'' due to a fault or miscommunication, or the rumor may contain information which is no longer true. Similar challenges with outdated information and/or false-positive activations are faced in Chemical Reaction Networks, e.g. in DNA strand displacement models~\cite{dna}. In that context, the detection problem has the intuitive interpretation of detecting if a given type of chemical or biological agent (e.g., a contaminant, cancer cell, or hormonal signal) is present in the population, and spreading this information among all agents.

\subsection{Our Results}

In Section~\ref{sec:protocols}, we show that both the~\textsc{BitBroadcast}, and the~\textsc{Detection} problem can be solved with protocols which converge in $O(\log^2 n)$ rounds to an outcome, with probability $1-O(1/n)$, starting from any configuration of the system. The solution to \textsc{BitBroadcast} guarantees a correct output. The solution to \textsc{Detection} admits one-sided $\eps$-error: in the absence of a source, all agents correctly identify it as absent, whereas when the source is present, at any moment of time after convergence the probability that at least $(1-\eps)n$ agents correctly identify the source as present is at least $1 - O(1/n)$.\footnote{The existence of one-sided error is inherent to the \textsc{Detection} problem in the asynchronous setting: indeed, if no agent of the population has not made any communication with the source over an extended period of time, it is impossible to tell for sure if the source has completely disappeared from the network, or if it is not being selected for interaction by the random scheduler.} Here, $\eps>0$ is a constant influencing the protocol design, which can be made arbitrarily small.

The designed protocols rely on the same basic building block, namely, a protocol realizing oscillatory dynamics at a rate controlled by the number of present source states in the population. Thus, these protocols display non-stationary behavior. In Section~\ref{sec:impossibility}, we show that such behavior is a necessary property in the following sense. We prove that in any protocol which solves \textsc{Detection} in sub-polynomial time in $n$ and which uses a constant number of states, the number of agents occupying some state has to undergo large changes: by a polynomially large \emph{factor} in $n$ during a time window of length proportional to the convergence time of the protocol. For the \textsc{BitBroadcast} problem, we show that similar volatile behavior must appear in a synthetic setting in which a unique source is transmitting its bit as random noise (i.e., selecting its input state $\{X_1,X_2\}$ uniformly at random in subsequent activations).

We note that, informally speaking, our protocols rely on the emergence of a ``signal'' passing through the population, whose period is $\Theta(\log \frac{n}{\many X})$ rounds when the number of source agents in state $X$ is $\many X$. In Section~\ref{sec:signal}, we then discuss how the behavior of \emph{any} oscillatory-type protocol controlled by the existence of $\many X$ has to depend on both $n$ and $\many X$. We prove that for any such protocol with rapid convergence, the cases of subpolynomial $\many X$ and $\many X = \Theta(n)$ can be separated by looking at the portion of the configuration space regularly visited by the protocol. This, in particular, suggests the nature of the dependence of the oscillation period on the precise value of $\many X$, and that the protocols we design with period $\Theta(\log \frac n {\many X})$ are among the most natural solutions to the considered problems.

The proofs of all theorems are deferred to the closing sections of the paper.

\begin{table}
\small
\begin{tabular}{p{5.5cm}p{4.5cm}p{4.5cm}}
\toprule
\textsc{Problem:} & \textsc{BitBroadcast} & \textsc{Detection}\\
\midrule
\textbf{Non-stationarity property:}\newline (Applies to all fast $O(1)$-state protocols) & no fixed points while source transmits random bits & no fixed points while source is present\\

\midrule
\textbf{New protocols with emergent signal:} & universal, 74 states & universal, 55 states\\
Convergence time: \\
-- No error (exact output) & $O(\log^2 n)$ & impossible\\
-- One-sided $\eps$-error &&  $O(\log^2 n)$ \\

\midrule
\textbf{Other protocols with $\omega(1)$-states:} & Clock-Sync (in synchronized round model)~\cite{bocz} & Time-to-Live~\cite{dna}\\
\bottomrule
\end{tabular}
\end{table}

\subsection{Comparison to the State-of-the-Art}

Our work fits into lines of research on rumor spreading, opinion spreading, population protocols and other interaction models, and emergent systems. We provide a more comprehensive literature overview of some of these topics in Subsection~\ref{related}.

\paragraph{Other work on the problems.} The \textsc{BitBroadcast} problem has been previously considered by Boczkowski, Korman, and Natale in~\cite{bocz}, in a self-stabilizing but synchronous (round-based) setting. A protocol solving their problem was presented, giving a stabilized solution in $\tilde O (\log n)$ time, using a number of states of agents which \emph{depends on population size}, but exchanging messages of bit size $O(1)$ (this assumption can be modeled in the population protocol framework as a restriction on the permitted rule set). In this sense, our result can be seen as providing improved results with respect to their approach, since it is applicable in an asynchronous setting and reduces the number of states to constant (the latter question was open~\cite{bocz}). We remark that their protocol has a more general application to the problem of deciding which of the sources $X_1$ or $X_2$ is represented by a larger number of agents, provided these two numbers are separated by a multiplicative constants. (Our approach could also be used in such a setting, but the required separation of agent numbers to ensure a correct output would have to be much larger: we can compare values if their \emph{logarithms} are separated by a multiplicative constant.) The protocol of~\cite{bocz} involves a routine which allows the population to create a synchronized modulo clock, working in a synchronous setting. The period of this clock is independent of the input states of the protocol, which should be contrasted with the oscillators we work on in this paper, which encode the input into a signal (with a period depending on the number of agents in a given input state).

The \textsc{Detect} problem was introduced in a work complementary to this paper~\cite{dna}. Therein we look at applications of the confirmed rumor spreading problem in DNA computing, focusing on performance on protocols based on a time-to-live principle and on issues of fault tolerance in a real-world model with leaks. The protocols designed there require $O(\log n)$ states, and while self-stabilizing, do not display emergent behavior (in particular, agents can be categorized as ``informed'' and ``uninformed'', the number of correctly informed agents tends to increase over time, while the corresponding continuous dynamical system stabilizes to a fixed point attractor).

\paragraph{Originality of methods.} The oscillatory dynamics we apply rely on an \emph{input-parameter-controlled} oscillator. The \emph{uncontrolled} version of the oscillator which we consider here is the length-$3$ cyclic oscillator of the cyclic type, known in population dynamics under the name of rock-paper-scissors (or RPS). This has been studied intensively in the physics and evolutionary dynamics literature (cf.e.g.~\cite{SMJSRP14} for a survey), while algorithmic studies are relatively scarce~\cite{ICALP15}. We remark that the uncontrolled cyclic oscillator with a longer (but $O(1)$ length cycle) has been applied for clock/phase synchronization in self-stabilizing settings and very recently in the population protocol setting when resolving the leader election problem~\cite{grzes}. (The connection to oscillatory dynamics is not made explicit, and the longer cycle provides for a neater analysis, although it does not seem to be applicable to our parameter-controllable setting.)
Whereas we are not aware of any studies of parameter-controlled oscillators in a protocol design setting (nor for that matter, of rigorous studies in other fields), we should note that such oscillators have frequently appeared in models of biological systems, most notably in biological networks and neuroendocrinology (\cite{bion} for a survey). Indeed, some hormone release and control mechanisms (e.g., for controlling GnRH surges in vertebrates) appear to be following a similar pattern. To the best of our knowledge, no \emph{computational} (i.e., interaction-protocol-based) explanation for these mechanisms has yet been proposed, and we hope that our work may provide, specifically on the \textsc{Detect} problem, may provide some insights in this direction.

In terms of lower bounds, we rely on rather tedious coupling techniques for protocols allowing randomization, and many of the details are significantly different from lower-bound techniques found in the population protocols literature. We remark that a recent line of work in this area~\cite{pplb1,pplb2} provides a powerful set of tools for proving lower bounds on the number of states (typically $\Omega(\log \log n)$ states) for fast (typically polylogarithmic) population protocols for different problems, especially for the case of deterministic protocols. We were unable to leverage these results to prove our lower bound for the randomized scenario studied here, and believe our coupling analysis is complementary to their results.

\subsection{Other Related Work}\label{related}

Our work fits into the line of research on rumor spreading, population protocols, and related interaction models.
Our work also touches on the issue of how distributed systems may spontaneously achieve some form of coordination with minimum agent capabilities. The basic work in this direction, starting with the seminal paper~\cite{L}, focuses on synchronizing timers through asynchronous interprocess communication to allow processes to construct a total ordering of events. A separate interesting question concerns local clocks which, on their own, have some drift, and which need to synchronize in a network environment (cf. e.g.~\cite{LL84,LLW}, or~\cite{LLSW} for a survey of open problems).

\paragraph{Rumor spreading.} Rumor spreading protocols are frequently studied in a synchronous setting. In a synchronous protocol, in each parallel round, each vertex independently at random activates a local rule, which allows it either to spread the rumor (if it is already informed), or possibly also to receive it (if it has not yet been informed, as is the case in the push-pull model). The standard push rumor spreading model assumes that each informed neighbor calls exactly one uninformed neighbor. In the basic scenario, corresponding to the complete interaction network, the number of parallel rounds for a single rumor source to inform all other nodes is given as $\log_2 n + \ln n + o(\log n)$, with high probability~\cite{rs,rs2}. More general graph scenarios have been studied in~\cite{rsnetworks} in the context of applications in broadcasting information in a network. Graph classes studied for the graph model include hypercubes~\cite{rsnetworks}, expanders~\cite{rssauerwald}, and other models of random graphs~\cite{fou}. The push-pull model of rumor spreading is an important variation: whereas for complete networks the speedup due to the pull process is in the order of a multiplicative constant~\cite{rspushpull}, the speed up turns out to be asymptotic, e.g., on preferential attachment graphs, where the rumor spreading time is reduced from $\Theta(\log n)$ rounds in the push model to $\Theta(\log n / \log \log n)$ rounds in the push-pull model~\cite{rsgraphs}, as well as on other graphs with a non-uniform degree distribution. The push-pull model often also proves more amenable to theoretical analysis. We note that asynchronous rumor spreading on graphs, in models closer to our random scheduler, has also been considered in recent work~\cite{PS,gnw}, with~\cite{gnw} pointing out the tight connections between the synchronous (particularly push-pull) and asynchronous models in general networks.

\paragraph{Population protocols.} Population protocols are a model which captures the way in which the complex behavior of systems (biological, sensor nets, etc.) emerges from the underlying local interactions of agents.
The original model of Angluin et al.\ \cite{AADFP06,angluin2008} was motivated by applications in sensor mobility. Despite the limited computational capabilities of individual sensors, such protocols permit at least (depending on available extensions to the model) the computation of two important classes of functions: threshold predicates, which decide if the weighted average of types appearing in the population exceeds a certain value, and modulo remainders of similar weighted averages. The majority function, which belongs to the class of threshold functions, was shown to be stably computable for the complete interaction graph~\cite{AADFP06}; further results in the area of majority computation can be found in~\cite{angluin2008,AR09,Mertzios2016,BCNPST14}. A survey of applications and models of population protocols is provided in~\cite{AR09,MCS11}. An interesting line of research is related to studies of the algorithmic properties of dynamics of chemical reaction networks~\cite{pplb1}. These are as powerful as population protocols, though some extensions of the chemical reaction model also allow the population size to change in time. Two very recent results in the population protocol model are worthy of special attention. Alistarh, Aspnes, and Gelashvili~\cite{AAG18} have resolved the question of the number of states required to solve the Majority problem on a complete network in polylogarithmic time as $\Theta(\log n)$. For the equally notable task of Leader Election, the papers of Gasieniec and Stachowiak~\cite{grzes} (for the upper bound) together with the work of Alistarh, Aspnes, Eisenstat, Gelashvili, and Rivest~\cite{pplb2} (for the lower bound) put the number of states required to resolve this question in polylogarithmic time as $\Theta(\log \log n)$. Both of these results rely on a notion of a self-organizing phase clock.

\paragraph{Nonlinearity in interaction protocols.}

Linear dynamical systems, as well as many nonlinear protocols subjected to rigorous analytical study, have a relatively simple structure of point attractors and repellers in the phase space. The underlying continuous dynamics (in the limit of $n \to +\infty$) of many interaction protocols defined for complete graphs would fit into this category: basic models of randomized rumor spreading~\cite{rs}; models of opinion propagation (e.g.~\cite{voting,AD15}); population protocols for problems such as majority and thresholds~\cite{AADFP06,angluin2008}; all reducible Markov chain processes, such as random walks and randomized iterative load balancing schemes.

Nonlinear dynamics with non-trivial limit orbits are fundamental to many areas of systems science, including the study of physical, chemical and biological systems, and to applications in control science. In general, population dynamics with interactions between pairs of agents are non-linear (representable as a set of quadratic difference equations) and have potentially complicated structure if the number of states is $3$ or more. For example, the simple continuous Lotka-Volterra dynamics~\cite{L1910} gives rise to a number of discrete models, for example one representing interactions of the form $A + B \to A + A$, over some pairs $A, B$ of states in a population (cf.~\cite{SMJSRP14} for further generalizations of the framework or~\cite{ICALP15} for a rigorous analysis in the random scheduler model). The model describes transient stability in a setting in which several species are in a cyclic predator-prey relation. Cyclic protocols of the type have been consequently identified as a potential mechanism for describing and maintaining biodiversity, e.g., in bacterial colonies~\cite{KRFB02,KR04}. Cycles of length 3, in which type $A_2$ attacks type $A_1$, type $A_3$ attacks type $A_2$, and type $A_1$ attacks type $A_3$, form the basis of the basic oscillator, also used as the starting point for protocols in this work, which is referred to as the RPS (rock-paper-scissors) oscillator or simply the 3-cycle oscillator, which we discuss further in Section~\ref{sec:preli}. This protocol has been given a lot of attention in the statistical physics literature. The original analytical estimation method applied to RPS was based on approximation with the Fokker-Planck equation~\cite{RMF06}. A subsequent analysis of cyclic $3$- and $4$-species models using Khasminskii stochastic averaging can be found in~\cite{DF12}, and a mean field approximation-based analysis of RPS is performed in~\cite{PK09}. In~\cite{ICALP15}, we have performed a study of some algorithmic implications of RPS, showing that the protocol may be used to perform randomized choice in a population, promoting minority opinions, in $\tilde O(n^2)$ steps. All of these results provide a good qualitative understanding of the behavior of the basic cyclic protocols. We remark that the protocol used in this paper is directly inspired by the properties of RPS, as we discuss further on, but has a more complicated interaction structure (see Fig.~\ref{fig:rules}).

For protocols with convergence to a single point in the configuration space in the limit of large population size, a discussion of the limit behavior is provided in~\cite{BournezFK12}, who provide examples of protocols converging to limit points at coordinates corresponding to any algebraic numbers.

We also remark that local interaction dynamics on arbitrary graphs (as opposed to the complete interaction graph) exhibit a much more complex structure of their limit behavior, even if the graph has periodic structure, e.g., that of a grid. Oscillatory behavior may be overlaid with spatial effects~\cite{SMJSRP14}, or the system may have an attractor at a critical point, leading to simple dynamic processes displaying self-organized criticality (SOC,~\cite{soc}).

\section{Preliminaries: Building Blocks for Population Protocols}

\subsection{Protocol Definition}

A randomized \emph{population protocol} for a population of $n$ agents is defined as a pair $P = (K_n, R_n)$, where $K_n$ is the set of \emph{states} and $R_n$ is the set of \emph{interaction rules}. The interaction graph is complete. We will simply write $P = (K, R)$, when considering a protocol which is \emph{universal} (i.e., defined in the same way for each value of $n$) or if the value of $n$ is clear from the context. All the protocols we design are universal; our lower bounds also apply to some non-universal protocols. The set of rules $R \subseteq{K^4 \times [0,1]}$ is given so that each rule $j \in R$ is of the form $j = (i_1(j), i_2(j), o_1(j), o_2(j), q_j)$, describing an interaction read as: $``(i_1(j), i_2(j)) \to (o_1(j), o_2(j)) \text{with probability $q_j$}$''. For all $i_1, i_2 \in K$, we define $R_{i_1, i_2} = \{j \in R : (i_1(j),i_2(j)) = (i_1,i_2)\}$ as the set of rules acting on the pair of states $i_1, i_2$, and impose that $\sum_{j\in R_{i_1,i_2}} q_j \leq 1$.

For a state $A\in K$, we denote the number of agents in state $A$ as $\many A$, and the \emph{concentration} of state $A$ as $a=\many A/n$, and likewise for a set of states $\mathcal A$, we write $\many \mathcal A = \sum_{A \in \mathcal{A}}\many A$.

In any configuration of the system, each of the $n$ agents from the population is in one of states from $K_n$. The protocol is executed by an asynchronous scheduler, which runs in \emph{steps}. In every step the scheduler uniformly at random chooses from the population a pair of distinct agents to interact: the \emph{initiator} and the \emph{receiver}. If the initiator and receiver are in states $i_1$ and $i_2$, respectively, then the protocol executes at most one rule from set protocol $R_{i_1, i_2}$, selecting rule $j\in R_{i_1, i_2}$ with probability $q_j$. If rule $j$ is executed, the initiator then changes its state to $o_1(j)$ and the receiver to $o_2(j)$. The source has a special state, denoted $X$ in the $\textsc{Detect}$ problem, or one of two special states, denoted $\{X_1,X_2\}$ in the $\textsc{BitBroadcast}$ problem, which is never modified by any rule.

All protocols are presented in the randomized framework, however, the universal protocols considered here are amenable to a form of conversion into deterministic rules discussed in~\cite{pplb2}, which simulates randomness of rules by exploiting the inherent randomness of the scheduler in choosing interacting node pairs to distribute weakly dependent random bits around the system.

All protocols designed in this work are \emph{initiator-preserving}, which means that for any rule $j \in R$, we have $o_1 (j) = i_1 (j)$ (i.e., have all rules of the form $A + B \to A + C$, also more compactly written as $A \act\!\! B \to C$), which makes them relevant in a larger number of application. As an illustrative example, we remark that the basic rumor spreading (epidemic) model is initiator-preserving and given simply as $1 \act\!\! 0 \to 1$. All protocols can also obviously be rewritten to act on unordered pairs of agents picked by the scheduler, rather than ordered pairs.

\subsection{Protocol Composition Technique}

Our protocols will be built from simpler elements. Our basic building block is the input-controlled oscillatory protocol $P_o$ (see Fig.~\ref{fig:rules}).
We then use protocol $P_o$ as a component in the construction of other, more complex protocols, without disrupting the operation of the original protocol.

Formally, we consider a protocol $P_B$ using state set $B = \{B_i : 1 \leq i \leq k_b\}$ and rule set $R_B$, and a \emph{protocol extension} $P_{BC}$ using a state set $B \times C = B \times \{C_i : 1 \leq i \leq k_c\}$, where $C$ is disjoint from $B$, and \emph{rule extension} set $R_{BC}$. Each rule extension defines for each pair of states from $B \times C$ (i.e., to each element of $(B \times C) \times (B \times C)$) a probability distribution over elements of $C \times C$.

The \emph{composed protocol} $P_B \circ P_{BC}$ is a population protocol with set of states $B \times C$. Its rules are defined so that, for a selected pair of agents in states $(B_i, C_j)$ and $(B_{i'}, C_{j'})$, we obtain a pair of agents in states $(B_{i^*}, C_{j^*})$ and $(B_{i'^*}, C_{j'^*})$ according to a probability distribution defined so that:
\begin{itemize}
\item Each pair $B_{i'^*}, B_{i'^*}$ appears in the output states of the two agents with the same probability as it would in an execution of protocol $P_B$ on a pair of agents in states $B_i$ and $B_{i'}$.
\item
    Each pair $C_{i'^*}, C_{i'^*}$ appears in the output states of the two agents with the probability given by the definition of $P_{BC}$.
\end{itemize}
In the above, the pairs of agents activated by $P_B$ and $P_{BC}$ are not independent of each other. This is a crucial property in the composition of protocol $P_o$ when composing it with further blocks to solve the \textsc{Detect} problem.

We denote by $\mathbf 1_B$ the identity protocol which preserves agent states on set of states $B$. For a protocol $P$, we denote by $P/2$ a lazy version of a protocol $P$ in which the rule activation of $P$ occurs with probability $1/2$, and with probability $1/2$  the corresponding rule of the identity protocol is activated. Note that all asymptotic bounds on expected and w.h.p. convergence time obtained for any protocol $P$ also apply to protocol $P/2$, in the regime of at least a logarithmic number of time steps. We also sometimes treat a protocol $P_{BC}$ extension as a protocol in itself, applied to the identity protocol $\mathbf 1_B$.

The \emph{independently composed protocol} $P_B + P_{BC}$ is defined as an implementation of the composed protocol $(P_B) \circ (P_C/2)$, realized with the additional constraint that in each step, either the rule of $P_B$ is performed with an identity rule extension, or the rule extension of $P_{BC}$ is performed on top of the identity protocol $\mathbf 1_B$. Such a definition is readily verified to be correct by a simple coupling argument, and allows us to analyze protocols $P_B$ and $P_{BC}$, observing that the pairs (identities) of agents activated by the scheduler in the respective protocols are independent.

All the composed protocols (and protocol extensions) we design are also initiator-preserving, i.e., $C_{i^*} = C_{i}$ and $B_{i^*} = B_{i}$, with probability $1$. In notation, rules omitted from the description of protocol extensions are implicit, occurring with probability $0$ (where $C_{j^*} \neq C_{j}$) or with the probability necessary for the normalization of the distribution to $1$, where the state is preserved (where $C_{j^*} = C_{j}$).

As a matter of naming convention, we name the states in the separate state sets of the composed protocols with distinct letters of the alphabet, together with their designated subscripts and superscripts. The rumor source $X$ is treated specially and uses a separate letter (and may be seen as a one state protocol without any rules, on top of which all other protocols are composed; in particular, its state is never modified).  The six remaining states of protocol $P_o$ are named with the letters $A_{_?}^{^?}$, as usual in its definition. Subsequent protocols will use different letters, e.g., $M_?$ and $L_?$.

\section{Overview of Protocol Designs}\label{sec:protocols}

\subsection{Main Routine: Input-Controlled Oscillator Protocol \texorpdfstring{$P_o$}{Po}}\label{sec:osc}

We first describe the main routine which allows us to convert local input parameters (the existence of source into a form of global periodic signal on the population. This main building block is the construction of a $7$-state protocol $P_o$ following oscillator dynamics, whose design we believe to be of independent interest.

The complete design of protocol $P_o$ is shown in~Fig.~\ref{fig:rules}. The source state is denoted by $X$. Additionally, there are six states, called $A_i^+$ and $A_i^{++}$, for~$i\in\{1,2,3\}$. The naming of states in the protocol is intended to maintain a direct connection with the RPS oscillator dynamics, which is defined by the simple rule ``$A_i \act A_{i-1}\mapsto A_i$, for $i=1,2,3$''. In fact, we will retain the convention $A_i = \{A_{i}^{+}, A_{i}^{++}\}$ and $a_i = a_{i}^{+} + a_{i}^{++}$, and consider the two states $A_i^+$ and $A_i^{++}$ to be different flavors of the same species $A_i$, referring to the respective superscripts as either lazy ($^+$) or aggressive ($^{++}$).

\begin{figure}[t]
\begin{minipage}{\textwidth}
\begin{framed}
\vspace*{-8mm}
\begin{alignat*}{4}
\intertext{(1) Interaction with an initiator from the same species makes receiver aggressive:}
A_i^{^?} &\act A_{i}^{^?} && \mapsto \phantom{\ \ \ \;\!}A_{i}^{++}\\[1mm]
\intertext{(2) Interaction with an initiator from a different species makes receiver lazy (case of no attack):}
A_i^{^?} &\act A_{i+1}^{^?} && \mapsto \phantom{\ \ \ \;\!}A_{i+1}^{+\phantom+}\\[1mm]
\intertext{(3) A lazy initiator has probability $p$ of performing a successful attack on its prey:}
A_i^+ &\act A_{i-1}^{^?} && \mapsto \begin{cases}A_{i}^+, &\text{with probability $p$,}\\
A_{i-1}^+, &\text{otherwise.}
\end{cases}\\
\intertext{(4) An aggressive initiator has probability $2p$ of performing a successful attack on its prey:}
A_i^{++}\! &\act A_{i-1}^{^?} && \mapsto \begin{cases}A_{i}^+, &\text{with probability $2p$,}\\
A_{i-1}^+, &\text{otherwise.}
\end{cases}\\
\intertext{(5) The source converts any receiver into a lazy state of a uniformly random species:}
X &\act A_{_?}^{^?} && \mapsto \begin{cases}A_1^+, &\text{with probability $1/3$,}\\
A_2^+, &\text{with probability $1/3$,}\\
A_3^+, &\text{with probability $1/3$.}\\
\end{cases}
\end{alignat*}
\vspace*{-5mm}
\end{framed}
\end{minipage}
\caption{Rules of the basic oscillator protocol $P_o$. The adopted notation for initiator-preserving rules is of the form $A \act\!\!\!\!\! B \mapsto C$, corresponding to transitions written as $A + B \to A + C$ in the notation of chemical reaction networks or $(A,B) \to (A,C)$ in the notation of population dynamics. All rules apply to $i\in\{1,2,3\}$, whereas a question mark $?$ in a superscript or subscript denotes a wildcard, matching any permitted combination of characters, which may be set independently for each agent. Probability $p>0$ is any (sufficiently small) absolute constant.}
\label{fig:rules}
\end{figure}

The protocol has the property that in the absence of $X$, it stops in a corner state of the phase space, in which only one of three possible states appears in the population, and otherwise regularly (every $O(\log n)$ steps) moves sufficiently far away from all corner states. An intuitive formalization of the basic properties of the protocol is given by the theorem below.

\begin{theorem}\label{thm:osc}
There exists a universal protocol $P_o$ with $|K|=7$ states, including a distinguished source state $X$, which has the following properties.
\begin{enumerate}
\item For any starting configuration, in the absence of the source $(\many X=0)$, the protocol always reaches a configuration such that:
\begin{itemize}
\item all agents are in the same state: either $A_1^{++}$, or $A_2^{++}$, or $A_3^{++}$;
\item no further state transitions occur after this time.
\end{itemize}
Such a configuration is reached in $O(\log n)$ rounds, with constant probability (and in $O(\log^2 n)$ rounds with probability $1-O(1/n)$).
\item For any starting configuration, in the presence of the source $(\many X \geq 1)$, we have with probability $1-O(1/n)$:
\begin{itemize}
\item for each state $i\in K$, there exists a time step in the next $O(\log \frac{n}{\many X})$ rounds when at least a constant fraction of all agents are in state $i$;
\item during the next $O(\log \frac{n}{\many X})$ rounds, at least a constant fraction of all agents change their state at least once.
\end{itemize}
\end{enumerate}
\end{theorem}

The proof of the Theorem is provided in Section~\ref{sec:upper}.

The RPS dynamics provides the basic oscillator mechanism which is still largely retained in our scenario. Most of the difficulty lies in controlling its operation as a function of the presence or absence of the rumor source. We do this by applying two separate mechanisms. The presence of rumor source $X$ shifts the oscillator towards an orbit closer to the central orbit $(A_1, A_2, A_3) = (1/3, 1/3, 1/3)$ through rule $(5)$, which increases the value of potential $\phi := \ln (a_1 a_2 a_3)$, where $a_i = \many A_i / n$. Conversely, independent of the existence of rumor source $X$, a second mechanism is intended to reduce the value of potential $\phi$. This mechanism exploits the difference between the aggressive and lazy flavors of the species. Following rule $(1)$, an agent belonging to a species becomes more aggressive if it meets another from the same species, and subsequently attacks agents from its prey species with doubled probability following rule $(4)$. This behavior somehow favors larger species, since they are expected to have (proportionally) more aggressive agents than the smaller species (in which pairwise interactions between agents of the same species are less frequent) --- the fraction of agents in $A_i$ which are aggressive would, in an idealized static scenario, be proportional to $a_i$. (This is, in fact, often far from true due to the interactions between the different aspects of the dynamics). As a very loose intuition, the destabilizing behavior of the considered rule on the oscillator is resemblant of the effect an eccentrically fitted weight has on a rotating wheel, pulling the oscillator towards more external orbits (with smaller values of $\phi$).

The intuition for which the proposed dynamics works, and which we will formalize and prove rigorously in Section~\ref{sec:upper}, can now be stated as follows: in the presence of rumor source $X$, the dynamics will converge to a form of orbit on which the two effects, the stabilizing and destabilizing one, eventually compensate each other (in a time-averaged sense). The period of a single rotation of the oscillator around such an orbit is between $O(1)$ and $O(\log n)$, depending on the concentration of $X$. In the absence of $X$, the destabilizing rule will prevail, and the oscillator will quickly crash into a side of the triangle.

For small values of $\many X>0$, the protocol can be very roughly (and non-rigorously) viewed as cyclic composition of three dominant rumor spreading processes over three sets of states $A_1$, $A_2$, $A_3$, one converting states $A_1$ to $A_3$, the next from $A_3$ to $A_2$, and the last from $A_2$ to $A_1$, which spontaneously take over at moments of time separated by $O(\log n)$ parallel rounds. For other starting configurations, and especially for the case of $\many X = 0 $, the dynamics of the protocol, which has $5$ free dimensions, is more involved to describe and analyze (see Section~\ref{sec:32}). We provide some further insights into the operation of the protocol in Section~\ref{xx}, notably formalizing the notion that an intuitively understood oscillation (going from a small number of agents in some state $A_i$, to a large number of agents in state $A_i$, and back again to a small number of agents in state $A_i$) takes $\Theta(\log \frac{n}{\many X})$ steps, with probability $1-O(1/n)$. As such, protocol $P_o$ can be seen as converting \emph{local input} $\many X$ into a \emph{global periodic signal} with period $\Theta(\log \frac{n}{\many X})$. What remains is allowing nodes to extract information from this periodic signal.

Simulation timelines shown in Fig.\ref{fig:osc_3x3} in the Appendix illustrate the idea of operation of protocol $P_o$ and its composition with other protocols.

\subsection{Protocols for \textsc{BitBroadcast}}

A solution to \textsc{BitBroadcast} is obtained starting with an independent composition of two copies of oscillator $P_o$, called $P_{o[1]}$ and $P_{o[2]}$, with states in one protocol denoted by subscript $[1]$ and in the other by subscript $[2]$. The respective sources are thus written as $X_{[1]}$ and $X_{[2]}$. In view of Theorem~\ref{thm:osc}, in this composition $P_{o[1]} + P_{o[2]}$, under the promise of the \textsc{BitBroadcast} problem, one of the oscillators will be running and the other will stop in a corner of its state space. Which of the oscillators is running can be identified by the presence of states $A_i^+[z]$, which will only appear for $z \in \{1,2\}$ corresponding to the operating oscillator. Moreover, by the same Theorem, every $O(\log n)$ rounds a constant number of agents of this oscillator will be in such a state $A_i^+[z]$, for any choice of $i \in \{1,2,3\}$. We can thus design the protocol extension $P_b$ to detect this. This is given by the pair of additional output states $\{Y_1, Y_2\}$ and the rule extension consisting of the two rules shown in Fig.~\ref{fig:Pb}.
\begin{figure}[t]
\begin{minipage}{\textwidth}
\begin{framed}
\vspace*{-8mm}
\begin{alignat*}{4}
\intertext{(6a) Interaction with agent state $A_{?[1]}^+$ sets agent output to $Y_1$:}
(A_{?[1]}^+, A_{?[2]}^{++}, Y_?) &\act (A_{?[1]}^?, A_{?[2]}^?, Y_?) &\mapsto Y_1\\[1mm]
\intertext{(6a) Interaction with agent state $A_{?[2]}^+$ sets agent output to $Y_2$:}
(A_{?[1]}^{++}, A_{?[2]}^+, Y_?) &\act (A_{?[1]}^?, A_{?[2]}^?, Y_?) &\mapsto Y_2
\end{alignat*}
\vspace*{-5mm}
\end{framed}
\end{minipage}
\caption{Protocol extension $P_b$ of protocol $(P_{o[1]} + P_{o[2]})$, with additional states $\{Y_1, Y_2\}$.}
\label{fig:Pb}
\end{figure}

\begin{theorem}[Protocol for \textsc{BitBroadcast}]\label{thm:br}
Protocol $(P_{o[1]} + P_{o[2]}) + P_b$, having $|K|=74$ states, including distinguished source states $X_{[1]}$, $X_{[2]}$ converges to an exact solution of \textsc{BitBroadcast}. This occurs in $O(\log^2 n)$ parallel rounds, with probability $1- O(1/n)$. In the output encoding, agent states of the form $(\cdot, \cdot, Y_z)$ represent answer ``$z$'', for $z\in \{1,2\}$.
\end{theorem}

The protocol $(P_{o[1]} + P_{o[2]}) + P_b$ is not ``silent'', i.e., it undergoes perpetual transitions of state, even once the output has been decided. As a side remark, we note that for the single-source broadcasting problem, or more generally for the case when the number of sources is small, $\max\{\many X_{[1]}, \many X_{[2]}\} = O(1)$, we can propose the following simpler silent protocol. We define protocol $P_o'$, by modifying protocol $P_o$ as follows. We remove from it Rule (5), and replace it by to the four rules shown in Fig.~\ref{fig:ss_broadcast}. The analysis of the modified protocol follows from the same arguments as those used to prove Theorem~\ref{thm:osc}(1). In the regime of $\max\{\many X_{[1]}, \many X_{[2]}\} = O(1)$, the effect of the source does not influence the convergence of the process and each of the three possible corner configurations, with exclusively species $\{A_1, A_2, A_3\}$, is reached in $O(\log n)$ steps with constant probability. However, rules $(5a)-(5d)$ enforce that the only stable configuration which will persist is the one in a corner corresponding to the identity of the source, i.e., $A_1$ for source $X_{[1]}$ and $A_2$ for source $X_{[2]}$; the source will restart the oscillator in all other cases. We thus obtain the following side result, for which we leave out the details of the proof.

\begin{figure}[t]
\begin{minipage}{\textwidth}
\begin{framed}
\vspace*{-8mm}
\begin{alignat*}{4}
\intertext{(5a-5b) Interaction with source $X_{[1]}$ alters corner configurations different from $A_1$:}
X_{[1]} &\act A_3^? & \mapsto A_1^+\\[1mm]
X_{[1]} &\act A_2^? & \mapsto A_3^+\\[1mm]
\intertext{(5c-5d) Interaction with source $X_{[2]}$ alters corner configurations different from $A_2$:}
X_{[2]} &\act A_1^? & \mapsto A_2^+\\[1mm]
X_{[2]} &\act A_3^? & \mapsto A_1^+
\end{alignat*}
\vspace*{-5mm}
\end{framed}
\end{minipage}
\caption{Protocol $P_o'$ is obtained from the basic oscillator protocol $P_o$ by removing Rule (5) and replacing it by the two rules shown above.}
\label{fig:ss_broadcast}
\end{figure}

\begin{observation}
Protocol $P_{o}'$, having $|K|=6+2 = 8$ states, including distinguished source states $X_{[1]}$, $X_{[2]}$ converges to an exact solution of \textsc{BitBroadcast}, eventually stopping with all agents in state $A_1^{++}$ if source $X_{[1]}$ is present and stopping with all agents in state $A_2^{++}$ if source $X_{[2]}$ is present, with no subsequent state transition. The stabilization occurs within $O(\log^2 n)$ parallel rounds, with probability $1- O(1/n)$, if $\max\{\many X_{[1]}, \many X_{[2]}\} = O(1)$, i.e., the broadcast originates from a constant number of sources.
\end{observation}

\subsection{Protocol for \textsc{Detect}}

The solution to problem~\textsc{Detect} is more involved. It relies on two auxiliary extensions added on top of a single oscillator $P_o$. The first, $P_m$, runs an instance of the 3-state majority protocol of Angluin et al.~\cite{angluin2008} \emph{within} each species $A_i$ of the oscillator. For this reason, the composition between $P_o$ and $P_m$ has to be of the form $P_o \circ P_m$ (i.e., it cannot be independent). The operation of this extension is shown in Fig.~\ref{fig:newrules_l} and analyzed in Section~\ref{pm}. It relies crucially on an interplay of two parameters: the time $\Theta(\log \frac{n}{\many X})$ taken by the oscillator to perform an orbit, and the time $\Omega(\log \frac{n}{\many X})$ it takes for the majority protocol (which is reset by the oscillator in its every oscillation) to converge to a solution. When parameters are tuned so that the second time length is larger a constant number of times than the first, a constant proportion of the agents of the population are involved in the majority computation, i.e., both of the clashing states in the fight for dominance still include $\Omega(n)$ agents. In the absence of a source, shortly after the oscillator stops, one of these states takes over, and the other disappears.

\begin{figure}
\begin{minipage}{\textwidth}
\begin{framed}
\vspace*{-8mm}
\begin{alignat*}{4}
\intertext{(6) Meeting an agent of a different type resets majority setting:}
&& ((X \text{ or } A_j^?),M_?) &\act (A_i^?,M_?) && {\mapsto} \begin{cases}
M_{+1}, & \text{with probability $1/2$}\\
M_{-1}, & \text{with probability $1/2$}\\
\end{cases}, \quad\text{for $j\neq i$.}\\[1mm]
\intertext{(7-10) Three-state majority protocol among agents of the same type:}
(7):&& (A_i^?,M_{-1}) &\act (A_i^?,M_{+1}) && {\mapsto} M_0, \quad\text{with probability $r$}.\\
(8):&&(A_i^?,M_{+1}) &\act (A_i^?,M_{-1}) && {\mapsto} M_0, \quad\text{with probability $r$}.\\
(9):&&(A_i^?,M_{+1}) &\act (A_i^?,M_0) && {\mapsto} M_{+1}, \quad\text{with probability $r$}.\\
(10):&&(A_i^?,M_{-1}) &\act (A_i^?,M_0) && {\mapsto} M_{-1}, \quad\text{with probability $r$}.
\end{alignat*}
\vspace*{-5mm}
\end{framed}
\end{minipage}
\caption{Protocol extension $P_m$ of protocol $P_o$, with additional states $\{M_{-1}, M_0, M_{+1}\}$. This protocol extension is applied through the composition $(P_o \circ P_m)$. Probability $r>0$ is given by an absolute constant, depending explicitly on $s$ and $p$, whose value is sufficiently small.}
\label{fig:newrules_m}
\end{figure}

The above-described difference can be detected by the second, much simpler, extension $P_l$, designed in Fig.~\ref{fig:newrules_l} and analyzed in Section~\ref{pl}. The number of ``lights'' switched on during the operation of the protocol will almost always be more than $(1-\eps)n$, where $\eps > 0$ is a parameter controlled by the probability of lights spontaneously disengaging, and may be set to and arbitrarily small constant.

\begin{figure}
\begin{minipage}{\textwidth}
\begin{framed}
\vspace*{-8mm}
\begin{alignat*}{4}
\intertext{(11) Light switch progresses from $L_{-1}$ to $L_{+1}$ in the presence of initiator $M_{-1}$:}
(A_?^?, M_{-1}, L_?)  & \act (A_?^?, M_{?}, L_{-1}) && \mapsto L_{+1}\\[1mm]
\intertext{(12) Light switch progresses from $L_{+1}$ to $L_{\mathit{on}}$ in the presence of initiator $M_{+1}$:}
(A_?^?, M_{+1}, L_?)  & \act (A_?^?, M_{?}, L_{+1}) && \mapsto L_{\mathit{on}}\\[1mm]
\intertext{(13) Light spontaneously turns off:}
(A_?^?, M_?, L_?)  & \act (A_?^?, M_{?}, L_{\mathit{on}}) && \mapsto L_{-1} && ,\quad \text{with probability $q(\eps)$.}
\end{alignat*}
\vspace*{-5mm}
\end{framed}
\end{minipage}
\caption{Protocol extension $P_l$ of protocol $(P_o \circ P_m)$, with additional states $\{L_{-1}, L_{+1}, L_{\mathit{on}}\}$. This protocol extension is applied through the composition $(P_o \circ P_m) + P_l$. Probability $q(\eps)>0$ is given by an absolute constant, depending explicitly on $s$, $p$, $r$, and $\eps$, whose value is sufficiently small.}
\label{fig:newrules_l}
\end{figure}

\begin{theorem}[Protocol for \textsc{Detect}]\label{thm:pr}
For any $\eps>0$, protocol $(P_o \circ P_m) + P_l$, having $|K|=6\cdot 3 \cdot 3 + 1 = 55$ states, including a distinguished source state $X$, which solves the problem of spreading confirmed rumors as follows:
\begin{enumerate}
\item For any starting configuration, in the presence of the source $(\many X \geq 1)$, after an initialization period of $O(\log n)$ rounds, at an arbitrary time step the number of agents in an output state corresponding to a ``yes'' answer is $(1-\eps) n$, with probability $1 - O(1/n)$.
\item For any starting configuration, in the absence of the source $(\many X=0)$, the system always reaches a configuration such that all agents are in output states corresponding to a ``no'' answer for all subsequent time steps. Such a configuration is reached in $O(\log^2 n)$ rounds, with probability $1- O(1/n)$.
\end{enumerate}
\end{theorem}

\section{Impossibility Results for Protocols without Non-Stationary Effects}\label{sec:impossibility}

For convenience of notation, we identify a configuration of the population with a vector $z = (z^{(1)}, \ldots, z^{(k)}) \in \{0,1,\ldots,n\}^k = Z$, where $z^{(i)}$, for $1\leq i \leq k$, denotes the number of agents in the population having state $i$, and $\on{z}=n$. Our main lower bound may now be stated as follows.

\begin{theorem}[Fixed points preclude fast stabilization]\label{thm:lb}
Let $\eps_1 >0$ be arbitrarily chosen, let $P$ be any $k$-state protocol, and let $z_0$ be a configuration of the system with at most $n^{\eps_0}$ agents in state $X$, where $\eps_0 \in (0, \eps_1]$ is a constant depending only on $k$ and $\eps_1$. Let $B$ be a subset of the state space around $z_0$ such that the population of each state within $B$ is within a factor of at most $n^{\eps_0}$ from that in $z_0$ (for any $z \in B$, for all states $i \in \{1,\ldots, k\}$,  we have $z_0^{(i)}/n^{\eps_0} <  z^{(i)} \leq n^{\eps_0} \max\{1,z_0^{(i)}\}$).

Suppose that in an execution of $P$ starting from configuration $z_0$, with probability $1 - o(1)$, the configurations of the system in the next $n^{2\eps_1}$ parallel rounds are confined to $B$.

Then, an execution of $P$ for $n^{2\eps_0}$ parallel rounds, starting from a configuration in which state $X$ has been removed from $z_0$, reaches a configuration in a $O(n^{6\eps_1})$-neighborhood of $B$, with probability $1-o(1)$.
\end{theorem}

In the statement of the Theorem, for the sake of maintaining the size of the population, we interpret ``removing state $X$ from $z_0$'' as replacing the state of all agents in state $X$ by some other state, chosen adversarially (in fact, this may be any state which has sufficiently many representatives in configuration $z_0$). The $O(n^{6\eps_1})$-neighborhood of $B$ is understood in the sense of the $1$-norm or, asymptotically equivalently, the total variation distance, reflecting configurations which can be converted into a configuration from $B$ by flipping the states of $O(n^{6\eps_1})$ agents.

The proof of Theorem~\ref{thm:lb} is provided in Section~\ref{sec:lb}. It proceeds by a coupling argument between a process starting from $z_0$ and a perturbed process in which state $X$ has been removed. The analysis differently treats rules and states which are seldom encountered during the execution of the protocol from those that are encountered with polynomially higher probability (such a clear separation is only possible when $k = O(\poly \log \log n)$). Eventually, the probability of success of the coupling reduces to a two-dimensional biased random walk scenario, in which the coordinates represent differences between the number of times particular rules have been executed in the two coupled processes.

We have the following direct corollaries for the problems we are considering. For \textsc{Detect}, if $B$ represents the set of configurations of the considered protocol, which are understood as the protocol giving the answer ``$\many X > 0$'', then our theorem says that, with probability $1-o(1)$, the vast majority of agents will not ``notice'' that $\many X$ had been set to $0$, even a polynomial number of steps after this has occurred, and thus cannot yield a correct solution. An essential element of the analysis is that it works only when state $X$ is removed in the perturbed process. Thus, there is nothing to prevent the dynamics from stabilizing even to a single point in the case of $X=0$, which is indeed the case for our protocol $P_r$. The argument for~\textsc{BitBroadcast} only applies to situations where the source agent is sending out white noise (independently random bits in successive interactions). Such a source can be interpreted as a pair of sources in states $X_1$ and $X_2$ in the population, each disclosing itself with probability $1/2$ upon activation and staying silent otherwise. In the cases covered by the lower bound, the scenario in which the source $X_1$ is completely suppressed cannot be distinguished from the scenario in which both $X_1$ and $X_2$ appear; likewise, the scenario in which the source $X_2$ is completely suppressed cannot be distinguished from the scenario in which both $X_1$ and $X_2$ appear. By coupling all three processes, this would imply the indistinguishability of the all these configurations, including those with only source $X_1$ and only source $X_2$, which would imply incorrect operation of the protocol.

Whereas we use the language of discrete dynamics for precise statements, we informally remark that the protocols covered by the lower bound of Theorem~\ref{thm:lb} include those whose dynamics $z_t/n$, described in the continuous limit $(n \to +\infty)$, has only point attractors,  repellers, and fixed points. In this sense, the use of oscillatory dynamics in our protocol seems inevitable. 

The impossibility result is stated in reference to protocols with a constant number of states, however, it may be extended to protocols with a non-constant number of states $k$, showing that such protocols require $n^{\exp[-O(\poly(k))]}$ time to reach a desirable output. (This time is larger than polylogarithmic up to some threshold value $k = O(\poly \log \log n)$.) The lower bound covers randomized protocols, including those in which rule probabilities depend on $n$ (i.e., non-universal ones).

\section{Input-Controlled Behavior of Protocols for \textsc{Detect}}\label{sec:signal}

In this Section, we consider the periodicity of protocols for self-organizing oscillatory dynamics, in order to understand how the period of a phase clock must depend on the input parameters. We focus on the setting of the \textsc{Detect} problem, considering the value $\many X$ of the input parameter. In Section~\ref{sec:osc}, we noted informally that the designed oscillatory protocol performs a complete rotation around the triangle in $\Theta(\log n/ \many X)$ rounds. Here, we provide partial evidence that the periodicity of any oscillatory protocol depends both on the value of $\many X$ and $n$. We do this by bounding the portions of the configuration space in which a protocol solving \textsc{Detect} finds itself in most time steps, separating the cases of sub-polynomial $\many X$ (i.e., $\many X < n^{\eps_0}$, where $\eps_0 > 0$ is a constant dependent on the specific protocol), and the case of $\many X = \Theta (n)$.


Any protocol on $k$ states (not necessarily of oscillatory nature) can be viewed as a Markov chain in its $k$-dimensional configuration space $[0,n]^{k}$, and as in Section~\ref{sec:impossibility} we identify a configuration with a vector $z \in \{0,1,\ldots,n\}^k = Z$. The configuration at time step $t$ is denoted $z(t)$. In what follows, we will look at the equivalent space of log-configurations, given by the bijection:
$$
Z \ni z = (z^{(1)}, \ldots, z^{(k)}) \mapsto (\lno z^{(1)}, \ldots, \lno z^{(k)} \equiv \lno z \in \{\lno 0,\lno 1,\ldots,\lno n\}^k),
$$
where $\lno a = \ln a$ for $a>0$ and $\lno a = -1$ for $a = 0$.

For $z_0 \in Z$, we will refer to the \emph{$d$-log-neighborhood} of $z_0$ as the set of points $\{z \in Z : |\lno z - \lno z_0| < d\}$.

Notice first that the notion of a box in the statement of Theorem~\ref{thm:lb} is closely related to the set of points in the $(\eps_0 \ln n)$-log-neighborhood of configuration $z_0$. It follows from the Theorem that any protocol for solving \textsc{Detect} within a polylogarithmic number of rounds $T$ with probability $1-o(1)$ must, in the case of $0 < \many X < n^{\eps_0}$, starting from $z_0$ at some time $t_0$, leave the $(\eps_0 \ln n)$-log-neighborhood of $z_0$ within $T$ rounds with probability $1-o(1)$.
We obtain the following corollary.

\begin{proposition}
Fix a universal protocol $P$ which solves the \textsc{Detect} problem with $\eps$-error in $T = O(\poly\log n)$ rounds with probability $1-o(1)$. Set $0 < \many X < n^{\eps_0}$, where $\eps_0 > 0$ is a constant which depends only on the definition of protocol $\many X$. Let $t_0$ be an arbitrarily chosen moment of time after at least $T$ rounds from the initialization of the protocol in any initial state. Then, within $T$ rounds after time $t_0$, there is a moment of time $t$ such that $z(t)$ is not in the $(\eps_0 \log n)$-neighborhood of $z(t_0)$, with probability $1-o(1)$.
\qed\end{proposition}

The above Proposition suggests that oscillatory or quasi-oscillatory behavior at low concentrations of state $X$ must be of length $\Omega(\log n)$. By contrast, the following Proposition shows that in the case $\many X = \Theta(n)$, the protocol remains tied to a constant-size log-neighborhood of its configuration space.

\begin{proposition}\label{pro:constbox}
Fix a universal protocol $P$ with set of states $K$ which solves the \textsc{Detect} problem with $\eps$-error in $T = O(\poly\log n)$ rounds with probability $1-o(1)$. Then, there exists a constant $\delta_0 > 0$, depending only on the design of protocol $P$, with the following property. Fix $\many X \in [cn, n/2]$, where $0 < c< 1/2$ is an arbitrarily chosen constant. Let $t$ be an arbitrarily chosen moment of time, after at least $T$ rounds from the initialization of the protocol at an adversarially chosen initial configuration $z(0)$, such that each coordinate $z^{(i)}(0)$ satisfies $z^{(i)}(0) = 0$ or $z^{(i)}(0) > 1/(2|K|)$, for all $i\in \{1,\ldots,|K|\}$. Then, with probability $1-e^{-n^{\Omega(1)}}$, $z(t)$ is in the $\delta_0$-neighborhood of $z(0)$.
\end{proposition}
The proof of the Proposition is deferred to Section~\ref{sec:constbox}.

Note that, in the regime of a constant-size log-neighborhood of configuration $z(0)$, the discrete dynamics of the protocol adheres closely to the continuous-time version of its dynamics in the limit $n \to +\infty$. (See Section~\ref{sec:upper}, and in particular Lemma~\ref{lem:taylor}, for a further discussion of this property). Since the latter is independent of $n$, any oscillatory behavior ``inherited'' from the continuous dynamic would have a period of $O(1)$ rounds. We leave as open the question whether some form of behavior of a protocol with polylogarithmic (i.e., or more broadly, non-constant and subpolynomial) periodicity for \textsc{Detect} can be designed in the regime of $\many X = \Theta(n)$ despite this obstacle. In particular, the authors believe that the existence of an input-controlled phase clock with a period of $\Theta(\log n)$ for any $\many X > 0$, and the absence of operation for $\many X = 0$, is unlikely in the class of discrete dynamical systems given by the rules of population protocols.

The remaining sections of the paper provide proofs of the Theorems from Sections~\ref{sec:protocols}, \ref{sec:impossibility}, and~\ref{sec:signal}.

\section{Analysis of Oscillator Dynamics \texorpdfstring{$P_o$}{Po}}\label{sec:upper}\label{Po}

This section is devoted to the proof of Theorem~\ref{thm:osc}.

\subsection{Preliminaries: Discrete vs. Continuous Dynamics}\label{sec:preli}

\paragraph{Notation.} For a configuration  of a population protocol, we write $z = (z^{(1)}, \ldots, z^{(k)}) $ to describe the number of agents in the $k$ states of the protocol, and likewise use vector $\uu = (\uu^{(1)}, \ldots, \uu^{(k)}) = z/n$ to describe their concentrations. The concentration of a state called $A$ which is the $i_A$-th state in vector $\uu$ is equivalently written as $a\equiv a(\uu) \equiv \uu^{(i_A)}$, depending on which notation is the easiest to use in a given transformation.

If vector $\uu$ represents the current configuration of the protocol and $\uu':= \uu'|\uu$ is the random variable describing the next configuration of the protocol after the execution of a single rule, we write $\Delta \uu := \uu'- \uu$.  We also use the notation $\Delta$ to functions of state $\uu$.

Next, we define the \emph{continuous dynamics} associated with the protocol by the following vector differential equation:
$$
\dot{\uu} \equiv \frac{d \uu}{d t} := n \E (\Delta \uu)
$$
and likewise, for each coordinate, $\dot{a} = n \E(\Delta a)$ (we use the dot-notation and $d/dt$ interchangeably for time differentials). This continuous description serves for the analysis only, and reflects the behavior of the protocol in the limit $n \to \infty$.\footnote{We note that some of our results rely on the stochasticity of the random scheduler model, and do not immediately generalize to the continuous case.}

\paragraph{Warmup: the RPS oscillator.} Our oscillatory dynamics may be seen as an extension of the rock-paper-scissors (RPS) protocol (see Related work). This is a protocol with three states $A_1$, $A_2$, $A_3$ and three rules:
\begin{equation*}
A_i \act A_{i-1} \mapsto A_i \quad \text{with probability } p,
\end{equation*}
where $p>0$ is an arbitrarily fixed constant, and the indices of states $A_i$ are always $1, 2,$ or $3$, and any other values should be treated as $\modd 3$ in the given range. For $i\in \{1,2,3\}$, the change of concentration of agents of state $A_i$ in the population in the given step can be expressed for the RPS protocol as:
\begin{equation}\label{eq:deltana}
\Delta a_i = \frac{1}{n} \cdot \Delta \many A_i =
\begin{cases}
+1/n, &\text{ with probability }pa_{i-1}a_{i},\\
-1/n, &\text{ with probability }pa_{i}a_{i+1},\\
0, &\text{ otherwise},
\end{cases}
\end{equation}
Thus, the corresponding continuous dynamics for RPS is given as:
$$
\dot a_i = n \E (\Delta a_i) = pa_{i-1}a_{i} - pa_{i}a_{i+1},
$$
for $i=1,2,3$. The orbit of motion for this dynamics in $\R^3$ is given by two constants of motion. First, $a_1 + a_2 + a_3 = 1$ by normalization. Secondly, for any starting configuration with a strictly positive number of agents in each of the three states, the following function $\phi$ of the configuration:
\begin{equation}\label{eq:potential}
 \phi= \ln (a_1 a_2 a_3)
\end{equation}
is easily verified to be constant over time $\dot \phi = 0$, hence $\phi = \ln (a_1 a_2 a_3) = \const < 0$ (or more simply, $a_1 a_2 a_3 = \const$). Thus, for the continuous dynamics, the initial product of concentrations completely determines its perpetual orbit, which is obtained by intersecting the appropriate curve $a_1 a_2 a_3 = \const$ with the plane $a_1 + a_2 + a_3 = 1$. As a matter of convention, the plane $a_1 + a_2 + a_3 = 1$ with conditions $a_i \geq 0$ is drawn as an equilateral triangle (we adopt this convention throughout the paper, for subsequent protocols). All of the orbits are concentric around the point $(1/3, 1/3, 1/3)$, which is in itself a point orbit maximizing the value of $\phi = -\ln 27$. The discrete dynamics follows a path of motion which typically resembles random-walk-type perturbations around the path of motion, until eventually, after $\tilde O(n^2)$ steps it crashes into one of the sides of the triangle.
Subsequently, if $a_i=0$, for some $i=1,2,3$, then no rule can make $a_i$ increase. (If $a_{i-1}>0$, in the next $O(\log n)$ steps, all remaining agents of $A_{i+1}$ will convert to $A_{i-1}$, and there will be only agents from $A_{i-1}$ left.) Thus, the protocol will terminate in a corner of the state space.

A further discussion of the RPS dynamics can be found in~\cite{HS98,ICALP15}.

\subsection{Proof Outline of Theorem~\ref{thm:osc}}

The rest of the section is devoted to the proof of Theorem~\ref{thm:osc}. We start by noting some basic properties in Subsection~\ref{sec:31}, then prove the properties of the protocol for the case of $X=0$ (Subsection~\ref{sec:32}, and finally analyze (the somewhat less involved) case of $X>0$ (Subsection~\ref{sec:33}). For the case of $X=0$, the proof is based on a repeated application of concentration inequalities for several potential functions (applicable in different portions of the $6$-dimensional phase space). In two specific regions, in the $O(1/\sqrt n)$-neighborhood of the center of the $(A_1, A_2, A_3)$-triangle and very close to its sides, we rely on stochastic noise to ``push'' the trajectory away from the center of the triangle, and also to push it onto one of its sides. Fortunately, each of these stages takes $O(\log n)$ parallel rounds, with strictly positive probability. Overall, the $O(\log n)$ parallel rounds bound for the case of $X=0$ is provided with constant probability; this translates into $O(\log n)$ parallel rounds in expectation, since subsequent executions of the process for $O(\log n)$ rounds have independently constant success probability, and the process has a geometrically decreasing tail over intervals of length $O(\log n)$.

\subsection{Properties of the Oscillator}\label{sec:31}

In the following, we define $s=a_1+a_2+a_3\in [0,1]$. Handling the case of $s<1$ allows us not only to take care of the fact that a fraction of the population may be taken up by rumor source $X$, but also allows for easier composition of $P_o$ with other protocols (sharing the same population). We set $p$ as a constant value independent of $n$, which is sufficiently small (e.g., $p = s^2/10^{12}$ is a valid choice; we make no efforts in the proofs to optimize constants, but the protocol appears in simulations to work well with much larger values of $p$).

We will occasionally omit an explanation of index $i$, which will then implicitly mean ``for all $i=1,2,3$''. We define $a_{\min} := \min_{i=1,2,3} a_i$ and $a_{\max} := \max_{i=1,2,3} a_i$.

From the definition of the protocol one obtains the distribution of changes of the sizes of states in a step:
$$
\Delta  a_i =
\begin{cases}
+1/n, & \text{with probability } \frac13 x(s-a_i) + pa_i^+a_{i-1}+ 2pa_i^{++}a_{i-1},\\
-1/n, & \text{with probability } \frac23 x a_i + pa_ia_{i+1}^+ + 2pa_{i+1}^{++}a_i,\\
0, & \text{otherwise.}
\end{cases}
$$
$$
\Delta  a_i^{++} =
\begin{cases}
+1/n, & \text{with probability } a_i(a_{i} - a_i^{++}),\\
-1/n, & \text{with probability } x a_i^{++} +  (s-a_i) a_i^{++},\\
0, & \text{otherwise.}
\end{cases}
$$

Taking the expectations of the above random variables, and recalling that $a_i=a_i^++a_i^{++}$, we obtain:
\begin{align}\
\dot{a}_i &=x(s/3 -a_i)+ pa_i^+a_{i-1}+ 2pa_i^{++}a_{i-1}-pa_ia_{i+1}^+-2pa_{i+1}^{++}a_i \nonumber \\
	  &=x(s/3 -a_i)+ pa_{i-1}(a_i+a_i^{++}) - pa_i(a_{i+1}+a_{i+1}^{++}) \label{eq:a_dot}\\
\dot{a}_i^{++} &= -x a_i^{++} +  a_i(a_{i} - a_i^{++}) - (s-a_i)a_i^{++} = -x a_i^{++} + a_i^2 - s a_i^{++}. \label{eq:app_dot}
\intertext{Moreover, we have by a simple transformation:}
\dot{\phi} &=\sum \frac{\dot{a}_i}{a_i}=\frac{x}{3}\bigg(\big(\sum\frac{s}{a_i}\big)-9\bigg) + p\bigg(\sum a_i^{++}\big(\frac{a_{i-1}}{a_i}-1\big)\bigg) \label{eq:derphi}
\end{align}

\subsection{Stopping in \texorpdfstring{$O(n \log n)$}{O(n log n)} Sequential Steps in the Absence of a Source}\label{sec:32}

Throughout this subsection we assume that $x=0$. We consider first the case where $a_i\ne 0$, for $i=1,2,3$ (noting that as soon as $a_i = 0$, we can easily predict the subsequent behavior of the oscillator, as was the case for the RPS dynamics).

The dynamics of $P_o$ is defined in such a way that that when $x=0$ and in the absence of the rules of the RPS oscillator, the value of $a_i^{++}$ would be close to $\frac{a_i^2}{s}$.
Consequently, we define $\kappa_i$, $i=1,2,3$ as the appropriate normalized corrective factor:
$$
\kappa_i = s \frac{a_i^{++}}{a_i}-a_i \quad\text{ thus }\quad a_i^{++}=\frac{a_i}{s}(a_i+\kappa_i).
$$
Note that as $0\le a_i^{++} \le a_i \le 1$, thus $-1 \leq \kappa_i \leq 1$. Next, we introduce the following definitions:
$$
\delta_i = a_i - a_{i-1}
$$
$$
\delta = \sqrt{\delta_1^2 + \delta_2^2 + \delta_3^2}
$$
$$
\kappa = \sqrt{\kappa_1^2 + \kappa_2^2 + \kappa_3^2}
$$
We also reuse potential $\phi$ from the original RPS oscillator. This time, it is no longer a constant of motion. By \eqref{eq:derphi} and the definition of $\kappa_i$, for $x=0$ we upper-bound $\dot{\phi}$ as:
\begin{align}
 \dot{\phi} &= \frac{p}{s}\bigg(\sum a_i(a_i+\kappa_i)(\frac{a_{i-1}}{a_i}-1)\bigg) \nonumber\\
            &= \frac{p}{s}\bigg(\sum (a_i+\kappa_i)(a_{i-1}-a_i)\bigg) \nonumber\\
            &= \frac{p}{s}\bigg(-\frac 12 \sum (a_i-a_{i-1})^2+\sum \kappa_i(a_{i-1}-a_i)\bigg) \nonumber\\
            &\leq \frac{p}{s}\bigg(-\frac 12 \delta^2+ \kappa \delta \bigg) \label{eq:phi_dot_xzero}
\end{align}
The above change $\dot{\phi}$ of the potential is indeed negative when $\kappa \approx 0$ (which is in accordance with our intention in designing the destabilizing rules for the oscillator).

The functions $\delta$, $\phi$ and $\kappa$ are intricately dependent on each other. In general, we will try to show that $\delta$ and $\phi$ increase over time, while $\kappa$ stays close to $0$. This requires that we first introduce a number of auxiliary potentials based on these two functions.


First, for $x=0$, we can rewrite~\eqref{eq:app_dot} as:
$$
\dot{a_i^{++}} = a_i^2 - s a_i^{++} = -a_i \kappa_i.
$$
Next, introducing the definition of $\kappa_i$ to~\eqref{eq:a_dot}, we obtain for $x=0$:
\begin{align*}
\dot{a_i} &= pa_{i-1}(a_i+a_i^{++}) - pa_i(a_{i+1}+a_{i+1}^{++}) = p a_i (a_{i-1} - a_{i+1}) + p(a_i^{++} a_{i-1} - a_{i+1}^{++} a_i) =\\
&= p a_i \delta_{i-1} + p \left(\frac1s a_i (a_i + \kappa_i) a_{i-1} - \frac1s a_{i+1} (a_{i+1} + \kappa_{i+1})  a_i\right) =\\
&= p a_i \delta_{i-1} + \frac ps a_i \left(\kappa_i a_{i-1} - \kappa_{i+1} a_{i+1} + (a_{i-1} a_i - a_{i+1}^2) \right).
\end{align*}
From the above, an upper bound on $|\dot{a_i}|$ follows directly using elementary transformations:
\begin{align*}
|\dot{a_i}| & \leq  p a_i |\delta_{i-1}| + \frac ps a_i \left(|\kappa_i| a_{i-1} + |\kappa_{i+1}| a_{i+1} + |a_{i-1} a_i - a_{i+1}^2| \right) \leq\\
& \leq p a_i\delta + \frac ps a_i\left(\kappa(a_{i-1} + a_{i+1}) + |a_{i-1} a_i - a_{i+1}^2|)\right)=\\
& = p a_i\delta + p a_i \kappa \frac{a_{i-1} + a_{i+1}}{s} + \frac ps a_i |(a_{i-1}-a_{i+1}) a_i - (a_{i+1}-a_i)a_{i+1} | \leq\\
& \leq p a_i (\delta + \kappa) + \frac ps a_i\left(|a_{i-1}-a_{i+1}| a_i + |a_{i+1}-a_i|a_{i+1}\right) \leq\\
& \leq p a_i(\delta + \kappa) + p a_i\delta\frac {a_i + a_{i+1}}{s}\\
& \leq p a_i(2\delta + \kappa).\\
\end{align*}

We are now ready to estimate $\dot \kappa_i$ for $x=0$, using the definition of $\kappa_i$ and the previously obtained formula for $\dot a_i^{++}$:
\begin{alignat*}{3}
\dot \kappa_i &= s \left(\frac{\dot a_i^{++}}{a_i} - \frac{a_i^{++}}{a_i^2}\dot a_i \right)-\dot a_i &&= -s  \kappa_i - \left(s\frac{a_i^{++}}{a_i^2}+1\right) \dot{a_i} \\
&= -s  \kappa_i - \left(s\frac{a_i^{++}}{a_i}-a_i\right) \frac{\dot{a_i}}{a_i} -2 \dot a_i &&= -s  \kappa_i - \frac{\dot{a_i}}{a_i} \kappa_i -2 \dot a_i .
\end{alignat*}
Next from the bound on $|\dot{a_i}|$:
\begin{align*}
\dot {\kappa^2_i} &= 2\kappa_i \dot{\kappa}_i = 2\big(-s \kappa_i^2 - \frac{\dot{a_i}}{a_i}\kappa_i^2 - 2\dot{a_i}\kappa_i\big)\\
&\leq 2\big(-s \kappa_i^2 + \frac{|\dot{a_i}|}{a_i}\kappa_i^2 + 2|\dot{a_i}||\kappa_i|\big)\\
&\leq 2\big(-s \kappa_i^2 + p(2\delta + \kappa)(\kappa_i^2 + 2a_i|\kappa_i|)\big) \\
&\leq 2\big(-s \kappa_i^2 + p(2\delta + \kappa)(\kappa + 2\kappa)\big) \\
&= -2s \kappa_i^2 + p(12\delta\kappa + 6\kappa^2).\\
\end{align*}

Next:
\begin{align}
\dot \kappa &= \frac{1}{2\kappa}\sum \dot {\kappa^2_i} \leq \frac{1}{2\kappa}\left(-2s \sum \kappa_i^2 + 3p (12\delta\kappa + 6\kappa^2)\right) \leq \nonumber\\
& \leq \frac{1}{2\kappa}\left(-2s \kappa^2 + p (36\delta\kappa + 18\kappa^2)\right) = \nonumber\\
&= (-s + 9p) \kappa + 18 p \delta \leq \nonumber\\
&\leq -\frac{s}{2} \kappa + 18 p \delta, \label{eq:kappa_dot}
\end{align}
where in the final transformation we took into account that $p \leq s/18$.

Now, we define the potential $\eta$ for any configuration with all $a_i>0$ as:
$$
\eta = \left(\ln\frac{s^3}{27} - \phi\right)^{1/2}= \left(-\sum_{i=1}^3\ln \frac {a_i}{s/3}\right)^{1/2}.
$$

We remark that $\eta$ is always well-defined when $a_{\min} > 0$, and that $\eta \geq 0$.

\paragraph{Overview of the proof.} The proof for the case of $X=0$ proceeds by following the trajectory of the discrete dynamics of $P_o$, divided into a number of stages. We define a series of time steps $t_0, t_1, \ldots, t_7$ by conditions on the configuration met at time $t_i$, and show that subject to these conditions holding, we have $t_{j+1}\leq t_j + O(n\log n)$ (we recall that here time is measured in steps), with at least constant probability. Overall, it follows that the configuration at time $t_7$, which corresponds to having reached a corner state, is reached from $t_0$, which is any initial configuration with $X=0$, in $O(n \log n)$ time steps, with constant probability.

The intermediate time steps may be schematically described as follows (see Fig.~\ref{fig:triangle1}). For configurations which start close to the center of the triangle ($\delta \leq s/12$), we define a pair of potentials $\psi^{(1)}$, $\psi^{(2)}$, based on a linear combination of modified versions of $\eta$ and $\kappa$. The dynamics will eventually escape from the area $\delta \leq s/12$; however, first it may potentially reach a very small area of radius $O(1/\sqrt n)$ around the center of the triangle with $\kappa \approx 0$ (Lemma~\ref{lem:pot_negative}, time $t_1$, reached in $O(n \log n)$ steps by a multiplicative drift analysis on potential $\psi^{(2)} < 0$), pass through the vicinity of center of the triangle, escaping it with $\kappa \approx 0$  (Lemma~\ref{lem:pot_aroundzero}, time $t_2$, reached in $O(n \log n)$ steps with constant probability by a protocol-specific analysis of the scheduler noise, which with constant probability increases $\eta$ without increasing $\kappa$ too much), and eventually escapes completely to the area of $\delta > s/12$ (Lemma~\ref{lem:pot_positive}, exponentially increasing value of potential $\psi^{(1)} > 0$).

\begin{figure}
  \centering
  \vspace{-1cm}
  \includegraphics[width=0.7\textwidth]{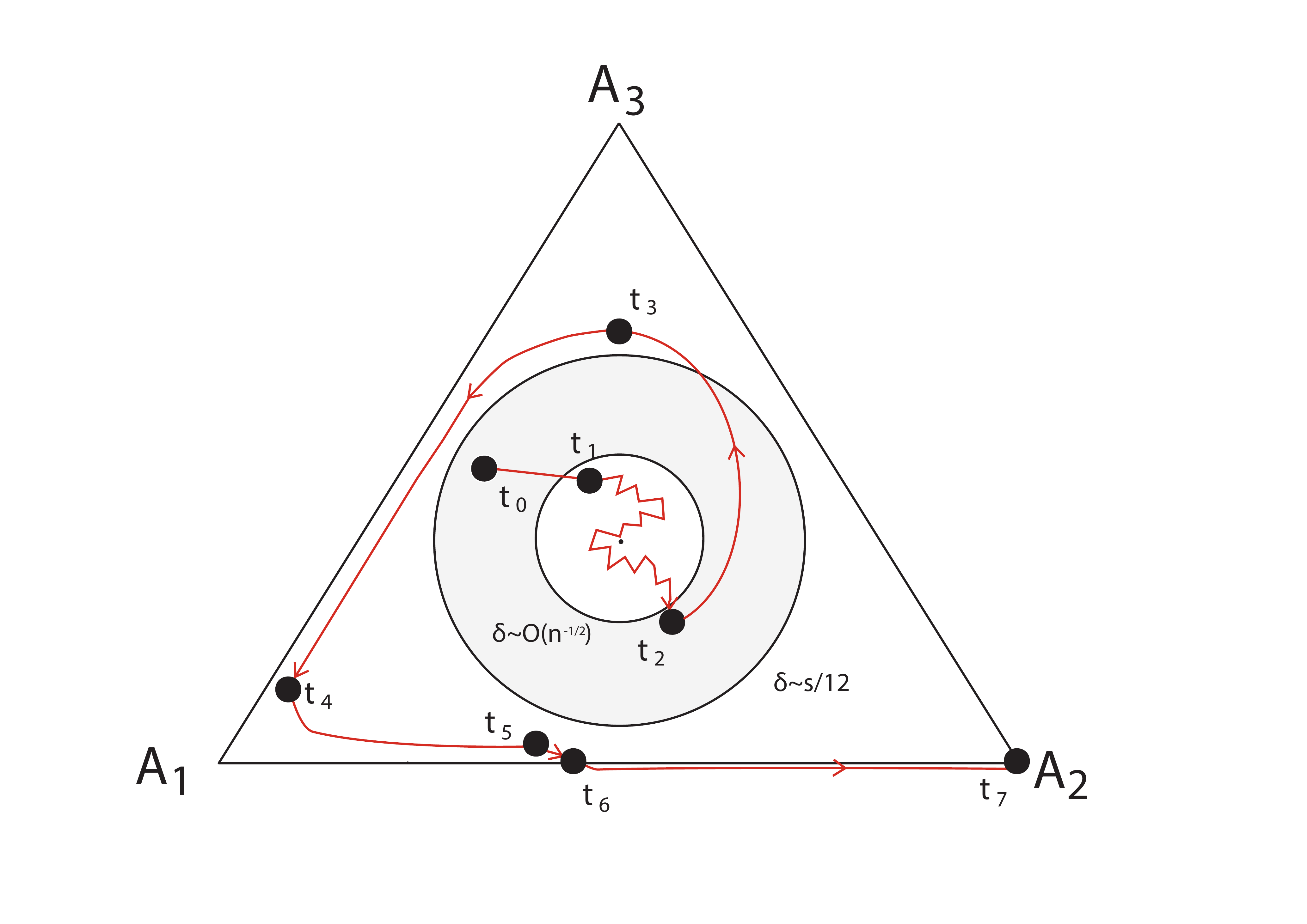}\vspace{-1cm}
  \caption{Schematic illustration of order of phases in the proof of stabilization of protocol $P_o$ for $X=0$.}\label{fig:triangle1}
\end{figure}

In the area of $\delta > s/12$, we define a new potential $\psi$ based on $\phi$ and $\kappa$. This increases (Lemma~\ref{lem:pot_psi_first}, additive drift analysis on $\psi$ with bounded variance) until a configuration at time $t_4$ with a constant number of agents of some species $A_i$ is reached. This configuration then evolves towards a configuration at time $t_5$ at which some species has $O(1)$ agents, and additionally its predator species is a constant part of the population~(Lemma~\ref{lem:pot_psi_second}, direct analysis of the process combined with analysis of potential $\psi$ and a geometric drift argument). Then, the species with $O(1)$ agents is eliminated in $O(n)$ steps with constant probability ($t_6$, Lemma~\ref{lem:eliminationoffirstspecies}), and finally one more species is eliminated in another $O(n \log n)$ steps (at time $t_7$, Lemma~\ref{lem:eliminationofsecondspecies}, straightforward analysis of the dynamics). At this point, the dynamics has reached a corner.

Throughout the proof, we make sure to define boundary conditions on the analyzed cases to make sure that the process does not fall back to a previously considered case with probability $1 - o(1)$.

\paragraph{Phase with $\delta \leq s/12$.} We then have $a_i \in [3s/12, 5s/12]$ and $\frac {a_i}{s/3} \in [3/4, 5/4]$, for $i=1,2,3$. In this range, we have:
$$ 
\frac{1}{3} \left(\frac {a_i}{s/3}-1\right)^2 < \left(\frac {a_i}{s/3}-1\right) - \ln \frac {a_i}{s/3} < \frac{3}{4} \left(\frac {a_i}{s/3}-1\right)^2.
$$
Summing the above inequalities for $i=1,2,3$ and noting that $\sum_{i=1}^3(\frac {a_i}{s/3}-1) = 0$, we obtain:
$$
\frac{1}{3} \sum_{i=1}^3\left(\frac {a_i}{s/3}-1\right)^2 < \eta^2 < \frac{3}{4} \sum_{i=1}^3 \left(\frac {a_i}{s/3}-1\right)^2.
$$
Next, we have:
$$
\sum_{i=1}^3\left(\frac {a_i}{s/3}-1\right)^2 = \left(\frac 3 s\right)^2 \sum_{i=1}^3 \left(a_i - s/3\right)^2 = \frac{3\delta^2}{s^2}.
$$
Combining the two above expressions gives the sought bound between $\eta$ and $\delta$ as:
$$
\frac{\delta}{s} < \eta < \frac{3}{2}\frac{\delta}{s}
$$
and equivalently
$$
\delta \in (\tfrac{2}{3} s\eta, s\eta).
$$

We have directly from \eqref{eq:phi_dot_xzero} and from the relations between $\eta$ and $\delta$:
\begin{equation}
\dot \eta = - \frac{\dot \phi}{2\eta} \geq \frac{p}{2s\eta}\bigg(\frac 12 \delta^2 - \kappa \delta \bigg) =
\frac{p}{4s\eta} \delta^2 - \frac{p}{2s\eta} \kappa \delta \geq
\frac{ps}{9} \eta - \frac{p}{2}\kappa,\label{eq:eta_dot}
\end{equation}
and from \eqref{eq:kappa_dot}:
\begin{equation}
\dot \kappa \leq -\frac{s}{2} \kappa + 18 p \delta \leq -\frac{s}{2} \kappa + 18ps\, \eta.\label{eq:kappa_dot_bis}
\end{equation}
Moving to the discrete-time model, it is advantageous to eliminate the discontinuity of partial derivatives of $\eta$ and $\kappa$ at points with $\eta=0$ and $\kappa=0$ respectively, which is a side-effect of the applied square root transformation in the respective definitions of $\eta$ and $\kappa$. We define the auxiliary functions $\eta^*$ and $\kappa^*$ by adding an appropriate corrective factor:
\begin{align*}
\eta^* &= \sqrt {\eta^2 + \frac1n}\\
\kappa^* &= \sqrt {\kappa^2 + \frac1n}\\
\end{align*}
and derive accordingly from~\eqref{eq:eta_dot} and \eqref{eq:kappa_dot_bis}:
\begin{align}
\dot \eta^* &= \frac{\eta}{\eta^*}\dot \eta \geq \frac{ps}{9} \left(\eta - \frac{1}{\sqrt n}\right) - \frac{p}{2}\kappa \geq
    \frac{ps}{9} \eta^* - \frac{p}{2}\kappa^* - \frac{2ps}{9 \sqrt n} \label{eq:etastar_dot}\\
\dot \kappa^* &= \frac{\kappa}{\kappa^*}\dot \kappa  \leq -\frac{s}{2} \left(\kappa - \frac{1}{\sqrt n}\right) + 18ps\, \eta \leq
    -\frac{s}{2} \kappa^* + 18ps\, \eta^* +  \frac{s}{\sqrt n}. \label{eq:kappastar_dot}
\end{align}


Let $\uu$ be the $5$-dimensional vector representing the current configuration of the system: $\uu := (a_1^+, a_1^{++}, a_2^+,\allowbreak a_2^{++},a_3^+) \equiv (\uu^{(1)}, \ldots, \uu^{(5)})$; note that the last element $a_3^{++}$ is determined as $a_3^{++} = s - \sum_{i=1}^5 \uu^{(i)}$.\footnote{In principle it is also correct to represent $\uu$ as a vector of dimension $6$, i.e., including $a_3^{++}$ in $\uu$ as a free dimension. However, such a representation would lead to second-order partial derivatives $\frac{\partial^2}{\partial \uu^{(i)} \partial \uu^{(j)}}\eta^*(\uu)$ which are too large for our purposes.} The following lemma is obtained by a folklore application of Taylor's theorem.
\begin{lemma}\label{lem:taylor}
Let $f : \R^5 \to \R$ be a $C^2$ function in a sufficiently large neighborhood of $\uu$, with $\min_{1\leq i\leq 5} \uu^{(i)} \geq 2/n$. Then, $|\E \Delta f(\uu) - \frac{\dot f}{n}| \leq \frac{2}{n^2} \max_{\|\uu^* - \uu\|_{\infty} \leq 1/n} D_f(\uu^*)$, where $D_f(\uu^*) := \max_{1\leq i, j \leq 5} \left|\frac{\partial f^2(\uu^*)}{\partial \uu^{(i)} \partial \uu^{(j)}}\right|$.
\end{lemma}
\begin{proof}
Let $\uu'$ be the random variable representing the configuration of the system after its next transition from configuration $\uu$.
Observe that in every non-idle step of execution of the protocol, exactly one agent changes its state, so
$\|\uu'-\uu\|_{\infty} \leq 1/n$.

Applying Taylor approximation we have:
\begin{align}
\E(\Delta f) &=\E(f(\uu')|\uu)-f(\uu) = \sum_{\uu'} (f(\uu')-f(\uu))\Pr(\uu'|\uu) = \sum_{\uu'} \left(\nabla f(\uu) \cdot (\uu'-\uu) + R_2(\uu,\uu')\right)\Pr(\uu'|\uu) =\nonumber\\
&= \nabla f(\uu) \cdot \sum_{\uu'} (\uu'-\uu) \Pr(\uu'|\uu) + R_2(\uu) = \nabla f(\uu) \cdot \frac1n(\dot \uu^{(1)}, \ldots, \dot \uu^{(5)})^T + R_2(\uu) = \frac{\dot f}{n} + R_2(\uu),\label{eq:taylor_approx}
\end{align}
where $\nabla f(\uu)$ is the gradient of $f$ at $\uu$, $R_2(\uu,\uu') \in \R$ denotes the second-order Taylor remainder for function $f$ expanded at point $\uu$ along the vector towards point $\uu'$, and $R_2(\uu) \in \R $ is subsequently an appropriately chosen value, satisfying:
$$
|R_2(\uu)| \leq \frac{1}{n^2} \max_{\|\uu^* - \uu\|_{\infty} \leq 1/n} D_f(\uu^*).
$$
\end{proof}

The following lemma is obtained directly by computing and bounding all second order partial derivatives of functions $\eta^*$ and $\kappa^*$ with respect to variables $(\uu^{(1)}, \ldots, \uu^{(5)})$.
\begin{lemma}
There exists a constant $c_1>1$ depending only on $s,p$, such that, for any configuration $\uu$ with $\delta(\uu) \leq s/12$:
\begin{itemize}
\item $\max_{\|\uu^* - \uu\|_{\infty} \leq 1/n} D_{\eta^*}(\uu^*) < c_1 \sqrt n$,
\item $\max_{\|\uu^* - \uu\|_{\infty} \leq 1/n} D_{\kappa^*}(\uu^*) < c_1 \sqrt n$.
\vspace{-5mm}\end{itemize}
\qed
\end{lemma}
In view of the above lemmas, we obtain from~\eqref{eq:eta_dot} and~\eqref{eq:kappa_dot_bis}, for an appropriately chosen constant $c_2 = 2c_1 + s$:
$$
\begin{cases}
\E\Delta \eta^* \geq \frac{1}{n}\left(\frac{ps}{18} \eta^* - \frac{p}{2}\kappa^* - \frac{c_2}{\sqrt n} \right), & \text{when $\delta \leq s/12$,}\\
\E\Delta \kappa^* \leq \frac{1}{n}\left(-\frac{s}{3} \kappa^* + 18ps\,\eta^* + \frac{c_2}{\sqrt n} \right), & \text{when $\delta \leq s/12$.}\\
\end{cases}
$$
For $j=1,2$, we now define two linear combinations of functions $\eta^*$ and $\kappa^*$:
$$
\psi^{(j)} = \eta^* - \frac{3jp}{s} \kappa^*.
$$
When $\delta \leq s/12$, we have:
\begin{align*}
\E\Delta \psi^{(j)} &\geq \frac{1}{n}\left( \frac{ps}{18} \eta^* - \frac{p}{2}\kappa^* - \frac{c_2}{\sqrt n} + jp \kappa^* - 54jp^2 \eta^* - \frac{3jp}{s} \frac{c_2}{\sqrt n} \right) \geq \frac{1}{n}\left( \frac{ps}{24} \eta^* + \frac{jp}{2}\kappa^* - \frac{2c_2}{\sqrt n} \right) \geq\\
&\geq  \frac{ps}{24 n} \left(\eta^* + \frac{3jp}{s} \kappa^* - \frac{48c_2}{ps\sqrt n} \right) \geq \frac{ps}{24 n} \left(|\psi^{(j)}| - \frac{c_3}{\sqrt n}\right),
\end{align*}
where we denoted $c_3 := \frac{48c_2}{ps}$ and used the fact that $p < \frac{s}{72\cdot 54\cdot 2}$.

We subsequently perform an analysis of $\psi^{(j)}_t = \psi^{(j)}(\uu_t)$, $j=1,2$, treating them as stochastic processes. We remark that $\psi^{(2)}_{t} \leq \psi^{(1)}_{t}$, since $\psi^{(1)}_{t} - \psi^{(2)}_{t} = \frac{3p}{s}\kappa^* \geq 0$.

\begin{lemma}\label{lem:pot_negative}
Let $\uu_{t_0}$ be an arbitrary starting configuration of the system. Then, with constant probability, for some $t_1=t_0+O(n \log n)$, a configuration $\uu_{t_1}$ is reached such that $\psi^{(1)}_{t_1} \geq \psi^{(2)}_{t_1} \geq -\frac{2c_3}{\sqrt n}$.
\end{lemma}
\begin{proof}

W.l.o.g.\ assume $t_0 = 0$. We subsequently only analyze process $\psi^{(2)}_{t}$.
Let $t_1$ be the first time step such that $\psi^{(2)}_{t_1} > -\frac{2 c_3}{\sqrt n}$. If $t_1\neq 0$, then $\psi^{(2)}_{0} < 0$. Note that then $\psi^{(2)}_{t} < 0$ for all $t \leq t_1$, from which it follows by a straightforward calculation from the definition of $\psi$, $\kappa$, and $\eta$, that $\delta_t < \frac{s}{12}$ for all $t \leq t_1$.

We now define the filtered stochastic process $\psi^{*(2)}_t$ as $\psi^{*(2)}_t := |\psi^{(2)}_t|$ for $t< t_1$, and put $\Delta \psi^{*(2)}_{t} := 0$ for $t \geq t_1$. For all $t\geq 0$, we then have:
$$
\E(\Delta \psi^{*(2)}_{t} | \psi^{*(2)}_{t} \neq 0) \leq \frac{ps}{48n} \psi^{*(2)}_{t}.
$$
Since $0\leq \psi^{*(2)}_{t} < 9$ for all time steps, a direct application of multiplicative drift analysis (cf.~\cite{muldrift}) gives:
$$
\E t_1 \leq \frac{48n}{ps} \left(1 + \ln \frac{9 \sqrt n}{2c_3}\right),
$$
and the claim follows by Markov's inequality.
\end{proof}

\begin{lemma}\label{lem:pot_aroundzero}
Let $\uu_{t_1}$ be an arbitrary starting configuration of the system such that $\psi^{(j)}_{t_1} \in [-\frac{2c_3}{\sqrt n}, \frac{4c_3}{\sqrt n}]$, for $j=1,2$. Then, with constant probability, for some $t_2=t_1+O(n)$, a configuration $\uu_{t_2}$ is reached such that $\psi^{(1)}_{t_2} \geq \frac{4c_3}{\sqrt n}$.
\end{lemma}
\begin{proof}
W.l.o.g. assume that $t_1 = 0$ and suppose that initially $\psi^{(2)}_{0} \leq \psi^{(1)}_{0} < \frac{4c_3}{\sqrt n}$ (i.e., that $t_2 \neq t_1$). Then, from the lower and upper bounds on $\psi^{(1)}_{0}$ and $\psi^{(2)}_{0}$  we obtain the following bounds on $\kappa_0$ and $\delta_0$:
\begin{align*}
\frac{3p}{s} \kappa^*_{0} = \psi^{(1)}_{0} - \psi^{(2)}_{0} \leq \frac{2c_3}{\sqrt n} + \frac{4c_3}{\sqrt n} &\implies \kappa_{0} \leq \frac{2c_3 s}{p \sqrt n},\\
\eta^*_{0} = 2 \psi^{(1)}_{0} - \psi^{(2)}_{0} \leq 2\frac{4c_3}{\sqrt n} + \frac{2c_3}{\sqrt n} & \implies \eta_{0} \leq \frac{10c_3}{\sqrt n} \implies \delta_{0} < \frac{10c_3 s}{\sqrt n}.
\end{align*}
It follows that, for $i=1,2,3$, $a_{i,0} \in [\frac{s}{3} - \frac{10c_3 s}{\sqrt n}, \frac{s}{3} + \frac{10c_3 s}{\sqrt n}]$ and $a_{i,0}^{++} = \frac{a_{i,0}}{s}(a_{i,0} + \kappa_{i,0}) \in [(\frac{1}{3} - \frac{10c_3}{\sqrt n})(\frac{s}{3} - \frac{10c_3 s}{\sqrt n} - \frac{2c_3 s}{p \sqrt n}), (\frac{1}{3} + \frac{10c_3}{\sqrt n})(\frac{s}{3} + \frac{10c_3 s}{\sqrt n} + \frac{2c_3 s}{p \sqrt n})]$. For the sake of clarity of notation, we will simply write $a_{i,0} = s/3 \pm O(1/\sqrt n)$ and $a_{i,0}^{++} = s/9 \pm O(1/\sqrt n)$, hence also $a_{i,0}^{+} = 2s/9 \pm O(1/\sqrt n)$.

We will consider now the sequence of exactly $n$ transitions of the protocol, between time steps $t=0,1,\ldots, n$.

For all $t$ we have $\E \Delta \psi^{(2)}_t \geq -\frac{c_3 ps}{24 n^{3/2}}$. Consider the Doob submartingale $Y_t = \sum_{\tau=0}^{t-1} X_t$ with increments $(X_t)$ given as:
$$
X_{t} = \begin{cases}
\Delta\psi^{(2)}_t + \frac{c_3 ps}{24 n^{3/2}}, &\text{if $Y_t > -\frac{c_3}{\sqrt n}$}\\
0, & \text{otherwise},
\end{cases}
$$
Noting that $|X_t| \leq \frac{9}{n}$, an application of the Azuma inequality for submartingales to $(Y_n)$ gives: $\Pr [Y_n \leq - \frac{c_3}{\sqrt n}] \leq \exp{[-c_3^2/162]} $ (cf. e.g.~\cite{concen.pdf}[Thm.\ 16]). From here it follows directly that:
$$
\Pr\left[\psi^{(2)}_n > -\frac{c_3}{\sqrt n} + \psi^{(2)}_0 - n\frac{c_3 ps}{24 n^{3/2}}\right] \geq 1 - \exp{[-c_3^2/162]} > 1/2.
$$
Noting that $\psi^{(2)}_{0} \geq -\frac{2c_3}{\sqrt n}$, we have:
\begin{equation}\label{eq:psi_lb}
\Pr\left[\psi^{(2)}_n \geq -\frac{2c_3}{\sqrt n}\right] > 1/2.
\end{equation}
We now describe the execution of transitions in the protocol for times $t=0,1,\ldots,n-1$ through the following coupling. First, we select the sequence of pairs of agents chosen by the scheduler. Let $V_2^+$ (respectively, $V_1^+$) denote the subsets of the set of $n$ agents, having initial states $A_2^+$ (resp., $A_1^+$) at time $0$, respectively, which are involved in exactly one transition in the considered time interval, acting in it as the initiator (resp., receiver). Let $S\subseteq\{0,1,\ldots,n-1\}$ denotes the subset of time steps at which the scheduler activates a transition involving an element of $V_2^+$ as the initiator and an element of $V_1^+$ as the receiver. The execution of the protocol is now given by:
\begin{itemize}
\item Phase $P_A$: Selecting the sequence of pairs of elements activated by the scheduler in time steps $(0,1,\ldots,n-1)$. This also defines set $S$. Executing the rules of the protocol in their usual order for time steps from set $\{0,1,\ldots,n-1\}\setminus S$.
\item Phase $P_B$: Executing the rules of the protocol for time steps from set $S$.
\end{itemize}
Observe that since elements of pairs activated in time steps from $S$ are activated only once throughout the $n$ steps of the protocol, the above probabilistic coupling does not affect the distribution of outcomes.

Directly from~\eqref{eq:psi_lb}, we obtain through a standard bound on conditional probabilities that at least a constant fraction of choices made in phase $P_A$ leads to an outcome ``$\psi^{(2)}_n \geq -\frac{2c_3}{\sqrt n}$'' with at least constant probability during phase $P_B$:
\begin{equation}\label{eq:bound_a}
\Pr\left[P_A : \Pr \left[\psi^{(2)}_n \geq -\frac{2c_3}{\sqrt n} \big| P_A\right] > 1/4\right] \geq 1/3.
\end{equation}
We now remark on the size of set $S$. The distribution of $|S|$ depends only on $a_{2,0}^+$, $a_{1,0}^+$, and the choices made by the random scheduler. We recall that $a_{2,0}^+ = 2s/9 \pm O(1/\sqrt n)$. Since the expected number of isolated edges in a random multigraph on $n$ nodes (representing the set of agents) and $n$ edges (representing the set of time steps) is $(1\pm o(1))e^{-4}n$, the number of such edges having the first endpoint in an agent in state $A_2^+$ and the second endpoint in an agent in state $A_1^+$ is $(1\pm o(1))\frac{4e^{-4}s^2}{81}n$. A straightforward concentration analysis (using, e.g., the asymptotic correspondence between $G(n,m)$ and $G(n,p)$ random graph models and an application of Azuma's inequality for functions of independent random variables) shows that the bound $|S| = (1\pm o(1))\frac{4e^{-4}s^2}{81}n$ holds with very high probability. In particular, we have:
\begin{equation}\label{eq:bound_s}
\Pr [|S| > c_4 n] = 1 - e^{\Omega(-n)},
\end{equation}
for some choice of constant $c_4$ which depends only on $s$.

Relations~\eqref{eq:bound_a} and \eqref{eq:bound_s} provide all the necessary information about phase $P_A$ that we need. Subsequently, we will only analyse phase $P_B$, conditioning on a fixed execution of phase $P_A$ such that the following event $F_A$ holds:
\begin{equation}\label{eq:f_a}
\Pr \left[\psi^{(2)}_n \geq -\frac{2c_3}{\sqrt n} \big| P_A \right] > 1/4 \quad \wedge \quad |S| > c_4 n.
\end{equation}
We remark that, by a union bound over~\eqref{eq:bound_a} and \eqref{eq:bound_s}, $\Pr[F_A] \geq 1/3 - e^{\Omega(-n)} > 1/4$.

In the remainder of our proof, our objective will be to show that:
$$
\Pr \left[\psi^{(1)}_n \geq \frac{4c_3}{\sqrt n} \big| P_A \right] > c_5,
$$
for some constant $c_5 > 0$ depending only on $s,p$, for any choice of $P_A$ for which event $F_A$ holds. When this is shown, the claim of the lemma will follow directly, with a probability value given as at least $c_5 \Pr[F_A] > c_5 / 4$ by the law of total probability.

We now proceed to analyze the random choices made during phase $P_B$. Each of the considered $|S|$ interactions involves a pair of agents of the form $(A_2^+, A_1^+)$, and describes the following transition:
$$
(A_2^+, A_1^+) \to \begin{cases}
(A_2^+, A_2^+), &\text{with probability $p$},\\
(A_2^+, A_2^+), &\text{with probability $1-p$},
\end{cases}
$$
independently at random for each transition. The only state changes observed during this phase are from $A_1^+$ to $A_2^+$, and we denote by $B$ the number of such state changes. The value of random variable $B$ completely describes the outcome of phase $P_B$.

We have $\E B = p|S|$, and by a standard additive Chernoff bound:
\begin{equation}
\label{eq:chernoff_b}
\Pr\left[|B - p|S|| \leq 2 \sqrt n \big| P_A \right] \geq 1 - 2 e^{-4} > 7/8.
\end{equation}
Let $\mathcal B \subseteq [p|S| - 2 \sqrt n, p|S| + 2 \sqrt n]$ be the subset of the considered interval containing values of $B$ such that
$\left(\psi^{(1)}_n \big| P_A, B\in \mathcal B\right) \geq \frac{4c_3}{\sqrt n}$. If $\Pr[B \in \mathcal B | P_A] \geq 1/8$, then the claim follows directly.

Otherwise, it follows from~\eqref{eq:f_a} and~\eqref{eq:chernoff_b} that there must exist a value $b \in [p|S| - 2 \sqrt n, p|S| + 2 \sqrt n] \setminus \mathcal B$, such that:
$$
\left(\psi^{(2)}_n \big| P_A, B=b \right) \geq -\frac{2c_3}{\sqrt n}.
$$
Given that:
$$
\left(\psi^{(1)}_n \big| P_A, B=b \right) \leq \frac{4c_3}{\sqrt n}.
$$
and recalling that $\psi^{(2)}_n \leq \psi^{(1)}_n$, we obtain the following bound on $\eta_n$:
\begin{equation}\label{eq:eta_bound1}
\left(\eta^*_{n} \big| P_A, B=b\right) = \left(2 \psi^{(1)}_{n} - \psi^{(2)}_{n} \big| P_A, B=b\right) \leq 2\frac{4c_3}{\sqrt n} + \frac{2c_3}{\sqrt n} =  \frac{10c_3}{\sqrt n}.
\end{equation}

We now consider lower bounds on the value of $\psi^{(2)}_n$, conditioned on $P_A, B=b^+$ (respectively, $P_A, B=b^-$), where $b^+$ (resp., $b^-$) is a value arbitrarily fixed in the range $b^+\in [b + \frac{20c_3s}{\sqrt n}, b + \frac{21c_3s}{\sqrt n}]$ (resp., $b^-\in [b - \frac{21c_3s}{\sqrt n}, b - \frac{20c_3s}{\sqrt n}]$). The executions of the protocol with $B=b^+$ and $B=b^-$ differ with respect to the execution with $B=b$ in the number of executed transitions from $a_1^+$ to $a_2^+$ by at least $\frac{20c_3s}{\sqrt n}$. Recalling that $\delta_2 = a_2 - a_1$, it follows that for some $b' \in \{b^+, b^-\}$ we have after $n$ steps:
$$
\left(\delta_{n} \big| P_A, B=b' \right) \geq \left(|\delta_{2,n}| \big| P_A, B=b' \right) \geq \frac{20c_3s}{\sqrt n}.
$$
Subsequently, we will assume that $b'=b^+$; the case of $b' = b^-$ is handled analogously. From the relation $\eta > \delta/s$ and~\eqref{eq:eta_bound1} we have:
\begin{equation}\label{eq:eta_bound2}
\left(\eta^*_{n} \big| P_A, B=b^+\right) \geq  \frac{20c_3}{\sqrt n} \geq \left(\eta^*_{n} \big| P_A, B=b\right) + \frac{10c_3}{\sqrt n}.
\end{equation}

When comparing the value of $\kappa^*_n$ in the two cases, $B=b^+$ and $B=b$, it is convenient to consider $\kappa^*$ as the length of the vector $(\kappa_1, \kappa_2, \kappa_3, 1/\sqrt n)$ in Euclidean space. For each of the coordinates $\kappa_i$, $i=1,2,3$, we have:
$$
\left|\left(\kappa_{i,n} \big| P_A, B=b^+\right) - \left(\kappa_{i,n} \big| P_A, B=b\right)\right | < \frac{40c_3}{\sqrt n},
$$
hence:
\begin{equation}\label{eq:kappa_bound}
\left(\kappa^*_{n} \big| P_A, B=b^+\right) < \left(\kappa^*_{n} \big| P_A, B=b\right) + \frac{120c_3}{\sqrt n}.
\end{equation}
Introducing \eqref{eq:eta_bound2} and \eqref{eq:kappa_bound} into the definition of $\psi^{(1)}_n$, we obtain directly:
$$
\left(\psi^{(1)}_{n} \big| P_A, B=b^+\right) > \left(\psi^{(1)}_{n} \big| P_A, B=b\right) + \frac{10 c_3}{\sqrt n} - \frac{3p}{s} \frac{120c_3}{\sqrt n} \geq - \frac{2c_3}{\sqrt n} + \frac{10 c_3}{n} - \frac{3p}{s} \frac{120c_3}{\sqrt n} > \frac{4c_3}{\sqrt n},
$$
where we again used the fact that $p$ is a sufficiently small constant w.r.t. $s$. We thus obtain:
$$
\left(\psi^{(1)}_{n} \big| P_A, B \in [b + \tfrac{20c_3s}{\sqrt n}, b + \tfrac{21c_3s}{\sqrt n}]\right) > \frac{4c_3}{\sqrt n},
$$
where by the definition of random variable $B$ as a sum of i.i.d.\ binary random variables and the choice of value $b$ in the direct vicinity of the expectation of $B$, the event $B \in [b + \frac{20c_3s}{\sqrt n}, b + \frac{21c_3s}{\sqrt n}]$ holds with constant probability. The case of $b'=b^{-}$ is handled analogously.
\end{proof}

\begin{lemma}\label{lem:pot_positive}
Let $\uu_{t_2}$ be an arbitrary starting configuration of the system such that $\max\{\psi^{(1)}_{t_2}, \psi^{(2)}_{t_2}\} = \psi^{(1)}_{t_2} \geq \frac{4c_3}{\sqrt n}$. Then, with constant probability, for some $t_3=t_2+O(n \log n)$, a configuration $\uu_{t_3}$ is reached such that $\delta_{t_3} > s/12$.
\end{lemma}
\begin{proof}
We subsequently consider only the process $\psi^{(1)}_{t}$. We start by showing the following claim.

\emph{Claim.} Suppose $\psi^{(1)}_{0} = A \geq \frac{4c_3}{\sqrt n}$. Then, with probability at least $1-\exp{[-A^2 psn / 46656]}$, for some time step $t \leq \frac{72n}{ps}$ the process reaches a value $\psi^{(1)}_{t} \geq 2A$, or $\delta_t > s/12$.

\emph{Proof (of claim).}
Consider the Doob submartingale $Y_t = \sum_{\tau=0}^{t-1} X_t$ with increments $(X_t)$ given as:
$$
X_{t} = \begin{cases}
\Delta\psi^{(1)}_t - \frac{ps A}{48 n}, &\text{if $Y_t > \frac{A}{2}$ or $\delta_t > s/12$},\\
0, & \text{otherwise}.
\end{cases}
$$
Noting that $|X_t| \leq \frac{9}{n}$, an application of the Azuma inequality for submartingales (cf. e.g.~\cite{concen.pdf}[Thm.\ 16]) to $(Y_T)$ with $T = \frac{72n}{ps}$ gives:
$$
\Pr [Y_T \leq - \frac{A}{2}] \leq \exp{[-A^2 psn / 46656]},
$$
Moreover, assuming the barrier $\delta_t>s/12$ was not reached, we have:
$$
\left(\psi^{(1)}_T \big| Y_T > - \frac{A}{2}\right) = \psi^{(1)}_0 + \frac{ps A}{48 n} T + Y_T > A + \frac{ps A}{48 n} \frac{72n}{ps} - \frac{A}{2} = 2A,
$$
which completes the proof of the claim.

We now prove the lemma by iteratively applying the claim over successive intervals of time $(\tau_0, \tau_1,\ldots)$, such that $\tau_0 = t_2$ and $\tau_{i+1}$ is the first time step not before $\tau_i$ such that $\psi^{(1)}_{\tau_{i+1}} \geq 2\psi^{(1)}_{\tau_{i}}$ or $\delta_{\tau_{i+1}} \geq s/12$. By the claim, we have:
$$
\Pr\left[\tau_{i+1} - \tau_{i} \leq \frac{72n}{ps}\right] \geq 1-\exp{[-(\psi^{(1)}_{\tau_{i}})^2 psn / 46656]}.
$$
Noting that $c_3 > 48/(ps)$ by definition, and that before the barrier $\delta > s/12$ is reached, we have $\psi^{(1)}_{\tau_{i}} \geq \frac{4c_3}{\sqrt n} 2^i \geq \frac{192}{ps\sqrt n} 2^i$, we obtain:
$$
\Pr\left[\tau_{i+1} - \tau_{i} \leq \frac{72n}{ps}\right] > 1-\exp{[-4^{i+1}]}.
$$
and further:
$$
\Pr\left[\tau_{i+1} \leq \frac{72n}{ps}(i+1)\right] > \prod_{j=0}^{i}\left(1-\exp{[-4^{j+1}]}\right) > 1 - \sum_{j=0}^{i} \exp{[-4^{j+1}]} > 0.98.
$$
In particular, putting $i=\log_2 n$, $\Pr\left[\tau_{i} \leq \frac{72n\log_2 n}{ps}\right] > 0.98$. Since for this value of $i$, we must have $\delta_{\tau_{i}} \geq s/12$ (since otherwise we would have $\psi^{(1)}_{\tau_{i}} = \omega(1)$, which is impossible), the claim of the lemma follows.
\end{proof}

\paragraph{Phase with $\delta > s/12$.}

The second phase of convergence corresponds to configurations of the system which are sufficiently far from the center point $(a_1,a_2,a_3) = (s/3, s/3, s/3)$. Formally, we analyze a variant of potential $\phi$ (with an additive corrective factor proportional to $\kappa^2$) to show that, starting from a configuration with $\delta > s/12$, we will eliminate one of the three populations $a_1, a_2, a_3$ in $O(n \log n)$ sequential steps with constant probability, without approaching the center point too closely (a value of $\delta = \Omega(1)$ will be maintained throughout).

For this part of the analysis, we define the considered potential as:
\begin{equation}\label{eq:psi}
\psi = \eta^2 - \frac {4 p} {s^2} \kappa^2 = \ln \frac{s^3}{27} - \phi - \frac {4 p} {s^2} \kappa^2,
\end{equation}
for any configuration $\uu$ with $a_{\min} > 0$.

We have directly from \eqref{eq:phi_dot_xzero} and \eqref{eq:kappa_dot}:
\begin{align}
\dot \psi &= -\dot \phi - \frac {4 p} {s^2} 2 \kappa \dot \kappa \geq \frac{p}{s}\bigg(\frac 12 \delta^2 - \kappa \delta + 4 \kappa^2 - 144 \frac p s \kappa\delta\bigg) =\nonumber\\
& = \frac{p}{s}\bigg(\frac 14 \delta^2 + \left(\frac{\delta}{2} - 2\kappa\right)^2 + \left(1 - 144 \frac p s\right)\kappa\delta\bigg)
\geq \frac 14 \frac{p}{s} \delta^2, \label{eq:psi_dot}
\end{align}
where in the last transformation we took into account that $p \leq s/144$.

For the sake of technical precision in formulating the subsequent lemmas, we also consider the stochastic process $\psi^*_t$, given as $\psi^*_t = \psi(\uu_t)$ for any $t < t_d$, where $t_d$ is defined as the first time in the evolution of the system such that a configuration with $a_{\min,t_d} < c_6/n$ is reached, where $c_6 = 313600/s$ is a constant depending only on $s$. For all $t \geq t_d$, we define $\psi^*_{t} := \psi^*_{t-1} + \frac{1}{n}$.

\begin{lemma}\label{lem:psi_average}
In any configuration $\uu_t$ with $\delta \geq s/20$ we have: $\E \Delta\psi_t^* \ge  \frac 18 \frac{p\delta^2}{sn}  > \frac{ps}{3600n}$.
\end{lemma}
\begin{proof}
We have:
\begin{equation}\label{eq:delta_psi}
\E \Delta \psi = - \E \Delta \phi - \frac {4 p} {s^2} \E \Delta(\kappa^2)
\end{equation}
Following the definition of $\phi$ in Eq.~\eqref{eq:potential}, we have by linearity of expectation:
$$
\E \Delta\phi = \E \left(\sum \ln(a_i+\Delta a_i)- \sum \ln a_i \right)= \sum \E \ln \left(1+\frac{\Delta a_i}{a_i}\right).
$$
Next, using the bound $\ln(1+b) \leq b$ which holds for $b >-1$, we have:
$$
\sum \E \ln (1+\frac{\Delta a_i}{a_i}) \le \sum \frac{\E \Delta a_i}{a_i} = \sum \frac{\dot{a_i}/n}{a_i} = \dot{\phi}/n
$$
from which it follows that:
\begin{equation}\label{eq:delta_phi}
\E \Delta \phi \leq \frac{1}{n}\dot \phi
\end{equation}
To analyze $\E \Delta(\kappa^2)$, we apply a variant of Lemma~\ref{lem:taylor}. A direct application of the lemma is not sufficient due to the singularity related to the $a_i^{-1}$ term in the definition of $\kappa_i$; however, this effect is compensated when we take into account that any change of the value of $\kappa_i^2$ occurs in the considered protocol with probability at most proportional to $a_i$. For the specific case of $\kappa_i^2$, for fixed $i=1,2,3$, we consider $\kappa_i^2 : \R^2 \to \R$ as a function of the restricted configuration $\bar \uu = (a_i^+, a_i^{++})$, and we rewrite expression~\eqref{eq:taylor_approx} as:
\begin{align*}
\E(\Delta \kappa_i^2) &= \sum_{\uu' = (a_i^{+'}, a_i^{++'}) \neq \uu} \left(\nabla f(\uu) \cdot (\uu'-\uu) + R_2(\uu,\uu')\right)\Pr(\uu'|\uu) \leq \nonumber\\
&\leq \frac{\dot {\kappa_i^2}}{n} +  \frac{1}{n^2} \max_{\|\uu^* - \uu\|_{\infty} \leq 1/n} D_{\kappa_i^2}(\uu^*) \sum_{\uu' = (a_i^{+'}, a_i^{++'}) \neq \uu} \Pr(\uu'|\uu) \leq  \frac{\dot {\kappa_i^2}}{n} +  \frac{1}{n^2} \max_{\|\uu^* - \uu\|_{\infty} \leq 1/n} D_{\kappa_i^2}(\uu^*) a_i.
\end{align*}

A straightforward computation from the definition of function $\kappa_i$ shows that:
$$
\max_{\|\uu^* - \uu\|_{\infty} \leq 1/n} D_{\kappa_i^2}(\uu^*) \leq \frac{8s^2}{a_i^2}.
$$
It follows that
$$
\E(\Delta \kappa_i^2) \leq \frac{\dot {\kappa_i^2}}{n} +  \frac{a_i}{n^2} \frac{8s^2}{a_i^2} =  \frac1n\left(\dot {\kappa_i^2} + \frac{8s^2}{a_i n}\right),
$$
and so:
\begin{equation}\label{eq:delta_kappa2}
\E (\Delta \kappa^2) \leq \frac{1}{n}\left(\dot {\kappa^2} + \frac{24s^2}{a_{\min} n}\right).
\end{equation}
Introducing~\eqref{eq:delta_phi} and~\eqref{eq:delta_kappa2} into~\eqref{eq:delta_psi}, we obtain:
\begin{align*}
\E \Delta \psi &= - \E \Delta \phi - \frac {4 p} {s^2} \E \Delta(\kappa^2) \geq
-\frac{1}{n} \dot \phi -  \frac{1}{n} \frac {4 p} {s^2}\left( \dot{\kappa^2} + \frac{24s^2}{a_{\min} n} \right)= \frac1n \dot \psi-\frac{96p}{n^2a_{\min}} \geq\\
&\geq \frac{p}{4sn}\left(\delta^2 - \frac{392s}{a_{\min} n}\right) \geq \frac{p}{8sn} \delta^2,
\end{align*}
where in the second-to-last transformation we used~\eqref{eq:psi_dot}, and in the last transformation we used the relation $ \frac{392s}{a_{\min} n} \leq \frac{\delta^2}{2}$ which holds when $\delta \geq s/20$ and $a_{\min} \geq c_6/n$.

The claim thus follows when $\psi^*_{t} = \psi_{t}$ and $\psi^*_{t+1} = \psi_{t+1}$, i.e., for $t < t_d$. For larger values of $t$, the claim follows trivially from the definition of $\psi^*_{t}$.
\end{proof}

The above Lemma is used to show that, starting from any configuration with $\delta > s/12$, we quickly reach a configuration in which some species has a constant number of agents.

\begin{lemma}\label{lem:psi_concentration}
If $\delta_t \geq s/20$, we have:
\begin{itemize}
\item[(i)] $|\Delta\psi^*_t| \leq c_7$,
\item[(ii)] $\Var(\Delta\psi^*_t) \leq \frac{c_8}{n}$.
\end{itemize}
where $c_7>0$ and $c_8>0$ are constants depending only on $s$. Moreover, in any configuration $\uu$ with $a_{\min} \geq 2/n$, we have:
\begin{itemize}
\item[(iii)] $|\Delta\psi(\uu)| \leq \frac{c_7}{n a_{\min}}$.
\item[(iv)] $\Var(\Delta\psi(\uu)) \leq \frac{c_8}{n^2 a_{\min}}$.
\end{itemize}
\end{lemma}
\begin{proof}
We first consider the case of a configuration with $a_{\min} \geq 2/n$. Using the definition of $\psi$ (and within it, of $\phi$ and $\kappa$). Consider any transition from a configuration $\uu$ to a subsequent configuration $\uu'$ and let $S \subseteq \{1,2,3\}$ be defined as the set of indices of configurations changing between $\uu$ and $\uu'$ ($S = \{ i : a_i^+(\uu) \neq a_i^+(\uu') \vee a_i^{++}(\uu) \neq a_i^{++}(\uu')\}$). We verify that there exists an absolute constant $c_7 > 0$ such that:
$$
|\psi (\uu') - \psi(\uu)| \leq \frac{c_7}{n \min_{i\in S} a_i} \leq c_7.
$$
Moreover, by the definition of the protocol a transition from $\uu$ to $\uu'$ occurs with probability $\Pr(\uu'|\uu) \leq \min_{i\in S} a_i$. Since there is only a constant number of possible successor configurations $\uu_{t+1}$ for $\uu_t$ (loosely bounding, not more than $3^6$), it follows that:
$$
\Pr \left[|\Delta \psi(\uu)| > \frac{1}{b}\right] \begin{cases}
< \frac{3^6 c_7 b}{n}, &\text{for any $b>0$}\\
= 0, & \text{for $b > \frac{c_7}{n a_{\min}}$}.
\end{cases}
$$
The bounds on the variance of $\Var(\Delta \psi(\uu))$ and that of $\Delta \psi^*_t = \Delta \psi(\uu)$ (for $t < t_d$) with  $a_{\min,t} \geq (c+6+1)/n$ follow directly. The analysis of $\Delta \psi^*_t$ when $a_{\min,t} = c_6/n$ and $t < t_d$ is performed analogously, noting that if the succeeding configuration $\uu' = \uu_{t+1}$ is such that $a_{\min}(\uu') < c_6/n $, then $\Delta \psi^*_t = \frac{1}{n}$. Finally, for $t\geq t_d$, the result holds trivially by the definition of $\psi^*_t$.
\end{proof}

\begin{lemma}\label{lem:pot_psi_first}
Let $\uu_{t_3}$ be an arbitrary starting configuration of the system such that $\delta_{t_3} > s/12$. Then, with probability $1-O(1/n)$, for some $t_4=t_3+O(n \log n)$, a configuration $\uu_{t_4}$ is reached such that $a_{\min,t_4} = c_6$.
\end{lemma}
\begin{proof}
W.l.o.g. assume that $t_3 = 0$. First we remark that, by the relation between $\eta$ and $\delta$ for $\delta \leq s/12$, a process starting with $\delta_0 > s/12$ satisfies:
$$
\psi_0 = \psi_0 = \eta^2_0 - \frac{4p}{s^2}\kappa_0^2 > \frac{(s/12)^2}{s^2} - \frac{12p}{s^2} > \frac{1}{150}.
$$
Moreover, for any configuration $\uu'$ with $\delta(\uu') \leq s/20$ we have:
$$
\psi(\uu') = \eta^2(\uu') - \frac{4p}{s^2}\kappa^2 \leq  \frac{(\frac{3}{2}\frac{s}{20})^2}{s^2} < \frac{1}{170}.
$$
Thus, initially $\psi_0 > \frac{1}{150}$ and as long as for all time steps $t$ we have $\psi_t \geq \frac{1}{170}$, the barrier condition $\delta_t \geq s/20$ has not been violated. Moreover, for $\psi_t \in [\frac{1}{170}, \frac{1}{150}]$, we have by Lemma~\ref{lem:psi_average} that $\E \Delta \psi_t \geq 0$. Moreover, by Lemma~\ref{lem:psi_concentration}~(iii) and the fact that $\delta_t < \frac{1}{144}$ which implies $a_{\min,t} > s/4$, we have that $|\Delta \psi_t| \leq \frac{4c_7}{sn}$.

It follows from a standard application of Azuma's inequality for martingales (resembling the analysis of the hitting time of the random walk with step size $O(1/n)$, from one endpoint of a path of length $\Theta(1)$ to the other) that:
$$
\Pr\left[\exists_{t < n^2 / \ln n}\ \psi_t < \frac{1}{170}\right] = O(1/n),
$$
hence also throughout the first $n^2 / \ln n$ steps of the process we have $\delta > s/18$, with probability $1 - O(1/n)$.

We are now ready to analyze the subsequent stages of the process, designing a Doob submartingale $Y_t =\sum_{\tau=0}^{t-1} X_t$ with time increments $(X_t)$ defined as:
$$
X_{t} = \begin{cases}
\Delta\psi^*_t - \frac{ps}{3600}, &\text{if $\psi_\tau > 1/170$ and $a_{\min,\tau} \geq c_6$ for all $\tau \leq t$}\\
0, & \text{otherwise}.
\end{cases}
$$
Using Lemma~\ref{lem:psi_concentration}~(i)~and~(ii) and applying the Azuma-McDiarmid inequality\footnote{If our objective in the proof of the lemma were to show a bound on $t_4$ which holds with constant probability (which would be sufficient for our purposes later on), rather than a w.h.p. bound, then this specific step of the proof can also be performed using Markov's inequality. In any case, we would need to make use of the bounded variance of $\psi^*_t$ in the proof of the next Lemma.} in the bounded variance version (cf. e.g.~\cite{concen.pdf}[Thm.\ 18]) to $Y_t$ for $t_c = c^3 n \ln n$, for some sufficiently large constant $c > 0$ depending only on $s$, we obtain:
$$
\Pr[Y_t \leq -c^2 \ln n] \leq \exp{\left[-\frac{c^4\ln^2 n}{2t\frac{c_8}{n} + \frac{2c_7c^2 \ln n}{3}}\right]} = \exp{\left[-\frac{c\ln n}{2 c_8 + \frac{2}{3} \frac{c_7}{c}}\right]} = 1 - O(1/n).
$$
If the event $X_{t} = \Delta\psi^*_t - \frac{ps}{3600}$ were to hold for all $t < t_c$ with $c=\frac{2 \cdot 3600}{ps}$ and if $Y_{t_c} > - c^2 \ln n$, then we would have $\psi^*_{t_c} = \psi^*_0 + Y_{t_c} + \frac{ps}{3600} t_c \geq 0 - c^2 \ln n + 3 c^2 \ln n = 2c^2 \ln n$, which would mean that $\psi^*_{t_c} \neq \psi_{t_c}$, since $\psi \leq 3\ln n + O(1)$ by definition. If $\psi^*_{t_c} \neq \psi_{t_c}$, then $t_4 < t_c$, and the proof is complete. (Indeed, to reach a configuration with $a_{\min} < c_6/n$, the protocol has to pass through a configuration with $a_{\min} = c_6/n$, since the size of each population changes by at most $1$ in each transition.) Otherwise, we must have that at least one of the following events holds: $Y_{t_c} \leq - c^2\ln n$, or $\psi_\tau \leq 1/170$ for some $\tau < t_c$, or $a_{\min,\tau} < c_6$ for some $\tau < t_c$. We have established that each of the first two of these events holds with probability $O(1/n)$, whereas if the latter event holds, then $t_4 < t_c$. Thus, $t_4 < t_c$ holds with probability $1-O(1/n)$ by a union bound.
\end{proof}

\begin{lemma}\label{lem:pot_psi_second}
Let $\uu_{t_4}$ be a starting configuration of the system such that $a_{\min,t_4} = c_6/n$. Then, with constant probability, for some $t_5=t_4+O(n \log n)$, a configuration $\uu_{t_5}$ is reached such that $a_{j,t_5} \leq c_6/n$ and $a_{j+1,t_5} > s/40$, for some $j \in \{1,2,3\}$.
\end{lemma}
\begin{proof}
W.l.o.g. assume that $\arg\min_{i=1,2,3} a_{i,t_4} =2$. If $a_{3,t_4} > s/40$, then the claim follows immediately, putting $t_5 = t_4$ and $j = 2$. Otherwise, we will show that with constant probability, the system will evolve so that $a_2$ will increase over time until within $O(n \log n)$ steps we will have a time step $t_5$ with $j =3$ (i.e., $a_{3, t_5} \leq c_6/n$ and $a_{1, t_5} > s/40$).

In the considered case, w.l.o.g. assume $t_4 =0$. Next, let $T = c n \ln n$ for a sufficiently large constant $c$; we choose as $c:= 2 \log_2 \frac{1}{0.005ps}$ for convenience in later analysis. Intuitively, in view of Lemmas~\ref{lem:psi_average} and~\ref{lem:psi_concentration}, the potential $\psi^*_T$ will be further increased in the next steps: the random variable $(\psi^*_T - \psi^*_0 | \uu_0)$ 
has an expected value of $\Theta(T/n) = \Theta(\log n)$, with a standard deviation of $\Theta(\sqrt{T/n}) = \Theta(\sqrt{\log n})$.

By an application of the Azuma-McDiarmid inequality for martingales with bounded variance similar to that in the proof of Lemma~\ref{lem:pot_psi_first}, we obtain the following result:\footnote{Such an analysis can also be performed using Chebyshev's inequality, obtaining a slightly weaker expression in the probability bound.}
\begin{equation}
\label{eq:stable_potential_blah}
\Pr \left[\forall_{t \leq T}\ \psi^*_t \geq \psi^*_0 + \frac{pst}{3600n} - \eps\frac{T}{n} \right] = 1 - n^{-\Omega(1)}, \quad\text{for any constant $\eps > 0$}.
\end{equation}

Observe that since $a_{2,0} = c_6/n = O(1/n)$, we have $\psi^*_0 \geq \ln n - O(1)$. Taking this into account, for our purposes, a slightly weaker and simpler form of expression~\eqref{eq:stable_potential_blah} will be more convenient:
\begin{equation}
\label{eq:stable_potential}
\Pr \left[\forall_{t \in [0.5 n \log_2 n, T]}\ \psi^*_t \geq (1 + 10^{-4}ps)\ln n \right] = 1 - o(1).
\end{equation}
The proof of the lemma is completed by a more fine-grained analysis of the considered protocol. In the initial configuration $t_4 = 0$, we have $a_{2,0} = c_6/n$ (there are exactly $c_6$ agents in state $A_2$), and since $t_4 \neq t_5$, we have $a_{3,0} < s/40$. Consequently, $a_{1,0} = s - s/40 - O(1/n) > 0.9s$. Informally, since the prey of $A_2$ (i.e., $A_1$) is more than twice more numerous than its predator (i.e., $A_3$), we should observe the increase in the size of population of $A_2$, regardless of the activities ($A_i^+$ or $A_i^{++}$) of the agents in the population. We consider the evolution of the system, finishing at the earliest time $t_e$ when $a_2(t_e) > s/100$. The following relations are readily shown (apply e.g.\ Lemma~\ref{lem:ranges_2} with $i=2$ and $\uu=0$):
\begin{align}
\E (\Delta a_{2,t} | a_{3,t} < 0.05 s, t < t_e) &\geq \frac{0.05ps}{n} a_2 \label{eq:exp_delta_a2}\\
\E (\Delta a_{3,t} | a_{3,t} < 0.05 s, t < t_e) &\leq -\frac{0.05ps}{n} a_3 \label{eq:exp_delta_a3}.
\end{align}
From~\eqref{eq:exp_delta_a3}, taking into account that $|\Delta a_{i,t}| \leq \frac{1}{n}$ and $a_{3,0} < s/40$, an application of Azuma's inequality for martingales shows that:
\begin{equation}
\label{eq:athree_bounded}
\Pr [\forall_{t \leq \min\{t_e, T\}}\ a_{3,t} < 0.05 s ] = 1- o(1).
\end{equation}
Taking into account the above, by a straightforward geometric growth analysis (compare e.g.\ proof of Lemma~\ref{lem:pot_positive}), we obtain from~\eqref{eq:exp_delta_a2}:
\begin{equation}
\label{eq:te_upper}
\Pr [t_e < T] = 1-o(1).
\end{equation}
Moreover, since the speed of increase of $a_2$ is bounded (even in the absence of predators) by that of a standard push rumor spreading process (formally, $\E (\Delta a_{2,t}) \leq a_{2,t}$), we have (compare e.g.~\cite{rsexact}):
\begin{equation}
\label{eq:te_lower}
\Pr [t_e > 0.5 n \log_2 n] = 1-o(1).
\end{equation}
Now, we observe that with constant probability, the size of population $A_2$ does not decrease in the time interval $[0, t_e]$ below the value $a_{2,0} = c_6/n$, attained at the beginning of this interval:
\begin{equation}\label{eq:csixbarrier}
\Pr \left[\forall_{t\in [0,t_e]}\ a_{2,t}\geq c_6/n\right] = \Omega(1).
\end{equation}

Indeed, with constant probability the value $a_{2,t}$ is initially non-decreasing: with constant probability, in the first $O(n)$ rounds each of the $c_6 = O(1)$ agents from $A_2$ will be triggered by the scheduler $O(1)$ times in total, and each interaction involving an agent from $A_2$ will have this agent as the initiator, and an agent from the largest of the three populations, $A_1$, as the receiver (the prey). Thus, with constant probability, the number of agents in population $A_2$ is increased to an arbitrary large constant (e.g., $1000c_6$). After this, we use the geometric growth property~\eqref{eq:exp_delta_a2} to show that $a_{2,t}$ reaches the barrier $a_{2,t} > s/100$ (at time $t_e$) before the event $a_{2,t} < c_6/n$ occurs (cf.~e.g.~proof of Lemma~\ref{lem:pot_positive}, or standard analysis of variants of rumor-spreading processes in their initial phase~\cite{rspushpull}).

When the event from bound~\eqref{eq:csixbarrier} holds, at least one of the following events must also hold:
\begin{itemize}
\item[(A)] $a_{\min,t} \geq c_6/n$, for all $t\leq t_e$,
\item[(B)] or there exists a time step $t < t_e$ such that $a_{1,t} \leq c_6/n$,
\item[(C)] or there exists a time step $t < t_e$ such that $a_{3,t} \leq c_6/n$.
\end{itemize}
To complete the proof, we will show that each of the events (A) and (B) holds with probability $o(1)$. Indeed, then in view of~\eqref{eq:csixbarrier}, event (C) will necessarily hold with probability $\Omega(1)$. This means that, with probability $\Omega(1)$, there exists a time step $t < t_e$ such that $a_{3,t} < c_6/n$ and $a_{2,t}< s/100$ (since $t < t_e$), and so also $a_{1,t} > s - s/100 - c_6/n > 0.98 s > s/40$; thus, the claim of the lemma will hold with $t_5 = t$ and $j=3$.

To show that event (B) holds with probability $o(1)$, notice that $a_{2,t} < s/100$ by definition of $t_e$, and moreover $a_{3,t} < 0.05s $ with probability $1-o(1)$, hence the event $a_{1,t} < s - s/100 - 0.05s = 0.94s$ holds with probability $o(1)$.

To show that event (A) holds with probability $o(1)$, notice that, substituting in~\eqref{eq:stable_potential} $t = t_e$, by a union bound over~\eqref{eq:stable_potential},~\eqref{eq:te_upper} and \eqref{eq:te_lower} we obtain:
$$
\Pr \left[\psi^*_{t_e} \geq (1 + 10^{-4}ps)\ln n \right] = 1 - o(1).
$$
This means that, with probability $1-o(1)$, we have  $\psi^*_{t_e} \neq \psi_{t_e}$ or $\psi_{t_e} \geq (1 + 10^{-4}ps)\ln n$. In the first case, event (A) cannot hold. In the second case, observe that $a_{2, t_e} = s/100 + O(1/n)$ by definition of $t_e$, so $a_{1,t_e} > s - s/100 - O(1/n) - c_6/n > 0.98 s$, and it follows that $\psi_{t_e} = \sum_{i=1}^3 \ln \frac{1}{a_{i, t_e}} + O(1) = \ln n +O(1)$. Since the condition  $\psi_{t_e} \geq (1 + 10^{-4}ps)\ln n$ is not fulfilled, event (A) can only hold with probability $o(1)$.
\end{proof}

\begin{lemma}\label{lem:eliminationoffirstspecies}
Let $\uu_{t_5}$ be a starting configuration of the system such that $a_{j,t_5} \leq c_6/n$ and $a_{j+1,t_5} > s/40$, for some $j \in \{1,2,3\}$. Then, with constant probability, for some $t_6=t_5+O(n)$, a configuration $\uu_{t_6}$ is reached such that $a_{\min,t_6} = 0$.
\end{lemma}
\begin{proof}
We consider the pairs of interacting agents chosen by the scheduler in precisely the next $n$ rounds after time $t_5$. Given that set $A_{j,t_5}$ has constant size, and set $A_{j+1,t_5}$ has linear size in $n$, it is straightforward to verify that with constant probability, the set of randomly chosen $n$ pairs of agents has all of the following properties:
\begin{itemize}
\item Each agent from $A_{j,t_5}$ belongs to exactly one pair picked by the scheduler, and is the receiver in this pair.
\item Each agent interacting in a pair with an agent from $A_{j,t_5}$ belongs to exactly one pair.
\item Each agent interacting in a pair with an agent from $A_{j,t_5}$ belongs to set $A_{j+1,t_5}$.
\end{itemize}
Conditioned on such a choice of interacting pairs by the scheduler, the protocol changes the state of all agents from set $A_{j,t_5}$ to state $j+1$ with probability at least $p^{|A_{j,t_5}|} \geq p^{c_6} = \Omega(1)$. State $j$ is then effectively eliminated.
\end{proof}

In the absence of species $j$, the interaction between species $j-1$ and $j+1$ collapses to a lazy predator-prey process, with transitions of the form $(A_{j-1}, A_{j+1})\to (A_{j-1}, A_{j-1})$ associated with a constant transition probability. A w.h.p.\ bound on the time of elimination of species $j+1$ follows immediately from the analysis of the push rumor spreading model, and we have the following Lemma.

\begin{lemma}\label{lem:eliminationofsecondspecies}
Let $\uu_{t_6}$ be a starting configuration of the system such that $a_{j,t_6} = 0$, for some $j \in \{1,2,3\}$. Then, with probability $1 - O(1/n)$, for some $t_7=t_6+O(n \log n)$, a configuration $\uu_{t_7}$ is reached such that for all $t \geq t_7$, $a_{j,t} = a_{j+1,t} = 0$ and $a_{j-1,t} = s$.
\qed
\end{lemma}
After a further $O(n \log n)$ steps after time $t_7$, the final configuration of all agents in the oscillator's population will be $a_{j-1,t}^{++} = s$.

\subsection{Operation of the Oscillator in the Presence of a Source}\label{sec:33}

In this section we prove properties of the oscillatory dynamics for the case $\many X>0$. It is possible to provide a detailed analysis of the limit trajectories of the dynamics in this case, as a function of the concentration of $x$. Here, for the sake of compactness we only show the minimal number of properties of the oscillator required for the proof of Theorem~\ref{thm:osc}. When the given configuration is such that $a_{\min}$ is sufficiently large, say $a_{\min} > 0.02s$, then both the subclaims of Theorem~\ref{thm:osc}(2) hold for the considered configuration. (The first subclaim hold directly; the second subclaim follows by a straightforward concentration analysis of the number of agents changing state in protocol $P_o$ over the next $0.01s n$ steps, since we will always have $a_{\min} \geq 0.01s$ during the considered time interval.) Otherwise, the considered configuration is close to one of the sides of the triangle. We will show that in the next $O(n \log n)$ steps, with high probability, the protocol will either reach a configuration with $a_{\min} > 0.02s$, or will visit successive areas around the triangle, as illustrated in Fig.~\ref{fig:triangle2}. The following Lemmas show that within each area, an exponential growth process occurs, which propagates the agent towards the next area.

\begin{figure}[ht]
\begin{centering}
\includegraphics[scale=1]{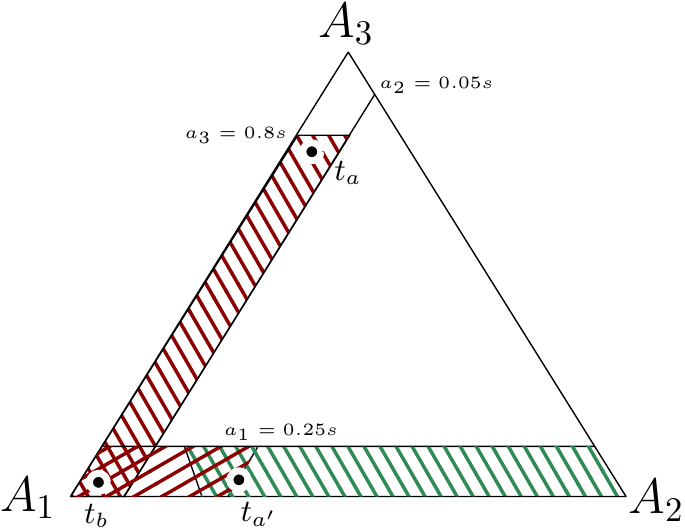}
\caption{Configurations considered in Lemmas~\ref{lem:ranges_1} and \ref{lem:ranges_1X} (descending dashed slope) and Lemmas \ref{lem:ranges_2} and~\ref{lem:ranges_2X} (ascending dashed slope) for $i=1$ (dark red) and $i=2$ (green). Unless the process leaves the shaded area towards the center o the triangle, the shape of the time trajectory is counter-clockwise around the triangle, w.h.p.}
\label{fig:triangle2}
\end{centering}
\end{figure}

\begin{lemma}\label{lem:ranges_1}
 If $a_{i-1}<0.8s$ and $a_{i+1}<0.05s$, then $\dot{a}_{i-1} \le  xs/3 - 0.05psa_{i-1}$.
\end{lemma}
\begin{proof}
From the assumptions we have that $a_{i}>0.15s$. Starting from \eqref{eq:a_dot} we obtain:
  \begin{align*}
    \dot{a}_{i-1} &= x(s/3 -a_{i-1})+ pa_{i+1}(a_{i-1}+a_{i-1}^{++}) - pa_{i-1}(a_i+a_i^{++})\\
   &\le xs/3+ 2pa_{i-1}a_{i+1}- pa_{i-1}a_i \\
   &\le xs/3+ pa_{i-1}(2a_{i+1}- a_i) \\
   &\le xs/3+ pa_{i-1}(0.1s-0.15s)=xs/3- 0.05psa_{i-1} \\
  \end{align*}
\end{proof}

From the above bound on expectation, the following Lemma follows directly by a standard concentration analysis. In what follows, we consider an execution in which the concentration $x$ is strictly positive and bounded by a sufficiently small absolute constant (i.e., $\many X$ is at most a given constant fraction of the entire population), with the required upper bounds on $x$ used in the proofs of lemmas given in their statements. This is a technical assumption, which allows us to simplify the proof structure. In particular, the assumption $\many X \leq c_{12} n$ can be omitted in the statement of Theorem~\ref{thm:osc}, and the claim of the theorem can even be proved for executions in which $\many X$ changes during the execution of the protocol, as long as the invariant $\many X > 0$ is preserved over the considered interval of time.

\begin{lemma}\label{lem:ranges_1X}
Let $\uu_{t_a}$ be a starting configuration of the system such that $a_{i-1,t_a}<0.75s$ and $a_{i+1,t_a}<0.05s$. Suppose $x < 10^{-3} ps$, starting from time $t_a$. Then, for some $t_b \in [t_a, t_a + c_9 n]$, where $c_9$ is a constant depending only on $p$ and $s$, with probability $1-e^{-n^{\Omega(1)})}$, the system reaches a configuration $\uu_{t_b}$ such that exactly one of the following two conditions is fulfilled:
\begin{itemize}
\item either $a_{\min,t_b} \geq 0.02s$,
\item or $a_{i+1, t_b} < 0.05s$ and $a_{i-1, t_b} < 0.02s$.
\end{itemize}
\end{lemma}
\begin{proof}
In the considered range of values of $a_{i-1}$, we have $a_{i-1,t_a} < 0.75s$ and $a_{i-1,t} \geq 0.02s$, for all $t$ until we leave the considered area at time $t_b$. Taking into account that $x < 10^{-3} ps$, it follows from Lemma~\ref{lem:ranges_1} that:
$$
\E\Delta{a}_{i-1} \le \frac1n (xs/3 - 0.05psa_{i-1}) \leq \frac1n (x < 10^{-3} ps^2/3 - 10^{-3}ps^2) < - 0.0005 \frac{ps^2}{n}.
$$
Taking into account that $|\Delta{a}_{i-1} \le \frac1n|$, it follows from a straightforward concentration analysis
 (cf. e.g. proof of Lemma~\ref{lem:pot_positive} for a typical analysis of this type of exponential growth process) that a boundary of the considered area (either $a_{i-1, t} < 0.02s$ or $a_{\min, t} > 0.02s$) must be reached within $O(n)$ steps with very high probability, as stated in the claim of the lemma.
\end{proof}

A similar analysis is performed for the next area.
\begin{lemma}\label{lem:ranges_2}
 If $a_{i+1}<0.25s$ and $a_{i-1}<0.05s$ then $\dot{a}_{i+1} \ge xs/12+0.6psa_{i+1}$ and $\dot{a}_{i-1} \le xs/3- 0.2psa_{i-1}$.
 \end{lemma}
\begin{proof}
From assumptions we have that $a_{i}>0.7s$. Starting from \eqref{eq:a_dot} we obtain:
  \begin{align*}
   \dot{a}_{i+1} &= x(s/3 -a_{i+1})+ pa_{i}(a_{i+1}+a_{i+1}^{++}) - pa_{i+1}(a_{i-1}+a_{i-1}^{++})\\
   &\ge x(s/3-a_{i+1}) + pa_{i+1}(a_{i}-2a_{i-1})\\
   &\ge x(s/3-0.25s) + pa_{i+1}(0.7s-2\cdot 0.05s) \\
   &= xs/12+0.6psa_{i+1}\\
   \\
   \dot{a}_{i-1} &= x(s/3 -a_{i-1})+ pa_{i+1}(a_{i-1}+a_{i-1}^{++}) - pa_{i-1}(a_{i+1}+a_{i+1}^{++})\\
   &\le xs/3+ pa_{i-1}(2a_{i+1}-a_{i}) \\
   &\le xs/3+ pa_{i-1}(2\cdot0.25s-0.7s)=xs/3- 0.2psa_{i-1}.
  \end{align*}
\end{proof}

Again, a concentration result follows directly.
\begin{lemma}\label{lem:ranges_2X}
Let $\uu_{t_b}$ be a starting configuration of the system such that $a_{i-1,t_b}<0.02s$ and $a_{i+1,t_b}<0.02s$. Suppose $x < 0.02 ps$, starting from time $t_b$. Then, for some $t_{a'} \in [t_b, t_b + c_{10}n \ln \frac{1}{\max\{1/n,a_{i+1,t_b}\}}] \subseteq [t_b, t_b + c_{10} n \ln n]$, where $c_{10}$ is a constant depending only on $p$ and $s$, with probability $1-O(1/n^3)$, the system reaches a configuration $\uu_{t_{a'}}$ such that exactly one of the following two conditions is fulfilled:
\begin{itemize}
\item either $a_{\min,t_{a'}} \geq 0.02s$,
\item or $a_{i-1, t_{a'}} < 0.05s$, $a_{i+1, t_{a'}} > 0.25s$, and (consequently) $a_{i, t_{a'}} < 0.75s$.
\end{itemize}
\end{lemma}
\begin{proof}
We first show that, starting from time $t_b$ onward, the process $a_{i-1, t}$ satisfies $a_{i-1, t} < 0.05s$ for all $t \in [t_b, t_*]$ with probability $1-e^{-n^{\Omega(1)}}$, where $t_*$ is defined as the minimum of time $t_a + c_{10} n \ln n$ and the last time moment such that $a_{i+1, t} \leq 0.25s$ holds for all $t \in [t_b, t_*]$. By Lemma~\ref{lem:ranges_2}, we have for all $t \in [t_b, t_*]$ such that $a_{i-1, t} > 0.01s$:
$$
\E \Delta a_{i-1} = \frac 1n \dot{a}_{i-1} \le \frac 1n (xs/3- 0.2psa_{i-1}) < \frac 1n (0.02ps^2/3 - 0.02ps^2) < -\frac{ 0.01ps^2}{n},
$$
where we took into account the assumption $x < 0.02ps$. The claim on $a_{i-1, t} < 0.05s$  follows from a standard concentration analysis, noting that $|\Delta a_{i-1,t}| \leq \frac1n$.

In order to analyze the process $a_{i+1,t}$, we apply a filter and consider the process $a'_{i+1,t}$, starting at time $t_b$, defined as follows. For as long as $a_{i-1, t} < 0.05s$, we put $a'_{i+1,t} := a_{i+1,t}$, and starting from the first time $t_{**}$ when $a_{i-1, t} > 0.05s$, we compute $a'_{i+1,t+1}$ as the subsequent value of $a_{i+1}$ after a simulation of a single step of the process for some state $\uu$ with concentrations of types: $x(\uu) = x$, $a_{i+1}(\uu) = a'_{i+1,t}$, $a_{i-1}(\uu) = 0.05s$, and $a_i(\uu) = s - a_{i-1}(\uu) - a_{i+1}(\uu)$.

For a given time step $t$, let $R_t$ denote the event that $\Delta a'_{i+1,t} \neq 0$. By the construction of protocol $P_o$, which always requires at least one agent of type $X$ or type $A_{i+1}$ to be involved in an interaction which creates or destroys an agent of type $A_{i+1}$, we have:
$$
\Pr [R_t] \leq x + 2 a'_{i+1,t}.
$$
Moreover, from Lemma~\ref{lem:ranges_2} it follows that for $a'_{i+1,t} < 0.25s$:
 $$
 \E\Delta a'_{i+1,t} = \frac{\dot a'_{i+1,t}}{n}\ge \frac1n (xs/12+0.6psa'_{i+1,t}).
$$
Since $\Delta a'_{i+1,t} | \neg R_t = 0$, we have:
$$
 \E\Delta a'_{i+1,t} | R_t \ge \frac{\frac1n (xs/12+0.6psa'_{i+1,t})}{\Pr[R_t]} \geq \frac{\frac1n (xs/12+0.6psa'_{i+1,t})}{x + 2 a'_{i+1,t}} \geq \frac{0.3 ps}{n},
$$
and moreover $\Delta a'_{i+1,t} | R_t \in \{-\frac{1}{n}, 0, \frac{1}{n}\}$. Analysis of this type of process is folklore (in the context of epidemic models with infection and recovery) but somewhat tedious; we sketch the argument for the sake of completeness. When considering only those steps for which event $R_t$ holds, the considered process can be dominated by a lazy random walk on the line $\{0,\frac{1}{n}, \frac{2,n}, \ldots\}$, with a constant bias towards its right endpoint. To facilitate analysis, we define points $Q_c = {c \lfloor \alpha \ln n \rfloor}{n}$, for $c = 0,1,\ldots$, where constant $\alpha > 0$ is subsequently suitably chosen, and for any point $Q_c$ to the right of the starting point of the walk (i.e, $c > c_{\min}$ where $c_{\min}$ is the smallest integer such that $Q_{c_{\min}+1} > a_{i+1,t_b}$) define $s_c$ as the number of steps of the walk until its first visit to $s_c$. For a suitable choice of constants $\alpha$ and $\beta > 0$ sufficiently large, we have that for any $c$, with probability at least $1 - O(1/n^2)$, $s_{c+1} - s_c \leq \beta \ln n$, and moreover between its step $s_c$ and its step $s_{c+1}$, the walk is confined to the subpath $(Q_{c-1}, Q_{c+1})$ of the considered path. Considering the original time $t$ of our process $a'_{i+1,t}$ (including moments with $\neg R_t$), let $t_c$ be the moment of time corresponding to the $s_c$-th step of the walk. Conditioning on events which hold with probability $1 - O(1/n^2)$, the value $t_{c+1} - t_c$ can be stochastically dominated by the sum of $\beta\ln n$ independent geometrically distributed random variables, each with expected value $O(\frac{n}{\max\{1, (c-1) \ln n\}})$. Let $c_{\max}$ be the largest positive integer such that $Q_{c_{\max}} < 0.25s$. Applying a union bound on the conditioning of all intervals $t_{c+1} - t_c$, for $c \geq c_{min}$ and a concentration bound on the considered geometric random variables, we eventually obtain that with probability $1 - O(1/n^3)$ the condition $a'_{i+1,t}$ is achieved for time:
\begin{align*}
t &< \sum_{c = c_{\min}}^{c_{\max}} (t_{c+1} - t_c) = O\left(n + \ln n \sum_{c = c_{\min}}^{c_{\max}} \frac{n}{\max\{1, (c-1) \ln n\}}\right) = \\ & O\left(n (\ln n - \ln \max\{1,a_{i+1,t_b} n\} + 1)\right) = O\left(n \ln \frac{1}{\max\{1/n,a_{i+1,t_b}\}}\right).
\end{align*}
Recalling that $a'_{i+1,t} = a_{i+1,t}$ holds throughout the considered time interval with very high probability, the claim follows.
\end{proof}

An iterated application of Lemmas~\ref{lem:ranges_1X} and Lemmas~\ref{lem:ranges_2X} moves the process along time moments $t_a$, $t_b$, $t_a', \ldots$, where time moment $t_a'$ is again be fed to Lemma~\ref{lem:ranges_1X}, considering the succeeding value of $i$. After a threefold application of both Lemmas, the process has w.h.p. in $O(n \log n)$ steps either performed a complete rotation, passing through three moments of time designated as ``$t_a$'', rotated by one third of a full circle, or has reached at some time $t'$ a point with $a_{\min,t'} \geq 0.02s$. In either case, the claim of Theorem~\ref{thm:osc}(2) follows directly.

\section{Analysis of Protocol for \textsc{Detect}}

\subsection{Further Properties of the Oscillator}\label{xx}

We start by stating a slight generalization of Lemma~\ref{lem:psi_average}, capturing the expected change of potential $\psi_t^*$ (given by~\eqref{eq:psi}) for the case $\many X > 0$, for configurations which are sufficiently far from both the center and the sides of the triangle.

\begin{lemma}\label{lem:psi_average_Xpos}
In any configuration $\uu_t$ with $10^{-6}s^2 \leq a_{\min} \leq 0.02s$ and $x < c_{12}$ we have: $\E \Delta\psi_t \ge \frac{ps}{7200n}$, where $c_{12}>0$ is an absolute constant which depends only on $s$ and $p$.
\end{lemma}
\begin{proof}
We can condition the expectation of $\E \Delta\psi_t$ on the event $E_t$, which holds if an agent in state $X$ participates in the current interaction. Conditioned on $\neg E_t$, the analysis corresponds directly to the computations performed for the case $x=0$, where we remark that the assumptions of Lemma~\ref{lem:psi_average} are satisfied due to the assumed upper bound on $a_{\min}$. Thus:
$$
\E \Delta\psi_t | \neg E_t > \frac{ps}{3600n}.
$$
Next, taking into account the lower bound on $a_{\min}$, by exactly the same argument as in Lemma~\ref{lem:psi_concentration}($iii$), we have $|\Delta\psi_t| < c'_{12}$/n, for some choice of constant $c'_{12} > 0$ which depends only on $s$ and $p$. Obviously,
$$
\E \Delta\psi_t | E_t > -c'_{12}/n,
$$
and since $\Pr[E_t] < 2x$, by the law of total expectation:
$$
\E \Delta\psi_t > (1-2x) \frac{ps}{3600n} - 2x \frac{c'_{12}}{n} > \frac{ps}{7200n},
$$
where the last inequality holds for any $x < c_{12}$, where $c_{12} := \frac{1}{2(2c'_{12}+1)}$.
\end{proof}

\begin{lemma}
\label{lem:oscillatory_intervals}
Suppose $a_{\min,t_0} < 10^{-6} s^4$ at some time $t_0$. Then, there exists an absolute constant $c_{13} > 0$, such that the following event holds with probability $1-e^{-n^{\Omega(1)})}$: for all $t \in [t_0, t_0+ e^{n^{c_{13}}}]$, we have $a_{\min,t} < 0.01s^2$.
\end{lemma}
\begin{proof}
Let $\hat \psi_t \equiv  \ln \frac{s^3}{27} - \psi_t$.
 Consider any $t \geq t_0$ such that $a_{\min,t} < 10^{-6} s^4$. Then, $\phi_{t} < \ln a_{\min,t} < \ln (10^{-6} s^4)$ and consequently:
 $$\hat \psi_t  =  \phi_t + \frac {4 p} {s^4} \kappa_t^2 \leq \ln (10^{-6} s^4) + \frac{4p}{s^4} \leq \ln (2 \cdot 10^{-6} s^4),$$
 where we recall that $\kappa_t^2 \leq 1$ and the last inequality follows for $p$ chosen to be sufficiently small ($4p/s^4 < \ln 2$).

Further, note that if for some time $t$ we have $\hat \psi_t < \ln (8 \cdot 10^{-6} s^4)$, then:
$$
\phi_t = \hat \psi_t - \frac {4 p} {s^4} \kappa_t^2 < \ln (8 \cdot 10^{-6} s^4),
$$
thus $a_{1,t} a_{2,t} a_{3,t} < 8 \cdot 10^{-6} s^4$, from which it follows that $a_{\min,t}^2 < 16 \cdot 10^{-6} s^4$, and so $a_{\min,t} < 0.01 s^2$.

Thus, for $\hat \psi_t < \ln (8 \cdot 10^{-6} s^4)$, at least one of the following holds:
\begin{itemize}
\item Either $\hat \psi_t < \ln (2 \cdot 10^{-6} s^4)$,
\item Or $\hat \psi_t \geq \ln (2 \cdot 10^{-6} s^4)$, thus $a_{\min, t} \geq 10^{-6} s^4$. Then, taking into account that $a_{\min,t} < 0.01 s^2$, we have by Lemma~\ref{lem:psi_average_Xpos}: $\E \Delta \hat \psi_t < - \frac{ps}{7200n}$.
 \end{itemize}
Taking into account the known properties of function $\psi_t$ (Lemma~\ref{lem:psi_concentration}), we have that starting from $\hat \psi_{t_0} < \ln (2 \cdot 10^{-6} s^4)$, it takes time exponential in a polynomial of $n$ ($e^{\Omega(n^{c_{13}})}$, for some absolute constant $c_{13} > 0$) to break the potential barrier for $\hat \psi$, i.e., to reach the first moment of time $t_1$ such that $\hat \psi_{t_1} \geq \ln (8 \cdot 10^{-6} s^4)$, with probability $1-e^{-n^{\Omega(1)})}$, for some absolute constant $c_{14} > 0$. To complete the proof, recall that for any $t < t_1$, we have $\hat \psi_{t}< \ln (8 \cdot 10^{-6} s^4)$, and so as previously established, $a_{\min,t} < 0.01 s^2$.
\end{proof}

For any execution of the oscillator protocol $P_o$, we can now divide the axis of time into maximal time intervals of two types, which we call \emph{oscillatory} and \emph{central}. A central time interval continues for as long as the condition $a_{\min,t} \geq 10^{-6} s^4$ is fulfilled, and turns into an oscillatory interval as soon as this condition no longer holds. An oscillatory time interval continues for as long as the condition $a_{\min,t} < 0.01s^2 $ is fulfilled, and turns into an oscillatory interval as soon as this condition no longer holds. Lemma~\ref{lem:oscillatory_intervals} implies that an oscillatory interval is of exponential length w.v.h.p.

\begin{lemma}\label{lem:fastoscillations}
Suppose $0 < x < c_{12}$. Let $t_0$ be an arbitrary moment of time such that $a_{\min,t_0} > 0$. Let $T = C n \ln \frac{1}{a_{\min,t_0}}$, for an arbitrarily fixed constant positive integer $C =O(1)$. With probability $1 - e^{-n^{\Omega(1)}}$, we have for all subsequent moments of time $t \in [t_0, t_0 + T]$:
$$
\ln\frac{1}{\max\{\frac{1}{n}, a_{\min,t}\}} \leq 100C \ln \frac{1}{a_{\min, t_0}}.
$$
\end{lemma}
\begin{proof}
Without loss of generality assume $t_0 = 0$.
We can assume in the proof that $a_{\min,0} > n^{-0.01 / C }$, otherwise the claim trivially holds. Thus, initially we have $\phi_{0} \geq 3\ln a_{\min,0} \geq -\frac{0.03}{C} \ln n \geq -0.03 \ln n$.

We proceed to show that potential $\phi$ does not decrease much during the considered motion. We have at any time $\dot \phi \geq -p$, which follows directly from~\eqref{eq:derphi}.

Suppose at some time $t$ we have $\phi_t \geq -0.2 \ln n$. We note that $a_{\min,t} \geq e^{\phi_t} \geq n^{-0.2}$.
Applying Lemma~\ref{lem:taylor}, we have under these assumptions:
$$
\E \Delta \phi_t \geq -\frac{p}{n} - O\left(\frac{1}{n^2 a_{\min,t}^2}\right) \geq -\frac{1}{n}
$$
and moreover by the properties of the natural logarithm (cf. e.g.~\cite{ICALP15}):
$$
|\Delta \phi_t\| < \frac{4}{n a_{\min, t}} \leq \frac{4 e^{-\phi_t}}{n} \leq 4 n^{-0.8}.
$$
As usual, we apply a Doob martingale, with $\phi'_t := \phi_t$ until the first moment of time $t$ such that $\phi_t < -0.2 \ln n$, and subsequently $\phi'_{t+1} := \phi'_t$ for larger $t$. We have $\E \Delta \phi_t \geq -n^{-1}$ and $|\Delta \phi_t\| < 4 n^{-0.8}$.

Considering $T = C n \ln \frac{1}{a_{\min,0}} < C n \ln n$ steps of the process starting from time $0$, by a standard application of Azuma's inequality, we obtain that with probability $1 - e^{-n^{\Omega(1)}}$, for all $t \in [0,T]$ we have:
$$
\phi'_t -  \phi'_0 \geq - \frac{T}{n} - n^{-0.2} \geq - C \ln \frac{1}{a_{\min,0}} - 1 \geq -0.1 \ln n.
$$
From the last inequality it follows that $\phi'_t \geq \phi'_0 - 0.1 \ln n \geq -0.2 \ln n$ for all $t \in T$, with probability $1 - e^{-n^{\Omega(1)}}$, and so $\phi'_t = \phi_t$. We now rewrite the same bound for $\phi_t$, using the relation $\phi \geq 3 \ln a_{\min}$:
$$
\phi_t \geq  \phi_0 - C \ln \frac{1}{a_{\min,0}} - 1 \geq
3 \ln a_{\min,0} - C \ln \frac{1}{a_{\min,0}} - 1 \geq -(C + 4) \ln \frac{1}{a_{\min,0}} \geq -100C \ln \frac{1}{a_{\min,0}}.
$$
Taking into account that $-\ln \frac{1}{a_{\min, t}} \geq \phi_t$ for $a_{\min, t} \neq 0$, we obtain the claim.
\end{proof}

\begin{lemma}\label{lem:goaround}
Suppose $0 < x < c_{12}$. Fix type $i \in \{1,2,3\}$. Let $t_0$ be any time such that $a_{\min,t_0} < 10^{-6}s^4$. Let $t^* > t$ be the first moment after $t_0$ such that $i$ is the most represented type, $a_{i,t^*} = a_{\max, t^*}$. Let $t^{**} > t^*$ be the first moment after $t^*$ such that $i$ is the least represented type, $a_{i,t^{**}} = a_{\min, t^{**}}$.

Then, with probability $1 - O(1/n^3)$,  $t^{**} \leq t_0 + c_{11}n \ln \frac{1}{\max\{1/n,a_{i+1,t_0}\}}$, where $c_{11}>0$ is an absolute constant depending only on $s$ and $p$.
\end{lemma}
\begin{proof}
From Lemma~\ref{lem:oscillatory_intervals}, we have that w.v.h.p.\, the protocol is in an oscillatory interval which will last super polynomial time, i.e., with probability $1-e^{-n^{\Omega(1)})}$: for all $t \in [t_0, t_0+ e^{n^{c_{13}}}]$, we have $a_{\min,t} < 0.02s$. Acting as in the previous subsection, we iteratively apply Lemma~\ref{lem:ranges_1X} and Lemma~\ref{lem:ranges_2X}. After at most 6 applications of both Lemmas, the process has performed two complete rotations around the triangle,  w.h.p., passing in particular through time moments $t^*$ where the designated type $i$ was a maximal type and $t^{**}$ where type $i$ was a minimal type. It remains to bound the time required to perform these iterations. We consider as an example a single application of Lemma~\ref{lem:ranges_2X} starting at a time $t_b$ and ending at a time $t_{a'}$, where $t_{a'} < t_b + c_{10}n \ln \frac{1}{\max\{1/n,a_{i+1,t_b}\}}$ with probability $1 - O(1/n^3)$. Applying Lemma~\ref{lem:fastoscillations} at time $t_b$, we obtain  $\ln \frac{1}{a_{\min, t_{a'}}} \leq 100C \ln \frac{1}{a_{\min, t_b}}$, w.v.h.p. We use this bound for the next application of Lemma~\ref{lem:ranges_1X}, and so on. After a total of at most 12 applications, we eventually obtain a bound  of the form $c_{11} \ln \frac{1}{a_{\min, t_0}}$ on the length of the considered time interval, where the value of $c_{11}$ is computed as a function of $C$.
\end{proof}

\subsection{Protocol Extension \texorpdfstring{$P_m$}{Pm}: Majority}\label{pm}

The composition of the extension is specified in Fig.~\ref{fig:newrules_m}.
In what follows, we denote $M^{(s)}_i := \many (A_i^?,M_s)$ and $m^{(s)}_i := M^{(s)}_i / n$, for $s \in \{-1, 0, +1\}$.

\begin{lemma}
Suppose $0 < x < c_{12}$. Let $t_0$ be an arbitrarily chosen moment of time and let $T>0$. For fixed $i \in \{1,2,3\}$, we have for all $t \in [t_0,t_0+T]$:
$$|m^{(+1)}_{i,t} - m^{(-1)}_{i,t}| \leq 2 e^{4rT/n} \max\{n^{-0.2}, a_{i,t_0}\},$$
with probability $1 - O(e^{-n^{1/6}})$.
\end{lemma}
\begin{proof}
W.l.o.g.\ assume $t_0=0$ and $i=1$. Denote $A_0 := \many A_{1,0}$, $D_t := M^{(+1)}_{1,t} - M^{(-1)}_{1,t}$, and $G_t := D_t^2$. As usual, we denote $\Delta D_t = D_{t+1} - D_t$ and $\Delta G_t = G_{t+1} - G_t$. For the subsequent analysis, we choose to use the ``squared potential'' $G_t$ to simplify considerations; this would be the usual potential of choice to analyze an unbiased random walk with a fair coin toss.

First we remark that $|D_0| \leq A_0$, so $|G_0| \leq A_0^2$. Next, observe that since at most one agent changes its state in a single time step, we have $|\Delta D_t| \leq 2$, and so:
\begin{equation}
\label{eq:deltaGt_max}
|\Delta G_t| \leq (|D_t| + 2)^2 - |D_t|^2 \leq 4(|D_t| + 1) \leq 4n +4.
\end{equation}
We now upper bound the expectation  $\E\Delta G_t$. We condition this expectation on the disjoint set of events $R_6, R_7, R_8, R_9, R_{10}, R_0$, where $R_r$, for $6 \leq r \leq 10$, corresponds to Rule $(r)$ being executed in the current step, and $R_0$ is the event that none of these rules is executed. We have the following:
\begin{itemize}
\item If event $R_0$ or $R_6$ holds, then at least one of the following three situations occurs: (1) the values of $M^{(+1)}_{1}$ and $ M^{(-1)}_{1}$ both remain unchanged at time $t$, (2) an agent changes state from type $A_1$ to another type, or (3) an agent turns from another type into type $A_1$. In case (1), we have $D_{t+1} = D_t$. In case (2), the probability that $|D_{t+1}| = |D_t| + 1$ is not more than the probability that $|D_{t+1}| = |D_t| - 1$, since by the construction of the protocol, the choice of the agent leaving the population is completely independent of its value $M_s$. In case (3), we have $\Pr[|D_{t+1}| = |D_t| - 1] = \Pr[|D_{t+1}| = |D_t| +1] = 1/2$ by construction. In all cases, $\Pr[|D_{t+1}| = |D_t| + 1] \leq \Pr[|D_{t+1}| = |D_t| - 1]$. We therefore have:
    $$
    \E [\Delta G_t | R_0 \vee R_6] \leq ((|D_t| + 1)^2 - |D_t|^2) + ((|D_t| - 1)^2 - |D_t|^2) \leq 2.
    $$
\item For events $R_7$ and $R_8$, we have $D_{t+1} | R_7 = D_t - 1$ and $D_{t+1} | R_8 = D_t + 1$. Since events $R_7$ and $R_8$ hold with equal probability, it follows that:
    $$
    \E [\Delta G_t | R_7 \vee R_8] \leq ((|D_t| + 1)^2 - |D_t|^2) + ((|D_t| - 1)^2 - |D_t|^2) \leq 2.
    $$
\item Finally, for events $R_9$ and $R_{10}$, we have $D_{t+1} | R_9 = D_t + 1$ and $D_{t+1} | R_{10} = D_t - 1$. Since $\Pr[R_9] \cdot M^{(-1)}_{1} = \Pr[R_{10}] \cdot M^{(+1)}_{1}$
    \begin{align*}
    \E [\Delta G_t | R_9 \vee R_{10}] &\leq \frac{\max\{M^{(+1)}_{1},M^{(-1)}_{1}\}}{A_1} ((|D_t| + 1)^2 - |D_t|^2) + \\
                                           & \frac{\min\{M^{(+1)}_{1},M^{(-1)}_{1}\}}{A_1} ((|D_t| - 1)^2 - |D_t|^2) \leq \\
                                           &\frac{4|D_t|^2}{A_1} + 2 = \frac{4 G_t}{A_1} + 2,
    \end{align*}
    where we assume in notation that $A_1 > 0$.
\end{itemize}
Applying the law of total expectation for $\Delta G_t$ over the set of events $R_0 \vee R_6, R_7 \vee R_8, R_9\vee R_{10}$ and noting that $\Pr[R_9 \wedge A_{10}] \leq r\frac{A_1}{n}$, we eventually obtain:
\begin{equation}
    \label{eq:deltaGt_e}
    \E \Delta G_t \leq 4r\frac{G_t}{n} + 2.
\end{equation}
Inequalities~\eqref{eq:deltaGt_max} and \eqref{eq:deltaGt_e} are sufficient to lower-bound the evolution of random variable $G_t$, which undergoes multiplicative drift with rate parameter $1 + 4r/n$ (up to lower order terms). Since known multiplicative drift lower bounds~(cf.e.g.~\cite{dx,dy}) do not appear to cover this case explicitly, we sketch the corresponding submartingale analysis (with slightly weaker parameters) for the sake of completeness.

Consider any moment of time $t$ such that $G_{t} \geq n^{1.6}$. Define target value $G_{\max} = (1 + 8r) G_{t} \geq n^{1.6}$. Consider the following filter, defining $G'_\tau$, $\tau \geq 0$ as the submartingale with $G'_{\tau} = G_{t+\tau}$ until the first moment of time at which $G_{t+ \tau} < G_{\max}$, and $G'_{\tau} = G'_{\tau+1} + \frac{5r}{n} G_{\max}$ for all subsequent moments of time. Note that $\E\Delta G'_\tau \leq \frac{5r}{n} G_{\max}$ by~\eqref{eq:deltaGt_e} and $|\Delta G'_\tau| \leq 4n + 4 \leq 5n$ by~\eqref{eq:deltaGt_max} (where we conduct the entire analysis for $n$ sufficiently large with respect to absolute constants of the algorithm). By Azuma's inequality, we have for any $\tau > 0$ and $z > 0$:
$$
\Pr[G'_{\tau} > G'_0 + \frac{5r}{n} G_{\max} \tau + z] \leq \exp\left[\frac{-z^2}{2 \tau (5n)^2}\right].
$$
Next, choosing any $\tau \leq n$ and $z = r G_{\max} \geq r n^{1.6}$ and noting that:
$$
G'_0 + \frac{5r}{n} G_{\max} \tau + z \leq G_{t} + 5r G_{\max} + r G_{\max} \leq G_{\max},
$$
we rewrite the concentration equality as:
$$
\Pr[G'_{\tau} > G_{\max}] \leq \exp\left[\frac{-(r n^{1.6})^2}{2 \tau (5n)^2}\right] < e^{-n^{0.19}}.
$$
Applying to the above a union bound over all $\tau \in [0,n]$, we obtain by another crude estimate:
$$
\Pr[\forall_{\tau \in [0,n]} G'_{\tau} \leq G_{\max}] \geq 1 - e^{-n^{0.18}},
$$
from which it follows directly by the definition of $G'_{\tau}$ that:
$$
\Pr[\forall_{\tau \in [0,n]} G_{t+\tau} \leq G_{\max}] \geq 1 - e^{-n^{0.18}}.
$$
Thus, given that $G_{\max} = (1 + 8r) G_{t}$, the value of $G_t$ increases by a factor of at most $(1+8r)$ over $n$ steps, with very high probability. Iterating the argument at most a logarithmic number of times and applying a union bound gives for arbitrary $t$:
$$
G_t \leq (1 + 8r)^{t/n + 1} \max \{n^{1.6},A_0^2\} \leq 2 e^{8rt/n} \max \{n^{1.6}, A_0^2\},
$$
with probability at least $1 - e^{-n^{1/6}}$, from which the claim of the lemma follows directly after taking the square root and normalizing by a factor of $n$.
\end{proof}

By considering the sizes of populations $m^{(-1)}_{i,t}$, $m^{(0)}_{i,t}$, and $m^{(+1)}_{i,t}$ (whose sum is $a_{i,t}$), we obtain the following corollary of the above Lemma, applied for a suitably chosen value $T = \frac{0.001n}{r} \ln \frac{1}{a_{i,t_0}}$.

\begin{lemma}\label{lem:mainmajorityx}
Suppose $0 < x < c_{12}$. Let $t_0$ be an arbitrarily chosen moment of time with $a_{i,t_0} \leq 0.02 s^2$. For fixed $i \in \{1,2,3\}$, we have for all $t \in [t_0,t_0+ \frac{0.001n}{r} \ln \frac{1}{\max\{\frac{1}n, a_{i,t_0}\}}]$:
$$
|m^{(+1)}_{i,t} - m^{(-1)}_{i,t}| \leq 0.1 a_{i,t} + 0.05s,
$$
with probability $1 - O(e^{-n^{1/6}})$.
\qed\end{lemma}

The above Lemma provides a crucial lower bound on the size of population $m^{(0)}_{i}$.

\begin{lemma}\label{lem:mainmajority}
Suppose $0 < x < c_{12}$. Let $t_0$ be an arbitrarily chosen moment of time with $a_{i,t_0} \leq 0.02 s^2$. For fixed $i \in \{1,2,3\}$, we have for all $t \in [t_0,t_0+ \frac{0.005n}{r} \ln \frac{1}{\max\{\frac{1}n, a_{i,t_0}\}}]$ such that $a_{i,t} > 0.25s$:
$$
\min\{m^{(+1)}_{i,t}, m^{(-1)}_{i,t}\} \geq  0.001 s^3,
$$
with probability $1 - e^{-n^{\Omega(1)}}$.
\end{lemma}
\begin{proof}
Note first that by the conditions $a_{i,t_0} \leq 0.02 s^2 < 0.1s < 0.25s < a_{i,t}$. Let $t_1 > t_0$ denote the last moment of time before $t_2$ such that $a_{i,t_1-1} < 0.1 s$ and let $t_2 \geq t_1 + 0.15 sn$ denote the first moment of time after $t_1$ such that $a_{i,t_2} > 0.25 s$. We have $t \geq t_2 + 0.05sn$.

We now consider the process $m^{(+1)}_{i,\tau}$ (exactly the same arguments may be applied for process $m^{(-1)}_{i,\tau}$). We have $\Delta m^{(+1)}_{i,\tau} \leq 1/n$. The analysis is divided into two phases:
\begin{itemize}
\item Phase 1: $\tau \in [t_1, t_2)$ (thus  $a_{i,\tau} \in [0.1 s, 0.25s]$). Initially, $m^{(+1)}_{i,t_1} \geq 0$. Suppose for some step $\tau$ we have $m^{(+1)}_{i,\tau} < 0.02s$. Rule (6) is executed with probability $(a_{i, \tau})(1-a_{i, \tau}) \geq 0.1s (1 - 0.25s) \geq 0.075s^2$, whereas rule (8) or (9) which reduces $m^{(+1)}_{i,\tau}$ is executed with probability at most $2r m^{(+1)}_{i,\tau} < 0.04rs < 0.01s^2$. A computation of the expected value provides:
    $$
    \E [\Delta m^{(+1)}_{i,\tau} | (m^{(+1)}_{i,\tau} < 0.02s)] > \frac{0.06s^2}{n}.
    $$
    An application of Azuma's inequality yields that $m^{(+1)}_{i,t_2} > \frac{1}{2} \frac{0.06s^2}{n} (t_2 - t_1) > 0.004 s^3$, with probability $1 - e^{-n^{\Omega(1)}}$. In case of failure, we consider the process no further.
\item Phase 2: $\tau \in [t_2, t)$ (thus  $a_{i,\tau} \geq 0.1s$). From Phase 1, we have that initially $m^{(+1)}_{i,t_2} > 0.004 s^3$. Suppose for some step $\tau$ we have $0.001s^3 < m^{(+1)}_{i,\tau} < 0.01s$. We consider two cases:
    \begin{itemize}
    \item If  $a_{i,\tau} \geq 0.25s$, then by Lemma~\ref{lem:mainmajorityx} we have:
    $$
    |m^{(+1)}_{i,t} - m^{(-1)}_{i,t}| \leq 0.1 a_{i,t} + 0.05s,
    $$
    with probability $1 - e^{-n^{\Omega(1)}}$. In case of failure we interrupt the analysis (this is an implicit application of union bounds over successive steps $\tau$). Under the assumption $m^{(+1)}_{i,\tau} < 0.01sr$, we conclude:
    $$
    m^{(-1)}_{i,t} \leq 0.1 a_{i,t} + 0.05s + m^{(+1)}_{i,t} < 0.3 a_{i,t},
    $$
    and hence:
    \begin{equation}\label{eq:zo}
    m^{(-1)}_{i,t} \leq \frac{1}{2}m^{(0)}_{i,t} - 0.005s.
    \end{equation}
    Now, to compute the expected value $\Delta m^{(+1)}_{i,\tau}$, we remark that rule (6) does not decrease this expected value since $m^{(0)}_{i,\tau}$. Moreover, in view of~\eqref{eq:zo}, the probability of executing rule (9) (which increases $m^{(-1)}_{i,t}$ by $1/n$) exceeds the probability of executing one of the rules (7) or (8) (which decrease $m^{(-1)}_{i,t}$ by $1/n$) by $0.005s m^{(+1)}_{i,\tau} r > 5 \cdot 10^-6 s^4 r$ by the assumption $0.001s^3 < m^{(+1)}_{i,\tau}$. We eventually obtain in this case:
    \begin{equation}\label{eq:fr}
    \E [\Delta m^{(+1)}_{i,\tau} | (0.001s^3 < m^{(+1)}_{i,\tau} < 0.01s)] > \frac{5 \cdot 10^-6 s^4 r}{n}.
    \end{equation}
    \item If $a_{i,\tau} \geq 0.25s$, then assuming $m^{(+1)}_{i,\tau} < 0.01s$, we can perform an analogous analysis as in the first phase to obtain:
    $$
    \E [\Delta m^{(+1)}_{i,\tau} | (m^{(+1)}_{i,\tau} < 0.01s)] > \frac{0.06s^2}{n},
    $$
    which, in particular, also implies~\eqref{eq:fr}.
    \end{itemize}
    We have thus shown that the expected change to $m^{(+1)}_{i,\tau}$ satisfies \eqref{eq:fr}. Noting that initially $m^{(+1)}_{i,t_2} > 0.004 s^3 = 0.001 s^3 + 0.003 s^3$, an application of Azuma's inequality to an appropriate Doob martingale with~\eqref{eq:fr} shows that the event $m^{(+1)}_{i,\tau} > 0.001 s^3$ will hold for all remaining steps of the process $\tau$, with probability $1- e^{-n^{\Omega(1)}}$.
\end{itemize}
\end{proof}

\begin{lemma}\label{lem:large_portion}
Let $t$ be any moment of time with $a_{\min,t} > 4r^{1/2}$. For all $i \in \{1,2,3\}$, we have:
$$
\min\{m^{(+1)}_{t}, m^{(-1)}_{t}\} \geq  r^{3/2} s
$$
with probability $1 - O(e^{-n^{\Omega(1)}})$.
\end{lemma}
\begin{proof}
Denote $c = 4r^{1/2}$. To show the claim, observe that necessarily for all $\tau \in [t - cn/2, t]$ we have $a_{\min,\tau} > c/2$. We consider the change of value $m^{(+1)}_{\tau}$ over time (the argument for $m^{(-1)}_{\tau}$ follows symmetrically). Initially, we have $m^{(+1)}_{t - cn/2} \geq 0$, and at every step $|\Delta m^{(+1)}_{\tau}| \leq 1/n$. At any time $\tau$ such that $m^{(+1)}_{\tau} < s/4$ we have the following cases:
\begin{itemize}
\item Rule (6) is executed, which happens with probability at least $a^2_{\min,\tau} > c^2/4$. Since $m^{(+1)}_{\tau} < s/4$, conditioned on this event, the expected value of $\Delta m^{(+1)}_{\tau}$ is at least $1/4$.
\item One of the rules (7)-(10) is executed, which occurs with probability at most $r$.
\item In all other cases, we have $\Delta m^{(+1)}_{\tau} = 0$.
\end{itemize}
Noting that $r = c^2/16$, the claim follows from a standard application of Azuma's inequality.
\end{proof}

\begin{lemma}\label{lem:majority_any}
Suppose $0 < x < c_{12}$. Let $t \geq 2c_{11}n \log n$ be an arbitrary moment of time. Then, $\min\{m_{+1, t}, m_{-1, t}\} > c_{15}$, with probability $1 - O(1/n)$, for some absolute constant $c_{15} > 0$ depending only on $s$, $p$, and $r$.
\end{lemma}
\begin{proof}
Assume w.l.o.g.\ $t = 2c_{11} n \ln n$. Instead of analyzing the evolution of the real system, we consider an execution of a system which is coupled with it over the first $t$ steps as follows. First, starting from time $0$, we perform $t$ steps of protocol $P_o$ (i.e., considering only rules $(1)-(5)$ of its definition, and without setting values of the second component $M_?$). Next, we once again activate the pairs of elements which were activated in the first part of the coupling, in the same order, applying rules $(6)-(10)$ of the protocol with the same outcome which they would have received in the original execution. Clearly, at time $t$ the same configuration $\uu(t)$ is reached by both the original and coupled execution.

Consider first the execution of $P_o$ from time $0$. If~$a_{\min, 0} < 10^{-6}s^4$, then the execution is in an oscillatory interval at time $0$, and will remain in it ($a_{\min} < 0.02 s^2$) until time $t = c \ln n$ with probability $1 - e^{-n^{\Omega(1)}}$. Then, for all $\tau \in [0,t]$ we assume that the claim of Lemma~\ref{lem:goaround} holds with $t_0 = \tau$ for all $i \in \{1,2,3\}$. By a crude union bound, this event holds with probability $1 - O(1/n)$; from now on we assume this is true. (Formally, to allow us to proceed, in the analysis we can say with implicitly couple the system with a different set of random choices to which the system switches in the low-probability event that the  claim of Lemma~\ref{lem:goaround} were not to hold for some $t_0 = \tau$ for the original system.) Given that in the claim of Lemma~\ref{lem:goaround} for $t_0=0$ we have $t^{**} < t$, and for $t_0=t$ we have $t^{*} \geq t$, by the properties of the time intervals $[t^*, t^{**}]$ we observe that there must exist a time $\tau \in [0,t]$ and a type $i \in \{1,2,3\}$ such that $A_i$ is the least represented type at time $\tau$ and the most represented type at time $t$ ($a_{\max, t} = a_{i, t}$ and $a_{\min, \tau} = a_{i, \tau}$), and moreover $t \in [t^*, t^{**}]$ in the claim of Lemma~\ref{lem:goaround} with choice of $t_0 = \tau$. Since $a_{\max, t} \geq s/3$, we now apply Lemma~\ref{lem:mainmajority} with $t_0 = \tau$ to obtain the claim, noting that $a_{\min, \tau}< 0.02s^2$ and moreover that $t < \tau+ \frac{0.005n}{r} \ln \frac{1}{\max\{\frac{1}n, a_{i,t_0}\}}$, given that we have $t < \tau + c_{11}n \ln \frac{1}{\max\{1/n,a_{i+1,t_0}\}}$, where $c_{11}$ is a constant depending only on $s$ and $p$, and noting that we can choose $r < \frac{c_11}{0.005}$.

It remains to consider the cases when the execution starts at time $0$ with $a_{\min, 0} \geq 10^{-6}s^4$. Then, if $a_{\min, t} \geq 10^{-6}s^4$ holds, the claim follows from Lemma~\ref{lem:large_portion} given that $r$ is chosen so that $10^{-6}s^4 > 4 r^{1/2}$. Otherwise, there must exist some last time $t' \leq t$ such that $a_{\min, t'} \geq 10^{-6}s^4$. We apply Lemma~\ref{lem:goaround} with $t_0 = t'$. If the obtained value $t^{**}$ satisfies $t^{**} < t$, then we can apply an analogous analysis as in the case $a_{\min, 0} < 10^{-6}s^4$ to obtain the claim. Otherwise, we have that $t \leq t^{**} \leq t + c n$, where the value of constant $c$, depending only on $s$ and $p$ follows from Lemma~\ref{lem:goaround}. By an iterated application of Lemma~\ref{lem:fastoscillations}, we obtain that $a_{\min, t} \geq c'$, where the value of constant $c'>0$, depending only on $s$ and $p$, follows from the application of Lemma~\ref{lem:fastoscillations}. Choosing $r$ sufficiently small so that $c' > 4 r^{1/2}$, we complete the proof using Lemma~\ref{lem:large_portion}.
\end{proof}

Finally, for the sake of completeness we state how the majority protocol stops in the case of $\many X = 0$.

\begin{lemma}\label{lem:majority_stopping}
Suppose $x=0$. Then, there exists a moment of time $t_s$ such that either $m_{+1,t} = 0$ or $m_{-1,t}= 0$ holds for all $t>t_s$. Moreover, $t_s < c_{16} n \log^2 n$ with probability $1 - O(1/n)$, for some absolute constant $c_{16} > 0$ depending only on $s$, $p$, and $r$.
\end{lemma}
\begin{proof}
By Theorem~\ref{thm:osc}(1), there exists a moment of time $t_0 = O(n \log^2 n)$ such that the system reaches a corner configuration (cf. Lemma~\ref{lem:eliminationofsecondspecies}). W.l.o.g., assume that $a_1 = s$ and $a_2 = a_3 = 0$. At this point, in the majority protocol Rule (7) will never again be activated, whereas the execution of rules (8)-(11) follows precisely the classical majority scenario of Angluin et al.~\cite{angluin2008}. By a standard concentration analysis (see also~\cite{angluin2008}), one of the two species $M_{+1}, M_{-1}$ will become extinct in $O(n \log^2 n)$ steps with probability $1 - O(1/n)$.
\end{proof}
As a side remark on Lemma~\ref{lem:majority_stopping} that it is possible to initialize the system entirely with a state $(A_i, M_0)$ so that  $m_{+1,t} = m_{-1,t}= 0$ holds throughout the process (even if the designed protocol will never enter such a configuration from most initial configurations).

\subsection{Protocol Extension \texorpdfstring{$P_l$}{Pl}: Detection with Lights}\label{pl}

To complete the proof of Theorem~\ref{thm:pr}, we design a protocol extension $P_d$, such that \textsc{Detection} is solved by the composition $((P_o \circ P_m) + P_l)$. Extension $P_l$, uses three states, $\{L_{-1}, L_{+1}, L_{\mathit{on}}\}$. We informally refer to states $L$ as lights. The composition is given in Fig.~\ref{fig:newrules_l}.
Informally, state $L_{-1}$ means that the agent is ``waiting for meeting $M_{-1}$'', then after meeting $M_{-1}$ it becomes $L_{+1}$,  ``waiting for $M_{+1}$'' and finally it becomes $L_{\mathit{on}}$.

To analyze the operation of the protocol, consider first the case of $x=0$. By Lemma~\ref{lem:majority_stopping}, after $O(n \log^2 n)$ steps, agents in at least one of the states $\{M_{+1}, M_{-1}\}$ are permanently eliminated from the system. Thus, either rule (11) or rule (12) will never again be executed in the future. An agent which is in state $L_{\mathit{on}}$ will spontaneously move to another state following rule (13) within $O(\frac{1}{q(\eps)}n \log n)$ steps, with probability $1 - O(1/n^2)$, and will never reenter such a state, since this would require the activation of both rule (11) and rule (12). By applying a union bound over all agents, we obtain that state $L_{\mathit{on}}$ never again appears in the population after $O(n \log n)$ steps from the termination of the majority protocol, with probability $1 - O(1/n)$.
Overall, all nodes reach a state having a different state than $L_{\mathit{on}}$ after $O(n \log^2 n)$ steps from the start of the process, with probability $1 - O(1/n)$, and all leave such a state eventually with certainty.

In the presence of the source $X$, the analysis of the process can be coupled with a Markov chain on three states $L_{-1}$, $L_{+1}$, and $L_{\mathit on}$. In view of Lemma~\ref{lem:majority_any}, transitions from state $L_{-1}$ to $L_{+1}$ and from state $L_{+1}$ to to $L_{-1}$ occur with at least constant probability (except for an $O(1/n)$-fraction of all time steps), this 3-state chain is readily shown to be rapidly mixing. For a choice of $q(\eps) > 0$ depending only on $s, p, r, \eps$ sufficiently small, we can lower-bound the number of agents occupying state $L_{on}$ by $(1 - \eps) n$, with high probability.

Under the natural decoding of states as ``informed'' (having component $L_{\mathit{on}}$) or ``uninformed'' (having component $L_{-1}$ or $L_{+1}$), the proof of Theorem~\ref{thm:pr} is complete. We remark that it is also possible to design a related protocol in which exactly one state is recognized as ``informed'' and exactly one state is recognized as ``uninformed''; we omit the details of the construction.

\section{Proof of Impossibility Result}\label{sec:lb}\label{sec:lower}

This Section is devoted to the proof of Theorem~\ref{thm:lb}. First, we restate some notation. We recall that the vector $z = (z^{(1)}, \ldots, z^{(k)}) \in \{0,1,\ldots,n\}^k = Z$ describes the number of agents having particular states, and $\on{z}=n$. In this section we will identify the set of states with $\{1,\ldots, k\} = [1,k]$. It is now also more convenient for us to work with a scheduler which selects unordered (rather than ordered) pairs of interacting agents; we note that both models are completely equivalent in terms of computing power under a fair random scheduler, since selecting an ordered pair of agents can be seen as selecting an unordered pair, and then setting their orientation through a coin toss. Indexing with integers $\{1,2,\ldots,r\}$ the set of all distinct rules of the protocol, where $r \leq k^4$, for a rule $j \equiv ``\{i_1(j), i_2(j)\} \to \{o_1(j), o_2(j)\}''$, $1 \leq j \leq r$,  $i_1(j), i_2(j), o_1(j), o_2(j) \in \{1,\ldots,k\}$, we will denote by $q_j$, the probability (selected by the protocol designer) that rule $j$ is executed as the next interaction rule once the scheduler has selected $(i_1, i_2)$ as the interacting pair, and by $p_j(z)$ the probability that $j$ is the next rule chosen in configuration $z$ (we have $p_j(z) = q_j \frac{z^{(i_1(j))}z^{(i_2(j))}}{n^2} (1 - O(1/n))$, where the $O(1/n)$ factor compensates the property of a scheduler which always selects a distinct pair of elements).

For any configuration $z_0 \in Z$, we define the \emph{$d$-box} $B_d(z_0)$ around $z_0$ as the set of all states $z \in Z$ such that $ z_0^{(i)}/d \leq z^{(i)} \leq d \max\{1,z_0^{(i)}\}$, for all $1\leq i \leq k$. We start the proof with the following property of boxes.

\begin{lemma}\label{lem:epsilons}
Fix $k \in \Z^+$ and let $0 < \eps_1 < 0.001$ be arbitrarily fixed. There exists $\eps_0 = \eps_0(k, \eps_1)$, $0 < \eps_0 < \eps_1$, such that, for any interaction protocol $P$ with $k$ states and any configuration $z_0 \in Z$, there exists a value $\eps = \eps(P, z_0) \in [\eps_0, \eps_1]$ such that, for any rule $j$ of the protocol, $1\leq j \leq r$, exactly one of the following bounds holds:
\begin{enumerate}
  \item{$(i)$} for all $z \in B_{n^{\eps_0}}(z_0)$, $p_j(z) \leq n^{\eps-1}$,
  \item{$(ii)$} for all $z \in B_{n^{\eps_0}}(z_0)$, $p_j(z) \geq n^{24\eps-1}$.
\end{enumerate}
and for any state $i$, $1\leq i \leq k$, exactly one of the following bounds holds:
\begin{enumerate}
  \item{$(iii)$} for all $z \in B_{n^{\eps_0}}(z_0)$, $z^{(i)} \leq n^\eps$,
  \item{$(iv)$} for all $z \in B_{n^{\eps_0}}(z_0)$, $z^{(i)} \geq n^{24\eps}$.
\end{enumerate}
\end{lemma}
\begin{proof}
Let $k$ be fixed and let $\eps_0 = 96^{-(k + k^4 + f+1)} \leq 96^{-(k+r+f+1)}$, where $ f = \log_2 (1/\eps_1)$. Consider the (multi)set $M$ of real values $M := \{\log_n \max\{n^{\eps_0},z_0^{(i)}\} : i \in \{1,\ldots,r\}\} \cup \{\log_n \max\{n^{\eps_0},np_j(z_0)\} : j \in \{1,\ldots,r\}\}  \subseteq [0,1]$. Since $|M| = k+r$, by the pigeonhole principle, there must exist an interval $I_l = [96^{-l}, 96^{-l+1})$, for some $l \in \{f, \ldots, k+r+f\}$, such that $I_l \cap M = \emptyset$. Now, we set $\eps = 2 \cdot 96 ^{-l} > 96\eps_0$, we also have $\eps < 2 \cdot 96^{-f} < \eps_1$. We immediately obtain that for any state $i$, $1\leq i \leq k$, we either have $z_0^{(i)} \leq n^{\eps/2}$ or $z_0^{(i)} \geq n^{48\eps}$.
Recalling that for any $z \in B_{n^{\eps_0}}(z_0)$, $ z_0^{(i)}/n^{\eps_0} \leq z^{(i)} \leq n^{\eps_0} \max\{1,z_0^{(i)}\}$, claims $(iii)$ and $(iv)$ follow.

To show claims $(i)$ and $(ii)$, notice that if rule $j$, $j \in \{1,\ldots,r\}$ is such that $\min\{z_0^{(i_1(j))}, z_0^{(i_2(j))}\} \leq n^{\eps}$, then for all $z \in B_{n^{\eps_0}}(z_0)$ we have  $\min\{z^{(i_1(j))}, z^{(i_2(j))}\} \leq  n^{\eps}$ (by (iii) and (iv)), and so $p_j(z) \leq n^{\eps-1}$ by the properties of the random scheduler. Otherwise, we have $\min\{z_0^{(i_1(j))}, z_0^{(i_2(j))}\} \geq n^{24\eps}$, and so $\frac{1}{2n^{2\eps_0}} \leq p_j(z_0) / p_j(z) \leq 2n^{2\eps_0}$, where we recall that $\eps > 96\eps_0$. Since we have $p_j(z_0) \leq n^{\eps/2}$ or $p_j(z_0) \geq n^{48\eps}$, claims (i) and (ii) follow.
\end{proof}

Given any $k$-state protocol $P$, we will arbitrarily choose a value of $\eps$ for which the claim of the above Lemma holds (e.g., the smallest possible such value of $\eps$). Note that a similar analysis is also possible for protocols using a super-constant number of states in $n$, however, then the value of $\eps_0$ is dependent on $n$; retracing the arguments in the proof, we can choose appropriately $\eps_0 \geq n^{\exp[-O(k^4)]}$. (We make no effort to optimize the polynomial in $k$ in the exponent.)

In what follows, let $z_0$ be a fixed configuration of the protocol (admitting a certain property which we will define later). We will then consider a rule $j$ to be a \emph{low probability (LP) rule} (writing $j\in LP$) in box $B_{n^{\eps_0}}(z_0)$ if it satisfies condition $(i)$ of the Lemma, and a \emph{high probability (HP) rule} in this box (writing $j\in HP$) if it satisfies condition $(ii)$. Note that $LP \cup HP = \{1,\ldots,r\}$.

Likewise, for $1\leq i \leq k$, we will classify $i$ as a \emph{low-representation (LR) state} (writing $i\in LR$) in box $B_{n^{\eps_0}}(z_0)$ if $i$ satisfies condition $(iii)$ of the Lemma, and a \emph{high representation (HR) state} (writing $i\in HR$) in this box if $i$ satisfies condition $(iv)$. Note that $LR \cup HR = \{1,\ldots,k\}$. Moreover, we define a set of \emph{very high representation (VHR) states}, $VHR \subseteq HR$, as the set of all $i$ such that for all $z' \in B_{n^{\eps_0}}(z_0)$, $z'_i \geq n^{1-8\eps}$. Denoting $HR' = HR \setminus VHR$, we have by the definition of a box that for all $i' \in HR'$, for all $z' \in B_{n^{\eps_0}}(z_0)$: $z'_i \leq n^{1-8\eps}/O(n^{2\eps_0}) <  n^{1-6\eps}$.

From now on, we assume that configuration $z_0$ admits the following property: for $T = n^{1+2\eps}$, an execution of the protocol starting from configuration $z_0$ passes through a sequence of configurations $z_t$, $t=1,2,\ldots, T$, such that the configuration does not leave the box around $z_0$ in any step with sufficiently large probability, lower-bounded by some absolute constant $\Pi \in (0,1]$:
\begin{equation}
\label{eq:boxproperty}
\Pr[\forall_{t < T}\ z_t \in B \subseteq B_{n^{\eps_0}}(z_0)] \geq \Pi,
\end{equation}
where $B$ is an arbitrarily fixed subset of $B_{n^{\eps_0}}(z_0)$.

We now show that the above property has the following crucial implication: for an interacting pair involving selected high and very high representation states, a rule creating a low representation state can only be triggered with sufficiently small probability. Informally, it seldom happens that in the protocol a low representation state is created out of any high representation state.

\begin{lemma}\label{lem:fail}
For a protocol having the property given by~Eq.~\eqref{eq:boxproperty}, for $i_1\in HR$ and $i_2 \in VHR$, let $R_{i_1, i_2}$ be the set of rules of the form $\{i_1, i_2\} \to \{o_1, o_2\}$, taken over all $o_1 \in LR, o_2 \in [1,r]$. Then, $\sum_{j \in R_{i_1, i_2}}q_j = O(n^{-14\eps})$.
\end{lemma}
\begin{proof}
Suppose, by contradiction, that $\sum_{j \in R_{i_1, i_2}} q_j > 3n^{-14\eps}$.

Associate with process $z_t$ a random variable $J_t\in \{0,1\}$, defined as follows. For all $t < t_e$, where $t_e$ is the first moment of time such that $z_{t_e} \notin B_{n^{\eps_0}}(z_0)$, we put $J_t = 1$ if a rule from $R_{i_1, i_2}$ is used for the interaction made by the protocol in process $z_t$ at time $t$, and set $J_t = 0$ otherwise. For all $t \geq t_e$, we set $J_t$ to $1$. We have $\E (J_t | z_1,\ldots,z_t) \geq 2n^{2\eps-1}$; indeed, for $t < t_e$, it holds that:
\begin{align*}
\E (J_t | z_1,\ldots,z_t)= \Pr[J_t = 1 | z_t]& = \sum_{j \in R_{i_1, i_2}}p_j(z_t) = \sum_{j\in R_{i_1, i_2}}q_j \frac{z^{(i_1(j))}z^{(i_2(j))}}{n^2} (1 - O(1/n)) \geq\\
&\geq 3n^{-14\eps}\frac{n^{24\eps} n^{1 - 8 \eps}}{n^2} (1-O(1/n)) > 2n^{2\eps - 1}.
\end{align*}
By a simple stochastic domination argument, $(J_t)$ can be lower-bounded by a sequence of independent binomial trials with success probability $n^{2\eps -1}$, hence by an application of a multiplicative Chernoff bound for $T = n^{1 + 2\eps}$:
$$
\Pr \left[\sum_{t=1}^{T} J_t > n^{4\eps} \right] = \Pr \left[\sum_{t=1}^{T} J_t > \frac{1}{2}2n^{2\eps - 1} T \right] = 1 - o(1),
$$
where the $o(1)$ factor is exponentially small in $n$.

We now show the following claim.

\emph{Claim.} With probability $\Pi - o(1)$, the following event holds: $z_t \in B$ for all $t\in [0,T)$ and the total number of rule activations in the time interval $[0,t)$ during which an agent changes state from a state in $LR$ to a different state is at most $O(k n^{3\eps}$).

\emph{Proof (of claim).} Acting similarly as before, we associate with process $z_t$ a random variable $L_t\in \{0,1\}$, defined as follows. For all $t < t_e$, we put $L_t = 1$ if a rule acting on at least one agent in a state from $LR$ is made by the protocol in process $z_t$ at time $t$, and set $L_t = 0$ otherwise. For all $t \geq t_e$, we set $L_t$ to a dummy variable set always to $0$, i.a.r. We observe that:
$$
\E (L_t | z_1,\ldots,z_t) \leq 2k \frac{n^{\eps}}{n},
$$
since $|LR| \leq k$, and for $t < t_e$, $z_i < n^{\eps}$ for any $i \in LR$, hence the scheduler selects an agent from a $LR$ state into an interacting pair with probability at most $2k \frac{n^{\eps}}{n}$. Applying an analogous argument as in the case of random variable $L_t$, this time for the upper tail, we obtain:
$$
\Pr \left[\sum_{t=1}^{T} L_t < 4kn^{3\eps} \right] = \Pr \left[\sum_{t=1}^{T} L_t > 2\cdot 2k n^{\eps} T \right] = 1 - o(1).
$$
The claim follows directly.

Now, by a union bound we obtain:
$$
\Pr \left[\sum_{t=1}^{T} J_t > n^{4\eps} \wedge \sum_{t=1}^{T} L_t < 4kn^{3\eps} \right] = 1 - o(1).
$$
Taking into account that $t_e>T$ holds with probability $\Pi=\Omega(1)$ by~\eqref{eq:boxproperty}, we have by a union bound that with probability at least $\Pi - o(1) = \Omega(1)$, the following event holds: $z_t \in B$ for all $t\in[0,T]$, $\sum_{t=1}^{T} J_t > n^{4\eps}$, and $L_t < 4kn^{3\eps}$. However then $\sum_{t=1}^{T} J_t - \sum_{t=1}^{T} L_t > n^{4\eps} - kn^{3\eps} > k n^{\eps}$, so there must exist at time $T$ a state $i \in LR$ such that $z^{(i)}_T > n^\eps$. This is a contradiction with $z_T \in B \subseteq B_{n^{\eps_0}}(z_0)$ by Lemma~\ref{lem:epsilons}(iii).
\end{proof}

In the rest of the proof, we consider the evolution of a protocol starting from configuration $z_0$ and having property~\eqref{eq:boxproperty}. We compare this evolution to the evolution of the same protocol, starting from a perturbed configuration $z^*_0$, such that:
\begin{enumerate}
\item[(C1)] $\on{z_0 - z^*_0} \leq n^{\eps}$.
\item[(C2)] for all low representation states $i \in LR$, we have $z^{*(i)}_0 \leq z^{(i)}_0$.
\end{enumerate}
Intuitively, the perturbed state $z^*_0$ may correspond to removing a small number of agents from $z_0$ (and replacing them by high representation states for the sake of normalization), e.g., as in the case of the disappearance of a rumor source from a system which has already performed a rumor-spreading process.

Our objective will be to show that, with probability at least $\Pi - o(1)$, after $T = n^{1+2\eps}$ the process $z^*_T$ is still not far from $z_0$, being constrained to a box in a similar way as process $z_t$. To achieve this, we define a coupling between processes $z_t$ and $z^*_t$ (knowing that process $z_t$ is constrained to a box around $z_0$ with probability $\Pi$). Informally, the analysis proceeds as follows. We run the processes together for $T = n^{1+2\eps}$ steps. In most steps, the 1-norm distance $\on{z_t - z^*_t}$ between the two processes remains unchanged, without exceeding $O(n^{3\eps})$. Otherwise, exactly one of the two processes executes a rule (and the other pauses). With a frequency of roughly $n^\eps/n$ steps (i.e., roughly $n^{3\eps}$ times in total during the process), an LP rule is executed which increases the distance between these two states. We think of this type of ``error'' as unfixable, contributing to the $O(n^{3\eps})$ distance of the processes; however, such errors are relatively uncommon. With a higher frequency of roughly $n^{3\eps}/n$ steps (i.e., roughly once every $n^{1-3\eps}$ steps), a less serious ``error'' occurs, when some HP rule $\iota$ increases the distance between the two states. The rate of such errors is too high to leave them unfixed, and we have a time window of about $n^{1-3\eps}$ steps to fix such an error (before the next such error occurs). We observe that since $\iota$ is an HP rule, which is activated with probability at least $n^{24\eps-1}$, rule $\iota$ will still be activated frequently during this time window. The coupling of transitions of states $z_t$ and $z^*_t$ is in this case performed so as to force the two processes to execute rule $\iota$ lazily, never at the same time. The number of executions of rule $\iota$ in the ensuing time window by each of the two processes follows the standard coupling pattern of a pair of lazy random walks on a line, initially located at distance $1$, until their next meeting (cf.~e.g.~\cite{aldous-book}). During this part of the coupling, we allow the distance $\on{z_t - z^*_t}$ to increase even up to $n^{6\eps}$ (as a result of executions of rule $\iota$), but the entire contribution to the distance related to rule $\iota$ is reduced to $0$ before the next HP rule ``error'' occurs, with sufficiently high probability (in this case, with probability $1-O(n^{-6\eps})$. Overall, the coupling is successful with probability $\Pi - O(n^{-\eps})$.

We remark that we use the bound on the number of states $k$ to enforce a sufficiently large polynomial separation between the frequencies of LR states and HR states, and likewise for LP rules and HP rules. We also implicitly assume that $k = n^{o(1)}$, throughout the process. The analysis also works for a choice of $k = O(\log \log n)$, with a sufficiently small hidden constant. The separation between LR/HR states and LP/HP rules is used in at least two places in the proof. First, it enforces that rules creating LR states from VHR states may appear in the definition of the protocol only with polynomially small probability (Lemma~\ref{lem:fail}), which helps to maintain over time the invariant  $z^{*(i)}_t \leq z^{(i)}_t$, for all LR states. Secondly, we use the separation of LP/HP rules in the analysis of the coupling to show that a fixable ``error'' caused by a HP rule can be sufficiently quickly repaired, before new errors occur.

In the formalization of the coupling, we make both processes $z_t$ and $z^*_t$ lazy, i.e., add to each process an additional independent coin-toss at each step, and enforce that with probability $1/2$ no rule is executed in a given step (i.e., the step is skipped by the protocol). We assume a random scheduler which picks uniformly a random pair of nodes at each step. Thus, if the scheduler picks a pair of agents in states $\{i_1, i_2\}$, and $j$ is a rule acting on this pair of states, the probability that the interaction corresponding to rule $j$ will be $q_j/2$. (The laziness of the process here is a purely technical assumption for the analysis, and corresponds to using a measure of time which is scaled by a factor of $2 \pm o(1)$ w.h.p.; this does not affect the asymptotic statement of the theorem.)

We will also find it convenient to apply an auxiliary notation for representing the evolution of a state. For process $z_t$ (resp., $z^*_t$, we define $\rho_t(j)$ (resp.~$\rho^*_t(j)$), for all $j\in [1,r]$, as the number of times rule $j$ has been executed since time $0$. Observe that the pair $(z_0, (\rho_t(j) : j \in [1,t])$ completely describes the evolution of a state (i.e., the order in which the rules were executed is irrelevant). Moreover, since each execution of a rule changes the states of at most $4$ agents, we have:
$$
\on{z_t - z^*_t} \leq 4\sum_{j=1}^r |\rho_t(j) - \rho^*_t(j)| + \on{z_0 - z^*_0} \leq + 4\sum_{j=1}^r |\rho_t(j) - \rho^*_t(j)| + n^\eps.
$$

\paragraph{Definition of the coupling.}

\begin{enumerate}
\item At each step $t$, we order the agents of configurations $z_t$ and $z^*_t$, so that $a_l(t)$ denotes the type of the $l$-th agent in $z_t$ and $a^*_l(t)$ is the type of the $l$-th agent in $z^*_t$. The orderings are such that $|\{l : a_l(t) = a^*_l(t)\}|$ is maximized; in particular, for any state $i$ such that $z^{(i)}(t) \leq z^{*(i)}(t)$ (respectively, $z^{*(i)}(t) \leq z^{(i)}(t)$) we have that if for some $l$, $a_l(t) = i$ (resp., $a_l^*(t) = i$), then $a_l^*(t) = i$ (resp., $a_l(t) = i$).

\item The scheduler then picks a pair of distinct indices $l_1, l_2 \in \{1,\ldots,n\}$ as the pair of interacting agents.
\begin{enumerate}
\item[2.1.] If $a_{l_1}{(t)} = a^*_{l_1}{(t)}$ and $a_{l_2}{(t)} = a^*_{l_2}{(t)}$, then the same rule $j=j^*$ acting on the pair of states $(a_{l_1}{(t)}, a_{l_2}{(t)})$ is chosen as the current interaction rule, with probability $q_j$.
\item[2.2.] Otherwise, a pair of (clearly distinct) rules $j$ and $j^*$ are picked independently at random for $z_t$ and $z^*_t$ from among the rules available for state pairs $(a_{l_1}{(t)}, a_{l_2}{(t)})$ and $(a_{l_1}^*{(t)}, a_{l_2}^*{(t)})$, with probabilities $q_j$ and $q_{j^*}$, respectively.
\end{enumerate}
\item The processes finally perform their coin tosses to decide which of the selected rules ($j$ for $z_t$ and $j^*$ for $z^*_t$) will be applied in the current step.
\begin{enumerate}
\item[3.1.] If $j = j^*$ and rule $j$ has been executed exactly the same number of times in the history of the two processes ($\rho_t(j) = \rho^*_t(j)$), then with probability $1/2$ both of the processes execute rule $j$, and with probability $1/2$ neither execute their rule.
\item[3.2.] If $j \neq j^*$, or if $j = j^*$ and rule $j$ has been executed a different number of times in the history of the two processes ($\rho_t(j) \neq \rho^*_t(j)$), then exactly one of the two processes performs its chosen rule and the other process waits, with the process performing the rule being chosen as $z_t$ or $z^*_t$, with probability $1/2$ each.
\end{enumerate}
\end{enumerate}

The correctness of the coupling (i.e., that the marginals $z_t$ and $z^*_t$ each correspond to a valid execution of the given protocol under a random scheduler) is immediate to verify.

\begin{lemma}
Let $z_t$ be a process satisfying property~\eqref{eq:boxproperty}, and let $z^*_0$ satisfy conditions (C1) and (C2). Then,  for $T = n^{1 + 2\eps}$, with probability $\Pi - O(n^{-\eps})$ we have $\on{z^*_T - z} = O(n^{6\eps})$, for some $z \in B$.
\end{lemma}
\begin{proof}
To prove the claim, it suffices to show that with probability $\Pi - O(n^{-\eps})$ the provided coupling succeeds, i.e., it maintains a sufficiently small difference $z_T(i) - z^*_T(i)$ for all states $i$, with $z_T \in B$.

In the analysis of the provided coupling, we will assume that the box condition $z_t \in B$ holds always throughout the process (otherwise, we assume the coupling does not succeed). To state this formally, we work with auxiliary processes $\z_t$ and $\z_t^*$, given as $\z_t = z_t$ and $\z_t^* = z_t^*$ for all $t < t_e$, where $t_e$ is the first moment of time such that $z_t\notin B_{n^{\eps_0}}(z_0)$, and set to the dummy value $\z_t = \z_t^* = z_0$ for all $t\geq t_e$. At the end of the process, we will thus have $\z_T = z_T$ and $\z_T^* = z_T^*$ with probability at least $\Pi$. In the following, we silently assume that $t < t_e - 1$ (in particular, that $z_t \in B$ and $z_{t+1}\in B$), and we will simply show that the coupling of $\z_t$ and $\z^*_t$ is successful with probability $1 - n^{-\eps}$. The condition of $t \geq t_e-1$ is trivially handled.

In addition to the box condition (which is now enforced) we try to maintain, with sufficiently high probability, throughout the first $T$ steps of the process several invariants (all at a time), corresponding to the following events holding:
\begin{itemize}
\item $F_D(t)$: for all states $i\in LR$, $\z_t^{*(i)} \leq \z_t^{(i)}$. (LR domination condition)
\item $F_{LR}(t)$: for all states $i\in LR$, $\z_t^{*(i)} \leq \z_t^{(i)} \leq n^{\eps}$. (LR state condition)
\item $F_{LP}(t)$: for all rules $j\in LP$, $\max\{p_j(\z_t), p_j(\z^*_t)\} \leq 2 n^{\eps-1}$. (LP rule condition)
\item $F_{HR}(t)$: for all states $i\in HR$, $\min\{\z_t^{(i)}, \z_t^{*(i)}\} \geq n^{24\eps}/2$. (HR state condition)
\item $F_{HP}(t)$: for all rules $j\in HP$, $\min\{p_j(\z_t), p_j(\z^*_t)\} \geq n^{24\eps-1}/2$. (HP rule condition)
\item $F_{HR'}(t)$: for all states $i\in HR'$, $\max\{\z_t^{(i)}, \z_t^{*(i)}\} \leq 2n^{1-6\eps}$. (HR' state condition)
\item a family of possible events $S_{w,d}(t)$, for some $d\in \{0,\ldots,4n^{3\eps}\}$ and $w \in \{0,\ldots, n^{6\eps}\}$, with specific events defined as follows:
    \begin{itemize}
    \item $S_{0,d}(t)$ holds if for all rules $j \in HP$ we have $\rho_t(j) = \rho^*_t(j)$, and $\sum_{j \in LP} | \rho_t(j) - \rho^*_t(j) | = d$. This implies, in particular, $\on{\z_t^* - \z_t}\leq 4d + n^{\eps} \leq 5n^{3\eps}$. (identical rate of HP execution)
    \item $S_{w,d}(t)$ for $w>0$ holds if there exists a rule $\iota \in HP$ such that for all rules $j \in HP \setminus\{\iota\}$ we have $\rho_t(j) = \rho^*_t(j)$, $|\rho_t(\iota) - \rho^*_t(\iota)| = w$, and moreover $\sum_{i \in LP} | \rho_t(i) - \rho^*_t(i) | = d$. This implies, in particular, $\on{\z_t^* - \z_t}\leq 4d + 4w + n^{\eps}\leq 5n^{6\eps}$. (single HP execution difference)
    \end{itemize}
\end{itemize}
We will call the coupling \emph{successful} if for all $t\leq T$, all events $F_\cdot(t)$ and some event $S_{w,d}(t)$ holds, and we will say it is a \emph{failure} otherwise. (We remark that condition $F_D(t)$ is implied by condition $F_{LR}(t)$, but we retain both for convenience in discussion.)

The analysis of the coupled process is now the following. First, we remark that all of the given events $F_{\cdot}(t)$ and event $S_{0,0}(t)$ hold for $t=0$. 

If the process meets condition $S_{0,d}$ at time $t$ and all conditions $F_{\cdot}(t)$, then we have the following:
\begin{itemize}
\item With probability at least $1 - O(n^{3\eps -1})$, the coupling will follow clauses 2.1 and 3.1 of its definition, and the two processes $\z$ and $\z^*$ will execute the same rule $j$ (or both pause). Hence, we continue to step $t+1$ satisfying condition $S_{0,d}$ and all of the conditions $F_{\cdot}(t+1)$, making use of the box condition for process $\z_t$. (We note that, to show $F_{LP}(t)$, when considering the special case of a rule involving a state from $LR$, we can make use of $F_{LR}(t)$ and note that the activation probability of such a rule is bounded by $2 n^{\eps-1}$ due to the $n^{\eps}$ bound on the population of a LR state).
\item With probability at most $O(n^{3\eps -1})$, the coupling will, however, select distinct rules, $j$ for $\z_t$ and $j^*$ for $\z_t^*$, and will select exactly one of them to execute, say $j' \in \{j, j^*\}$.
    \begin{itemize}
    \item If $j' \in LP$, which happens in the current step of the process with probability at most $2n^{\eps-1}$ by $F_{LP}$, then the event $S_{0,d+1}(t+1)$ will hold in the next step (provided $d+1 \leq 4n^{3\eps}$; otherwise, if $d+1 > 4n^{3\eps}$, we will say that the coupling has failed).
    \item If $j' \in HP$, which happens in the coupling with probability $O(n^{3\eps -1})$ (as bounded due to clause 2.2), then the event $S_{1,d}(t+1)$ will hold in the next step. The condition $F_D(t+1)$ requires more careful consideration. Taking into account that $F_D(t)$ holds, we need to consider two cases: either $j'=j$ and the rule applied to $\z_t$ changed at least one of the two interacting states $\{i_1(j), i_2(j)\}$, say $i_1(j) \in LR$, so that $\z^{i_1(j)}(t) = \z^{*i_1(j)}(t)$ and $\z^{i_1(j)}(t+1) \leq \z^{*i_1(j)}(t+1)-1$, or $j'=j^*$ and the rule applied to $\z^*_t$ created a pair of states $\{o_1(j^*), o_2(j^*)\}$, say $o_1(j^*) \in LR$. In the first case, by the description of the ordering given in clause 1 of the definition of the coupling, the problem occurs only if one of the agents picked by the scheduler belongs to an $LR$ state, and the other agent is at a position in which the states of $\z$ and $\z^*$ differ in the ordering of the agents; hence, the probability that the coupling fails at this step is at most $O(\frac{k n^{\eps} \cdot n^{3\eps}}{n^2}) \leq O(n^{5\eps - 2})$. In the second case, we likewise analyze the ordering of the agents considered by the scheduler, and note that the interacting agent, which belongs to the part of the ordering in which $\z_t$ and $\z^*_t$ differ, must be in a HR state, since the agents in a LR state in $\z^*$ are matched by their counterparts in $\z$ (as noted in clause 1 of the discussion of the coupling). If the other interacting agent is in a state from $LR \cup HR'$, then such an event occurs with probability $O(\frac{n^{3\eps} \cdot k n^{1-6\eps}}{n^2}) \leq O(n^{-2.9\eps - 1})$, and we say that with this probability the coupling has failed. Finally, if the other interacting agent is in a state from $VHR$, then by Lemma~\ref{lem:fail}, we have that the probability of picking a rule under which the coupling fails is at most $O(n^{-14\eps})$, conditioned on the event $j \neq j^*$ holding, hence overall the probability of failure is $O(n^{-14\eps} n^{3\eps-1}) = O(n^{-11\eps -1})$. Overall, we obtain that $F_D(t+1)$ holds with probability $O(n^{-2.9\eps - 1})$. Given $F_D(t+1)$, $S_{1,d(t+1)}$, and the box condition, the remaining conditions $F_{\cdot}(t+1)$ follow directly.
    \end{itemize}
\end{itemize}
Overall, we obtain that following a time $t$ satisfying $S_{0,d}(t)$ and all conditions $F_{\cdot}(t)$, we reach the following successor state (see Fig.~\ref{fig:lb}):
\begin{figure}[t]
\includegraphics[width=0.8\textwidth]{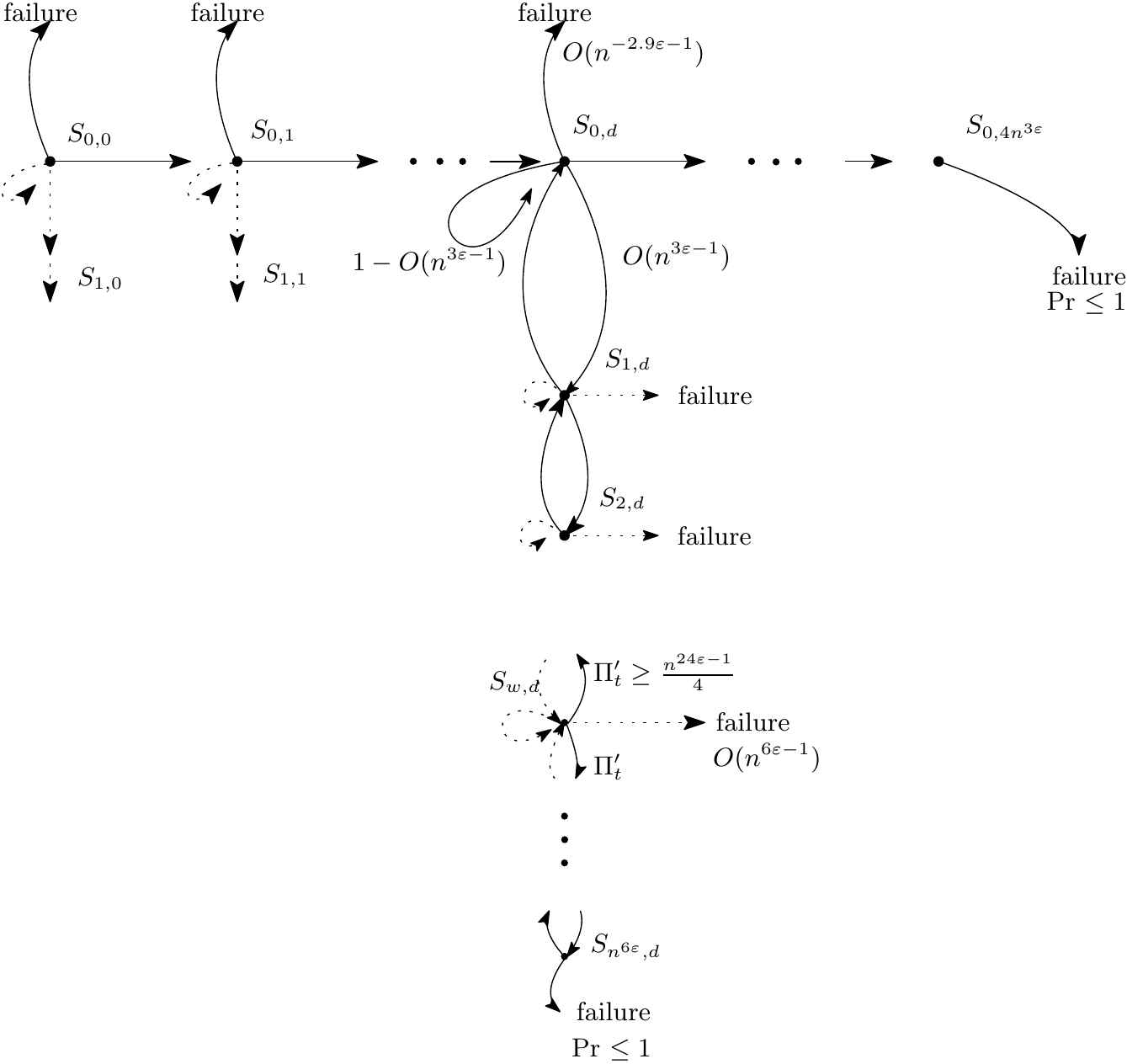}
\caption{Illustration of transitions between states $S_{w,d}$ for the coupling in the lower bound proof.}
\label{fig:lb}
\end{figure}
$$
\begin{cases}
S_{0,d}(t+1) \wedge F_{\cdot}(t+1), & \text{with probability $1 - O(n^{3\eps -1})$},\\
S_{0,d+1}(t+1)  \wedge F_{\cdot}(t+1), & \text{with probability $\leq 2n^{\eps -1}$, if $d+1 \leq 4n^{3\eps}$},\\
S_{1,d}(t+1) \wedge F_{\cdot}(t+1), & \text{with probability $O(n^{3\eps -1})$},\\
\text{failure:} & \begin{cases}
\text{with probability $O(n^{-2.9\eps - 1})$,} & \text{if $d+1 \leq 4n^{3\eps}$},\\
\text{with some probability $\leq 1$,} & \text {otherwise.}
\end{cases}
\end{cases}
$$
At this point, before proceeding further, we can provide some intuition on the meaning of the respective events $S$. The coupling process can be seen as a walk along the path $(S_{0,d} : d\leq 4n^{3\eps}$), starting from state $S_{0,0}$, and at each step, either staying in the current state $S_{0,d}$, moving on to the next state $S_{0,d+1}$, branching to a side branch $S_{1,d}$ (which we will analyze later), or failing. The process also fails if it reaches the endpoint of its path ($d = 4n^{3\eps}$). Since the process is run for $T = n^{1 + 2\eps}$ steps, the probability that failure will occur before the end of the path is reached is $O(n^{-0.9 \eps})$, and the probability of reaching the end of the path and failing is exponentially small in $n^\eps$ by a Chernoff bound (in expectation, the process will progress halfway along the path). Hence, we have that the process succeeds with probability $1- O(n^{-0.9 \eps})$, or otherwise may fail in a side branch $S_{\cdot,d}$.

A side branch is entered with probability $O(n^{\eps -1})$. To show that the coupling succeeds with the required probability, it suffices to show that we return from any state $S_{1,d}$ to state $S_{0,d}$ with probability at least $1 - O(n^{-6\eps})$; then, all (i.e., w.h.p. at most $O(n^{1+2\eps} n^{3\eps -1}) = O(n^{5\eps})$) excursions into side branches during the process will succeed with probability $1 - O(n^{-\eps})$.

Consider now an excursion into a side branch $S_{w,d}$ ($w\geq 1$) associated with a rule $\iota \in HP$, which has been executed a different number of times in $\z_t$ and $\z^*_t$.
Now, if the process meets condition $S_{w,d}$ at time $t$ and all conditions $F_{\cdot}(t)$, then we have the following:
\begin{itemize}
\item With probability at least $1 - O(n^{6\eps -1})$, the coupling will follow clause 2.1 of its definition, selecting a single rule $j$.

\begin{itemize}
\item If $j \neq \iota$, then clause 3.1 will follow, and the two processes $\z$ and $\z^*$ will execute the same rule $j$ (or both pause). Hence, at time $t+1$, all of the conditions $F_{\cdot}(t+1)$ and condition $S_{w,d}(t+1)$ is satisfied.
\item Else, the event $j = \iota$ occurs. The probability of such an event is denoted $\pi_{t} \in [p_\iota(\z_t) - O(n^{6\eps -1}), p_\iota(\z_t)]$ (due to the conditioning performed in the first clause of the coupling); since $p_\iota(\z_t) \geq n^{24\eps - 1}$ by the box condition for HP rules, it follows that $2\pi_{t} \geq n^{24\eps - 1} - O(n^{6\eps -1}) \geq n^{24\eps -1}/2$. Now, following clause 3.2 of the coupling, depending on which of the two processes $\z_t$, $\z^*_t$ is chosen to execute the rule, with probability $\pi_t/2 =: \pi'_t$ the system moves to $S_{w-1,d}(t+1)$, and with probability $\pi'_t$ the system moves to $S_{w+1,d}(t+1)$ (unless $w+1 > n^{6\eps}$, in which case the coupling has failed). As before, given there was no failure, all conditions $F_{\cdot}(t+1)$ are readily verified to be satisfied in the new time step.
\end{itemize}
\item With probability at most $O(n^{6\eps -1})$, for simplicity of analysis we assume the coupling has failed.
\end{itemize}
This time, for a time $t$ satisfying $S_{w,d}(t)$ for $w\geq 1$ and all conditions $F_{\cdot}(t)$, we obtain the following distribution of successor states:
$$
\begin{cases}
S_{w-1,d}(t+1) \wedge F_{\cdot}(t+1), & \text{with probability exactly $\pi'_t \geq n^{24\eps - 1}/4$},\\
S_{w+1,d}(t+1) \wedge F_{\cdot}(t+1), & \text{with probability exactly $\pi'_t$, if $w+1 \leq n^{6\eps}$},\\
\text{failure:} &\begin{cases}
\text{with probability $O(n^{6\eps -1})$}, & \text{if $w+1 \leq n^{6\eps}$,}\\
\text{with some probability $\leq 1$}, & \text{otherwise,}\\
\end{cases}\\
S_{w,d}(t+1)  \wedge F_{\cdot}(t+1), & \text {otherwise.}
\end{cases}
$$
The picture here corresponds to a lazy random walk along the side line $S_{w,d}$ for $w \in [0,n^{6\eps}]$, with an additional failure probability at each step. The walk starts at $w=1$ and ends with a return to the primary line $S_{0,d}$ if the endpoint $w=0$ is reached, or ends with failure if the other endpoint $w\geq n^{6\eps} =: w_{\max}$ is reached. At each step, the walk is lazy (with probability of transition depending on the current step), but unbiased with respect to transitions to the left or to the right. Assuming that failure does not occur sooner, with probability $1 - O(\frac{1}{w_{\max}}) = 1 - O(n^{-6\eps})$ the walk will reach point $w=0$ in $O(w_{\max}^2) = O(n^{12\eps})$ moves (transitions along the line), without reaching the other endpoint of the line sooner. Since a move is made in each step $t$ with probability $\pi'_t \geq n^{24\eps - 1}/4$, by a straightforward Chernoff bound, the number of steps spent on this line is given w.h.p. as at most $O(n^{12\eps} / n^{24\eps - 1}) = O(n^{1 - 12\eps})$. As the probability of failure in each of these steps is $O(n^{6\eps -1})$, the probability that the process fails during these steps is $O(n^{-6\eps})$. Overall, by a union bound, we obtain that the process successfully returns to $S_{0,d}$ with probability $1 - O(n^{-6\eps})$ (and within $O(n^{1 - 12\eps})$ steps). In view of the previous observations, we have that with probability $1 - O(n^{-\eps})$, all conditions $F_{\cdot}$ and some condition $S_{w,d}$ hold at time $T$. Thus, with probability $\Pi - n^{-\eps}$, process $z^*_T$ is sufficiently close to $B$, i.e., there exists a point $z \in B$ such that $\on{z^*_T - z} = O(n^{6\eps})$.
\end{proof}

\section{Proof of Proposition~\ref{pro:constbox}}\label{sec:constbox}

\begin{proof}
Fix protocol $P$ with set of states $K$, in which the minimum positive probability of executing some rule is $p$. Let $K' \subseteq K$, $K' \ni X$ be any minimal subset of the set of states such that no evolution of protocol $P$ starting in a configuration containing only states from set $K'$ will ever contain an agent in a state outside $K'$. Denote $\kappa = |K'|-1$. Consider an initialization of protocol $P$ at time $t_0 = 0$, at a configuration $z(0)$ with $x \in [c, 1/2]$ and with all other states from $K'$ represented by the same number of agents, i.e., for each $Q \in K'$, we have $q(0) = (1-x)/\kappa$.

Let $t \geq kn$ be an arbitrarily chosen time step. Let $t_1 = t-(\kappa-1) n$. Fix $Q_1 \in K' \setminus\{X\}$ as any state such that $q_1(t_1) \geq 1/2\kappa$ (we can, e.g., fix $q_1$ as the state from $K' \setminus\{X\}$ having the most agents at time $t_1$). Observe that from the minimality of $K'$ it follows that there must exist a sequence of states $(Q_1,\ldots, Q_{\kappa})$, with $\{Q_1,\ldots, Q_{\kappa} = K' \setminus \{X\}$, such that in the definition of protocol $P$, for all $i \in \{1,\ldots, \kappa-1\}$, some rule of protocol $P$ creates at least one agent (i.e., either 1 or 2 agents) in state $Q_{i+1}$ from an interaction of either the pair of agents in states $(Q_j, Q_i)$ or the pair of agents in states $(Q_i, Q_j)$, for some $j \in \{1,\ldots,i\}$. (Indeed, if for some $i$ there was no possibility of choosing $Q_{i+1}$ in any way, then $K'' = \{X, Q_1,\ldots Q_i\} \subseteq K'$ would be closed under agent creation, contradicting the minimality of the choice of $K'$.) Now, we consider intervals of time steps $[t_s, t_{s+1}]$, with $t_s = t_1 + (s-1) n$ for $s>1$. We make the following claims:
\begin{itemize}
\item[(1)] Fix $i \in \{1, 2, \ldots \kappa\}$. If $q_i (t_s) = \Omega (1)$, then $q_i (t) \geq 0.11 q_i (t_{s})$ for all $t \in [t_s, t_{s+1}$, with probability $1 - e^{-n^\Omega(1)}$. Indeed, in a sequence of $n$ steps, the expected number of agents which do not participate in any interaction in the time interval $[t_s, t_{s+1}]$ following the asynchronous scheduler is $(\frac{n-2}{n})^n n > 0.13 n$, and thus the number of non-interacting agents is at least $0.12n$, with probability $1 - e^{-n^\Omega(1)}$ following standard concentration bounds for the number of isolated vertices in a random graph on $n$ nodes with $n$ edges. Since the choice of agents by the scheduler is independent of their state, and the probability for a uniformly random agent to be in state $Q_i$ at time $t_s$ is $q_i(t_{s})$, a simple concentration bound shows that $0.11 q_i (t_{s}) n$ having state $Q_i$ at time $t_s$ do not participate in any interaction in the interval $[t_s, t_{s+1}]$.
\item[(2)] Fix $i \in \{1, 2, \ldots \kappa-1\}$.  Denote $m_i = \min_{j\leq i}q_j (t_i)$. If $m_i = \Omega (1)$, then $q_{i+1}(t_{i+1}) \geq 0.01 p m_i^2$, with probability $1 - e^{-n^\Omega(1)}$. Indeed, consider the value $j\leq i$ such that the interaction $(Q_j, Q_i)$ or $(Q_i, Q_j)$ creates an agent in state $Q_{i+1}$. At any time $t$ within the interval $[t_i, t_{i+1}]$, we have by Claim (1) that $q_j(t) \geq 0.11 m_i$ and $q_i(t) \geq 0.11 m_i$, with probability $1 - e^{-n^\Omega(1)}$. It follows that an interaction creating a new agent in state $Q_{i+1}$ is triggered with probability at least $p(0.11 m_i)^2$ at each step. The number of agents in state $Q_{i+1}$ at time step $t_{i+1}$ may thus be dominated from below by the number of successes in a sequence of $n$ Bernoulli trials with success probability $p(0.11 m_i)^2$, and the claim follows.
\end{itemize}
By applying Claim (2) iteratively for $i = \{1, 2, \ldots \kappa-1\}$, where we note that $m_1 \geq 1/2\kappa$, we have $q_{i+1}(t_{i+1}) \geq (0.01p / 2\kappa)^{2^i}$, with probability $1 - e^{-n^\Omega(1)}$ (through successive union bounds). Then, applying Claim (1) up to time $t = t_{\kappa-1}$, we have $q_{i+1}(t_{\kappa-1}) \geq 0.11^{\kappa-1-i-1} (0.01p / 2\kappa)^{2^i} \geq (0.01p / 2\kappa)^{2^\kappa}$, with probability $1 - e^{-n^\Omega(1)}$. Applying once again a union bound, we have shown that for all $q \in K'$, we have $q(t) \geq (0.01p / 2\kappa)^{2^\kappa} \equiv C_0$, with probability $1 - e^{-n^\Omega(1)}$. The claim of the lemma follows for a suitable choice of $\delta_0 > 0$.

\end{proof}

\paragraph*{Acknowledgment.}

We sincerely thank Dan Alistarh and Przemek Uzna\'nski for inspiring discussions, and Lucas Boczkowski for many detailed comments which helped to improve this manuscript.
\enlargethispage{5mm}

\bibliographystyle{abbrv}
\bibliography{biblio}

\begin{thebibliography}{10}

\bibitem{AD15}
M.~A. Abdullah and M.~Draief.
\newblock Majority consensus on random graphs of a given degree sequence.
\newblock {\em CoRR}, abs/1209.5025, 2012.

\bibitem{aldous-book}
D.~Aldous and J.~A. Fill.
\newblock Reversible markov chains and random walks on graphs, 2002.
\newblock Unfinished monograph, recompiled 2014, available at
  \url{http://www.stat.berkeley.edu/~aldous/RWG/book.html}.

\bibitem{pplb2}
D.~Alistarh, J.~Aspnes, D.~Eisenstat, R.~Gelashvili, and R.~L. Rivest.
\newblock Time-space trade-offs in population protocols.
\newblock In {\em Proc.\ Twenty-Eighth Annual {ACM-SIAM} Symposium on Discrete
  Algorithms, {SODA} 2017, Barcelona, Spain}, pages 2560--2579, 2017.

\bibitem{AAG18}
D.~Alistarh, J.~Aspnes, and R.~Gelashvili.
\newblock Space-optimal majority in population protocols.
\newblock {\em (To appear, SODA 2018.) CoRR}, abs/1704.04947, 2017.

\bibitem{dna}
D.~Alistarh, B.~Dudek, A.~Kosowski, D.~Soloveichik, and P.~Uznanski.
\newblock Robust detection in leak-prone population protocols.
\newblock In {\em {DNA}}, volume 10467 of {\em Lecture Notes in Computer
  Science}, pages 155--171. Springer, 2017.

\bibitem{AADFP06}
D.~Angluin, J.~Aspnes, Z.~Diamadi, M.~J. Fischer, and R.~Peralta.
\newblock Computation in networks of passively mobile finite-state sensors.
\newblock {\em Distributed Computing}, 18(4):235--253, 2006.

\bibitem{angluin2008}
D.~Angluin, J.~Aspnes, and D.~Eisenstat.
\newblock A simple population protocol for fast robust approximate majority.
\newblock {\em Distributed Computing}, 21(2):87--102, 2008.

\bibitem{AAER07}
D.~Angluin, J.~Aspnes, D.~Eisenstat, and E.~Ruppert.
\newblock The computational power of population protocols.
\newblock {\em Distributed Computing}, 20(4):279--304, 2007.

\bibitem{AR09}
J.~Aspnes and E.~Ruppert.
\newblock An introduction to population protocols.
\newblock {\em Bulletin of the {EATCS}}, 93:98--117, 2007.

\bibitem{BCNPST14}
L.~Becchetti, A.~E.~F. Clementi, E.~Natale, F.~Pasquale, R.~Silvestri, and
  L.~Trevisan.
\newblock Simple dynamics for majority consensus.
\newblock {\em CoRR}, abs/1310.2858, 2013.

\bibitem{bocz}
L.~Boczkowski, A.~Korman, and E.~Natale.
\newblock Minimizing message size in stochastic communication patterns: Fast
  self-stabilizing protocols with 3 bits.
\newblock In {\em Proc.\ Twenty-Eighth Annual {ACM-SIAM} Symposium on Discrete
  Algorithms, {SODA} 2017, Barcelona, Spain}, pages 2540--2559, 2017.

\bibitem{BournezFK12}
O.~Bournez, P.~Fraigniaud, and X.~Koegler.
\newblock Computing with large populations using interactions.
\newblock In {\em {MFCS}}, volume 7464 of {\em Lecture Notes in Computer
  Science}, pages 234--246. Springer, 2012.

\bibitem{concen.pdf}
F.~R.~K. Chung and L.~Lu.
\newblock Survey: Concentration inequalities and martingale inequalities: {A}
  survey.
\newblock {\em Internet Mathematics}, 3(1):79--127, 2006.

\bibitem{voting}
C.~Cooper, R.~Els{\"{a}}sser, and T.~Radzik.
\newblock The power of two choices in distributed voting.
\newblock In {\em Proc.\ 41st International Colloquium on Automata, Languages,
  and Programming, {ICALP} 2014, Copenhagen, Denmark, Part {II}}, pages
  435--446, 2014.

\bibitem{ICALP15}
J.~Czyzowicz, L.~Gasieniec, A.~Kosowski, E.~Kranakis, P.~G. Spirakis, and
  P.~Uznanski.
\newblock On convergence and threshold properties of discrete
  {L}otka-{V}olterra population protocols.
\newblock In {\em Proc.\ 42nd International Colloquium on Automata, Languages,
  and Programming, {ICALP} 2015, Kyoto, Japan, Part {I}}, pages 393--405, 2015.

\bibitem{DF12}
A.~Dobrinevski and E.~Frey.
\newblock Extinction in neutrally stable stochastic {L}otka-{V}olterra models.
\newblock {\em Phys. Rev. E}, 85:051903, May 2012.

\bibitem{dx}
B.~Doerr, C.~Doerr, and T.~K{\"{o}}tzing.
\newblock The right mutation strength for multi-valued decision variables.
\newblock In T.~Friedrich, F.~Neumann, and A.~M. Sutton, editors, {\em
  Proceedings of the 2016 on Genetic and Evolutionary Computation Conference,
  Denver, CO, USA, July 20 - 24, 2016}, pages 1115--1122. {ACM}, 2016.

\bibitem{rsgraphs}
B.~Doerr, M.~Fouz, and T.~Friedrich.
\newblock Social networks spread rumors in sublogarithmic time.
\newblock {\em Electronic Notes in Discrete Mathematics}, 38:303--308, 2011.

\bibitem{muldrift}
B.~Doerr, D.~Johannsen, and C.~Winzen.
\newblock Multiplicative drift analysis.
\newblock {\em CoRR}, abs/1101.0776, 2011.

\bibitem{rsexact}
B.~Doerr and M.~K{\"{u}}nnemann.
\newblock Tight analysis of randomized rumor spreading in complete graphs.
\newblock In {\em 2014 Proc.\ Eleventh Workshop on Analytic Algorithmics and
  Combinatorics, {ANALCO} 2014, Portland, Oregon, USA, January 6, 2014}, pages
  82--91, 2014.

\bibitem{pplb1}
D.~Doty.
\newblock Timing in chemical reaction networks.
\newblock In {\em Proc.\ Twenty-Fifth Annual {ACM-SIAM} Symposium on Discrete
  Algorithms, {SODA} 2014, Portland, Oregon, USA, January 5-7, 2014}, pages
  772--784, 2014.

\bibitem{rsnetworks}
U.~Feige, D.~Peleg, P.~Raghavan, and E.~Upfal.
\newblock Randomized broadcast in networks.
\newblock In {\em Algorithms, International Symposium {SIGAL} '90, Tokyo,
  Japan, 1990}, pages 128--137, 1990.

\bibitem{fou}
N.~Fountoulakis and K.~Panagiotou.
\newblock Rumor spreading on random regular graphs and expanders.
\newblock {\em Random Struct. Algorithms}, 43(2):201--220, 2013.

\bibitem{rs2}
A.~M. Frieze and G.~R. Grimmett.
\newblock The shortest-path problem for graphs with random arc-lengths.
\newblock {\em Discrete Applied Mathematics}, 10(1):57--77, 1985.

\bibitem{grzes}
L.~Gasieniec and G.~Stachowiak.
\newblock Fast space optimal leader election in population protocols.
\newblock {\em (To appear, SODA 2018.) CoRR}, abs/1704.07649, 2017.

\bibitem{gnw}
G.~Giakkoupis, Y.~Nazari, and P.~Woelfel.
\newblock How asynchrony affects rumor spreading time.
\newblock In {\em Proc.\ 2016 {ACM} Symposium on Principles of Distributed
  Computing, {PODC} 2016, Chicago, IL, USA}, pages 185--194, 2016.

\bibitem{bion}
A.~Goldbeter.
\newblock Computational approaches to cellular rhythms.
\newblock {\em Nature}, 420(6912):238--245, Nov. 2002.

\bibitem{HS98}
J.~Hofbauer and K.~Sigmund.
\newblock {\em {E}volutionary {G}ames and {P}opulation {D}ynamics}.
\newblock Cambridge University Press, 1998.

\bibitem{rspushpull}
R.~M. Karp, C.~Schindelhauer, S.~Shenker, and B.~V{\"{o}}cking.
\newblock Randomized rumor spreading.
\newblock In {\em 41st Annual Symposium on Foundations of Computer Science,
  {FOCS} 2000, 12-14 November 2000, Redondo Beach, California, {USA}}, pages
  565--574, 2000.

\bibitem{KRFB02}
B.~Kerr, M.~A. Riley, M.~W. Feldman, and B.~J.~M. Bohannan.
\newblock Local dispersal promotes biodiversity in a real-life game of
  rock-paper-scissors.
\newblock {\em Nature}, 418(6894):171--174, Jul 2002.

\bibitem{KR04}
B.~C. Kirkup and M.~A. Riley.
\newblock Antibiotic-mediated antagonism leads to a bacterial game of
  rock-paper-scissors in vivo.
\newblock {\em Nature}, 428(6981):412--414, Mar 2004.

\bibitem{L}
L.~Lamport.
\newblock Time, clocks, and the ordering of events in a distributed system.
\newblock {\em Commun. {ACM}}, 21(7):558--565, 1978.

\bibitem{dy}
P.~K. Lehre and C.~Witt.
\newblock General drift analysis with tail bounds.
\newblock {\em CoRR}, abs/1307.2559, 2013.

\bibitem{LLSW}
C.~Lenzen, T.~Locher, P.~Sommer, and R.~Wattenhofer.
\newblock Clock synchronization: Open problems in theory and practice.
\newblock In {\em Proc.\ 36th Conference on Current Trends in Theory and
  Practice of Computer Science, {SOFSEM} 2010, Spindleruv Ml{\'{y}}n, Czech
  Republic}, pages 61--70, 2010.

\bibitem{LLW}
C.~Lenzen, T.~Locher, and R.~Wattenhofer.
\newblock Clock synchronization with bounded global and local skew.
\newblock In {\em 49th Annual {IEEE} Symposium on Foundations of Computer
  Science, {FOCS} 2008, October 25-28, 2008, Philadelphia, PA, {USA}}, pages
  509--518, 2008.

\bibitem{L1910}
A.~J. Lotka.
\newblock Contribution to the theory of periodic reactions.
\newblock {\em The Journal of Physical Chemistry}, 14(3):271--274, 1909.

\bibitem{LL84}
J.~Lundelius and N.~A. Lynch.
\newblock A new fault-tolerant algorithm for clock synchronization.
\newblock In {\em Proc.\ Third Annual {ACM} Symposium on Principles of
  Distributed Computing, Vancouver, B. C., Canada}, pages 75--88, 1984.

\bibitem{Mertzios2016}
G.~B. Mertzios, S.~E. Nikoletseas, C.~L. Raptopoulos, and P.~G. Spirakis.
\newblock Determining majority in networks with local interactions and very
  small local memory.
\newblock {\em Distributed Computing}, 30(1):1--16, 2017.

\bibitem{MCS11}
O.~Michail, I.~Chatzigiannakis, and P.~G. Spirakis.
\newblock {\em New Models for Population Protocols}.
\newblock Synthesis Lectures on Distributed Computing Theory. Morgan {\&}
  Claypool Publishers, 2011.

\bibitem{PS}
K.~Panagiotou and L.~Speidel.
\newblock Asynchronous rumor spreading on random graphs.
\newblock {\em CoRR}, abs/1608.01766, 2016.

\bibitem{PK09}
M.~Parker and A.~Kamenev.
\newblock Extinction in the {L}otka-{V}olterra model.
\newblock {\em Phys. Rev. E}, 80:021129, 2009.

\bibitem{rs}
B.~Pittel.
\newblock On spreading a rumor.
\newblock {\em SIAM J. Appl. Math.}, 47(1):213--223, Mar. 1987.

\bibitem{soc}
V.~B. Priezzhev, D.~Dhar, A.~Dhar, and S.~Krishnamurthy.
\newblock Eulerian walkers as a model of self-organized criticality.
\newblock {\em Phys. Rev. Lett.}, 77:5079--5082, Dec 1996.

\bibitem{RMF06}
T.~Reichenbach, M.~Mobilia, and E.~Frey.
\newblock Coexistence versus extinction in the stochastic cyclic
  {L}otka-{V}olterra model.
\newblock {\em Phys. Rev. E}, 74:051907, Nov 2006.

\bibitem{rssauerwald}
T.~Sauerwald.
\newblock On mixing and edge expansion properties in randomized broadcasting.
\newblock {\em Algorithmica\enlargethispage{7mm}}, 56(1):51--88, 2010.

\bibitem{SMJSRP14}
A.~Szolnoki, M.~Mobilia, L.~Jiang, B.~Szczesny, A.~M. Rucklidge, and M.~Perc.
\newblock Cyclic dominance in evolutionary games: {A} review.
\newblock {\em CoRR}, abs/1408.6828, 2014.

\end{thebibliography}

\newpage

\appendix

\begin{sidewaysfigure}
\noindent\textbf{\large APPENDIX: Simulation Examples for Oscillator Dynamics}\\
    \begin{subfigure}[b]{0.3\textwidth}
        \includegraphics[width=\textwidth]{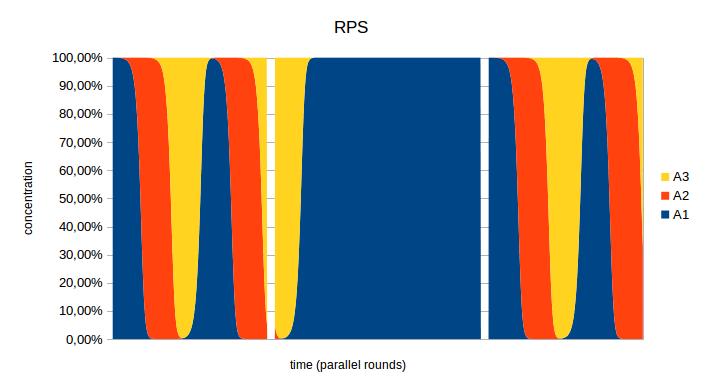}
    \end{subfigure}
    ~
    \begin{subfigure}[b]{0.3\textwidth}
        \includegraphics[width=\textwidth]{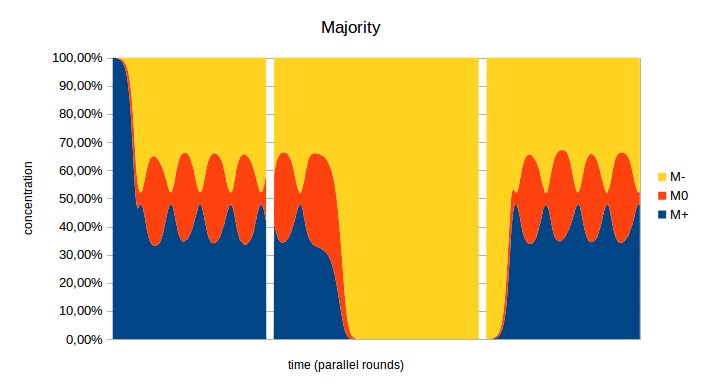}
    \end{subfigure}
    ~
    \begin{subfigure}[b]{0.3\textwidth}
        \includegraphics[width=\textwidth]{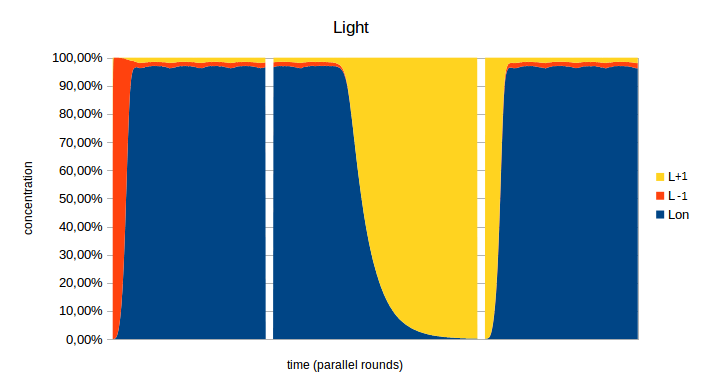}
    \end{subfigure}
    ~
    \begin{subfigure}[b]{0.3\textwidth}
        \includegraphics[width=\textwidth]{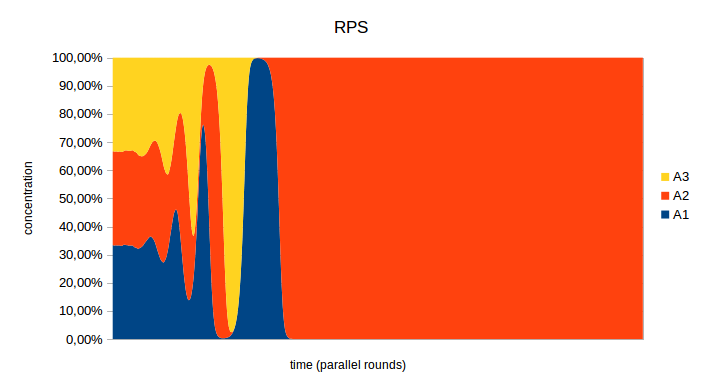}
    \end{subfigure}
    ~
    \begin{subfigure}[b]{0.3\textwidth}
        \includegraphics[width=\textwidth]{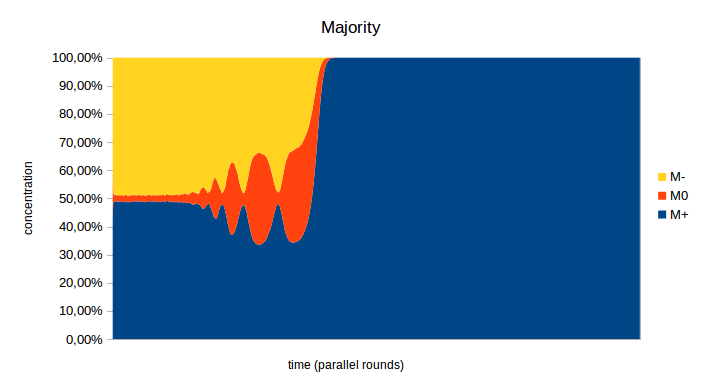}
    \end{subfigure}
    ~
    \begin{subfigure}[b]{0.3\textwidth}
        \includegraphics[width=\textwidth]{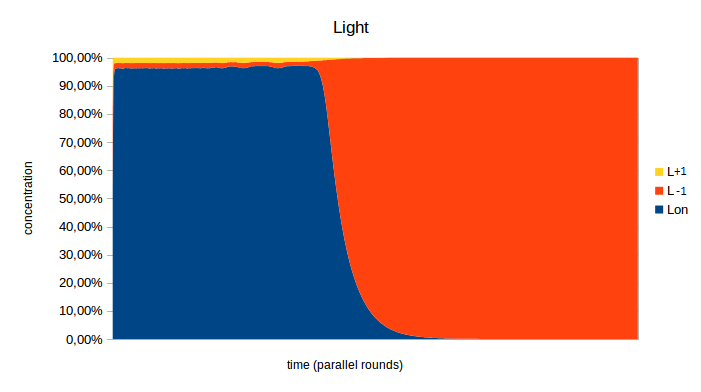}
    \end{subfigure}
    ~
    \begin{subfigure}[b]{0.3\textwidth}
        \includegraphics[width=\textwidth]{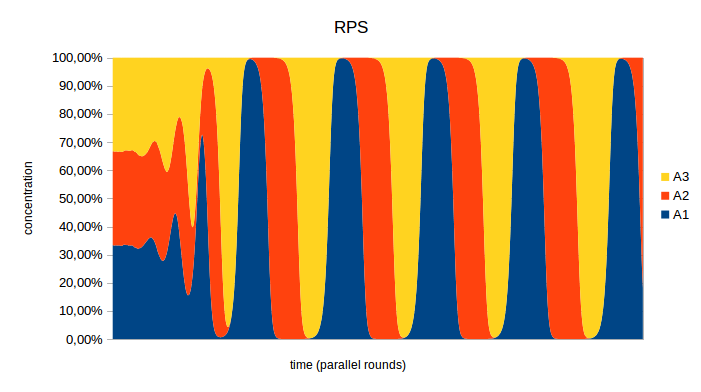}
    \end{subfigure}
    ~
    \begin{subfigure}[b]{0.3\textwidth}
        \includegraphics[width=\textwidth]{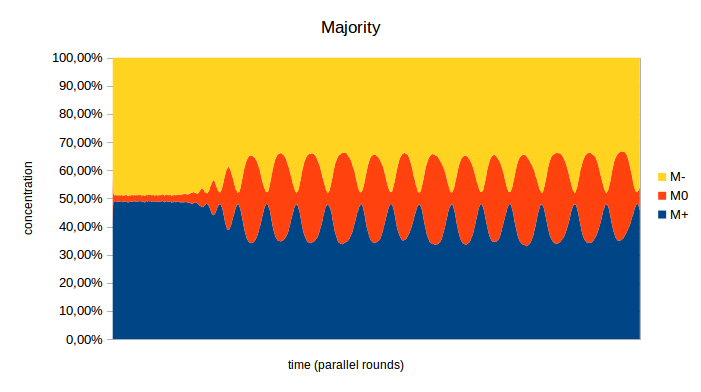}
    \end{subfigure}
    ~
    \begin{subfigure}[b]{0.3\textwidth}
        \includegraphics[width=\textwidth]{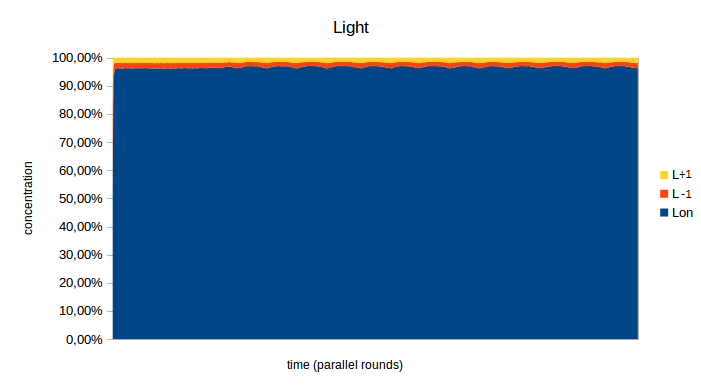}
    \end{subfigure}
    ~
    \begin{subfigure}[b]{0.1\textheight}
        \includegraphics[width=\textwidth]{charts/empty.pdf}
    \end{subfigure}
    \vspace{-2.5cm}
    \caption{Illustration of concentration of various features as a function of time steps for a simulation of the protocol $P_o$ for $n = 10^6$, $p=6\cdot10^{-2},q=10^{-2},r=10^{-1}$, $s=1$ in various scenarios.
    Left column shows concentration of species $A_1$, $A_2$, $A_3$, middle one the majority layer and right one light.
    \textbf{First row}: initialization from a corner configuration ($\many X =1$, left panel), further dynamics of the protocol after rumor source is removed ($\many X =0$, middle panel), further dynamics of the protocol after rumor source is reinserted ($\many X =1$, right panel).
    \textbf{Middle row}: initialization from a configuration with $A_1 = A_2 = A_3 = 1/3$ and $\many X = 0$. \textbf{Bottom row}: as above, with $\many X = 1$. The range of the horizontal time scale corresponds to 2500 parallel rounds of the protocol.
}
\label{fig:osc_3x3}
\end{sidewaysfigure}

\end{document}